%% file: main.tex
\documentclass[colorlinks=true,citecolor=blue,linkcolor=magenta,a4paper,UKenglish, autoref]{lipics-v2019}

\title{Dynamic Byzantine Reliable Broadcast \\\normalsize(Technical Report) \thanks{Author names appear in alphabetical order. This is an extended version of a conference article, appearing in the proceedings of the 24th Int. Conference on Principles of Distributed Systems (OPODIS 2020).
This work has been supported in part by the Interchain Foundation, Cross-Chain Validation project.}}

\titlerunning{DBRB} 

 \author{Rachid Guerraoui}{EPFL}{}{}{}
\author{Jovan Komatovic}{EPFL}{}{}{}
\author{Petr Kuznetsov}{LTCI, T\'el\'ecom Paris\\Institut Polytechnique Paris}{}{}{}
\author{Yvonne-Anne Pignolet}{DFINITY}{}{}{}
\author{Dragos-Adrian Seredinschi}{Informal Systems}{}{}{}
\author{Andrei Tonkikh}{National Research University\\Higher School of Economics}{}{}{}

\authorrunning{R. Guerraoui, J. Komatovic, P. Kuznetsov, Y.-A. Pignolet, D.-A. Seredinschi, A. Tonkikh} 

\Copyright{Rachid Guerraoui, Jovan Komatovic, Petr Kuznetsov, Yvonne-Anne Pignolet, Dragos-Adrian Seredinschi, Andrei Tonkikh} 



\nolinenumbers 

\hideLIPIcs  

\EventEditors{----}
\EventNoEds{2}
\EventLongTitle{OPODIS 2020}
\EventShortTitle{OPODIS 2020}
\EventAcronym{OPODIS}
\EventYear{2016}
\EventDate{December 24--27, 2016}
\EventLocation{Little Whinging, United Kingdom}
\EventLogo{}
\SeriesVolume{42}
\ArticleNo{23}

\input{sty/defs.tex}
\input{sty/code.tex}
\usepackage{amssymb,amsmath,amsthm}
\newtheorem{assumption}{Assumption}
\newtheorem{property}{Property}

\usepackage{cleveref}
\crefname{section}{\S}{\S\S}
\Crefname{section}{\S}{\S\S}
\crefformat{section}{\S#2#1#3}

\Crefname{assumption}{Assumption}{Assumptions}
\Crefname{property}{Property}{Properties}

\usepackage{paralist}
\usepackage[subtle]{savetrees}

\begin{document}

\maketitle

\input{sections/abstract}

\input{sections/intro}
\input{sections/model.tex}
\input{sections/specification.tex}
\input{sections/overview.tex}

\input{sections/algorithm.tex}

\input{sections/correctness-informal.tex}
\input{sections/conclusion.tex}



\bibliography{main}

\appendix

\newpage
\section{Appendix}
This appendix includes the formal proof of the correctness of \name algorithm and its optimality.
Namely, \Cref{sec:preliminary_definitions} introduces a few definitions that we use throughout the entire section.
In \Cref{sec:preliminary_lemmata}, we show that all views ``created'' by our algorithm do form a sequence of views.
Moreover, we show that processes update system membership from their perspective if the system reconfigures.
We discuss the View Discovery protocol in \Cref{subsection:view_discovery}.
In \Cref{sec:dynamicity_proof}, we show that correct processes do eventually join or leave the system.
\Cref{sec:broadcast_proof} builds on previous subsections of the appendix and shows that our \name algorithm satisfies properties from \cref{definition:spec,definition:non-triv-spec}.
Lastly, in \Cref{sec:optimality-impossibilities}, we show that guarantees we provide are indeed optimal in the model we consider.

\vspace*{-2mm}
\input{sections/appendix/correctness.tex}
\input{sections/appendix/optimality.tex}

\end{document}

%% file: sty/defs.tex
\usepackage{xspace}

\newcommand{\dbrbBroadcast}[1]{\textsc{dbrb-broadcast}($#1$)}
\newcommand{\dbrbBroadcastAlone}{\textsc{dbrb-broadcast}{}}

\newcommand{\dbrbDeliver}[1]{\textsc{dbrb-deliver}($#1$)}

\newcommand{\dbrbJoin}{\textsc{dbrb-join}{}}
\newcommand{\dbrbLeave}{\textsc{dbrb-leave}{}}
\newcommand{\sdbrbLeave}{\textsc{sdbrb-leave}{}}

\newcommand{\brb}{\textsc{brb}{}}

\newcommand{\operation}[2]{$\langle$\textsc{#1}, \ensuremath{#2\rangle}\xspace}
\newcommand{\msg}[1]{\text{\textsc{#1}}}

\newcommand{\reliablbcast}{\textsc{brb}\xspace}
\newcommand{\name}{\textsc{dbrb}\xspace}

\usepackage{xcolor}
\usepackage{framed}
\definecolor{lightgray}{gray}{0.90}
\renewenvironment{leftbar}[1][\hsize]
{%
\MakeFramed{\hsize#1\advance\hsize-\width\FrameRestore}%
}
{\endMakeFramed}

\newcommand\ignore[1]{}


%% file: sty/code.tex
\usepackage{amsthm}
\usepackage[utf8]{inputenc}
\usepackage{algorithm}
\usepackage{algpseudocode}
\usepackage{caption}

\usepackage{algorithmicx}
\usepackage{algorithm} 
\usepackage{xifthen}
\usepackage{graphicx}

\newcommand{\ap}[1]{{\left\langle#1\right\rangle}}

\renewcommand{\algorithmiccomment}[1]{\bgroup\hfill//~#1\egroup}

\newcommand{\event}[3]{
    \ifthenelse
    {\equal{#3}{}}
    {\ap{#1.\textrm{#2}}}
    {\ap{#1.\textrm{#2} \mid #3}}
}

\algnewcommand\Instance[2]{\State #1, \textbf{instance} #2}
\algnewcommand\InstanceSystem[3]{\State #1, \textbf{instance} #2, \textbf{system} #3}
\algnewcommand\Trigger[3]{\State \textbf{trigger} $\event{#1}{#2}{#3}$}

\newcommand\algvspace{\vspace{0.6\baselineskip}}

\algsetblockdefx[Implements]{Implements}{EndImplements}{}{}{\textbf{Implements:}}{}
\algtext*{EndImplements}

\algsetblockdefx[Uses]{Uses}{EndUses}{}{}{\textbf{Uses:}}{}
\algtext*{EndUses}

\algsetblockdefx[Parameters]{Parameters}{EndParameters}{}{}{\textbf{Parameters:}}{}
\algtext*{EndParameters}

\algsetblockdefx[Upon]{Upon}{EndUpon}{}{}[3]{\textbf{upon event} $\event{#1}{#2}{#3}$ \textbf{do}}{}
\algtext*{EndUpon}

\algsetblockdefx[UponRecipt]{UponReceipt}{EndUponReceipt}{}{}[2]{\textbf{upon receipt of} #1 from $#2$}{}
\algtext*{EndUponReceipt}

\algsetblockdefx[UponReciptFromST]{UponReceiptFromST}{EndUponReceiptFromST}{}{}[3]{\textbf{upon receipt of} #1 from $#2$ \textbf{such that} $#3$}{}
\algtext*{EndUponReceiptFromST}

\algsetblockdefx[UponReciptFromSTLong]{UponReceiptFromSTLong}{EndUponReceiptFromSTLong}{}{}[3]{\textbf{upon receipt of} #1 from $#2$ \Statex \textbf{such that} $#3$}{}
\algtext*{EndUponReceiptFromSTLong}

\algsetblockdefx[UponReciptST]{UponReceiptST}{EndUponReceiptST}{}{}[4]{\textbf{upon receipt of} $\langle #1, #2 \rangle$ from $#3$ \textbf{such that} $#4$}{}
\algtext*{EndUponReceiptST}

\algsetblockdefx[UponReciptS]{UponReceiptS}{EndUponReceiptS}{}{}[4]{\textbf{upon receipt of} $\langle #1, #2 \rangle_{\sigma_{#3}}$ from $#4$}{}
\algtext*{EndUponReceiptS}

\algsetblockdefx[UponReciptCond]{UponReceiptCond}{EndUponReceiptCond}{}{}[1]{\textbf{upon receipt of} #1 from $v.q$ processes in $v$}{}
\algtext*{EndUponReceiptCond}

\algsetblockdefx[UponReciptCondST]{UponReceiptCondST}{EndUponReceiptCondST}{}{}[2]{\textbf{upon receipt of} #1 from $v.q$ processes in $v$ \textbf{such that} $#2$}{}
\algtext*{EndUponReceiptCondST}

\algsetblockdefx[UponReciptCondS]{UponReceiptCondS}{EndUponReceiptCondS}{}{}[2]{\textbf{upon receipt of} $\langle #1, #2 \rangle_{\sigma_{SK_p}}$ from $v.q$ processes in $v$}{}
\algtext*{EndUponReceiptCondS}

\algsetblockdefx[UponExists]{UponExists}{EndUponExists}{}{}[2]{\textbf{upon exists} $#1$ \textbf{such that} $#2$ \textbf{do}}{}
\algtext*{EndUponExists}

\algsetblockdefx[UponExistsLong]{UponExistsLong}{EndUponExistsLong}{}{}[2]{\textbf{upon exists} $#1$ \Statex \textbf{\;such that} $#2$ \textbf{do}}{}
\algtext*{EndUponExistsLong}

\algsetblockdefx[UponCondition]{UponCondition}{EndUponCondition}{}{}[1]{\textbf{upon} \emph{#1} \textbf{do}}{}
\algtext*{EndUponCondition}

\algsetblockdefx[Procedure]{Procedure}{EndProcedure}{}{}[2]{\textbf{procedure} \emph{#1}($#2$)}{}
\algtext*{EndProcedure}

\algsetblockdefx[Function]{Function}{EndFunction}{}{}[2]{\textbf{function} \emph{#1}($#2$)}{}
\algtext*{EndFunction}

\algsetblockdefx[Variables]{Variables}{EndVariables}{}{}[0]{\textbf{variables:}}{}
\algtext*{EndVariables}

\algnewcommand{\IfThen}[2]{
  \State \algorithmicif\ #1\ \algorithmicthen\ #2}

\algsetblockdefx[IfExists]{IfExists}{EndIfExists}{}{}[2]{\textbf{if exists} $#1$ \textbf{such that} $#2$}{}
\algsetblockdefx[While]{While}{EndWhile}{}{}[1]{\textbf{while} $#1$ \textbf{do}}{\textbf{end while}}
\algsetblockdefx[NTimes]{NTimes}{EndNTimes}{}{}[1]{\textbf{for} $#1$ \textbf{times do}}{\textbf{end for}}
\algsetblockdefx[For]{For}{EndFor}{}{}[3]{\textbf{for} $#1 \in #2..#3$ \textbf{do}}{\textbf{end for}}
\algsetblockdefx[ForAll]{ForAll}{EndForAll}{}{}[2]{\textbf{for all} $#1 \in #2$ \textbf{do}}{\textbf{end for}}

\algrenewcommand{\algorithmiccomment}[1]{\hskip3em$\triangleright$ #1}
\algnewcommand{\LineComment}[1]{\State \(\triangleright\) #1}

%% file: sections/abstract.tex

\begin{abstract}

\vspace*{-1mm}
Reliable broadcast is a communication primitive guaranteeing, intuitively,  that all processes in a distributed system deliver the same set of messages.
The reason why this primitive is appealing is twofold:
\emph{(i)} we can implement it deterministically in a completely asynchronous environment, 
unlike stronger primitives like consensus and total-order broadcast, and yet 
\emph{(ii)} reliable broadcast is powerful enough to implement important applications like payment systems.

The problem we tackle in this paper is that of \emph{dynamic} reliable broadcast, 
i.e., enabling processes to join or leave the system.
This property is desirable for long-lived applications 
(aiming to be highly available), yet has been precluded in previous asynchronous reliable broadcast protocols.
We study this property in a general adversarial (i.e., Byzantine) environment.

We introduce the first specification of a dynamic Byzantine reliable broadcast (\name) primitive that is amenable to an asynchronous implementation.
We then present an algorithm implementing this specification in an asynchronous network. 
Our \name algorithm ensures that if any correct process in the system broadcasts a message, 
then every correct process delivers that message unless it leaves the system. 
Moreover, if a correct process delivers a message, then every correct process that has not expressed its will to leave the system delivers that message.
We assume that more than $2/3$ of processes in the system are correct at all times, which is tight in our context.

We also show that if only one process in the system can fail---and it can fail only by crashing---then it is impossible to implement a stronger primitive, ensuring that if any correct process in the system 
broadcasts or delivers
a message, then every correct process in the system delivers that message---including those that leave.
\end{abstract}

%% file: sections/intro.tex

\section{Introduction}

\vspace*{-2mm}
Networks typically offer a reliable form of communication channels: TCP. 
As an abstraction, these channels ensure that if neither the sender nor the destination of a message fail, then the message is eventually delivered. 
Essentially, this abstraction hides the unreliability of the underlying IP layer, so the user of a TCP channel is unaware of the lost messages.


Yet, for many applications, TCP is not reliable enough.
Indeed, think of the situation where a message needs to be sent to all processes of a distributed system.
If the sender does not fail, TCP will do the job; but otherwise, the message might reach only a strict subset of processes.
This can be problematic for certain applications, such as a financial notification service when processes subscribe to information published by other processes.
For fairness reasons, one might want to ensure that if the sender fails, either \emph{all or no process} delivers that message.
Moreover, if the correct processes choose to deliver, they must deliver the same message, even when the sender is Byzantine.     
We talk, therefore, about \emph{reliable broadcast}.
Such a primitive does not ensure that messages are delivered in the same total order, but 
simply in the ``all-or-nothing'' manner. 

Reliable broadcast is handy for many applications, including, for example, 
cryptocurrencies. Indeed, in contrast to what was 
implicitly considered since Nakamoto's original paper \cite{nakamotobitcoin}, there 
is no need to ensure consensus on the ordering of messages, i.e., 
to totally order messages, if the goal is to perform secure payments. 
A reliable broadcast scheme suffices~\cite{guerraoui2019consensus}. 

Reliable broadcast is also attractive because, unlike stronger primitives 
such as total order broadcast and consensus, it can
be implemented deterministically in a completely asynchronous environment \cite{bracha1987asynchronous}.
The basic idea uses a quorum of correct processes, and makes that 
quorum responsible for ensuring that a message is transmitted to all 
processes if the original sender of the message fails. 
If a message does not reach the quorum, it will not be delivered by any process.
It is important to notice at this point a terminology difference 
between the act of ``receiving'' and the act of ``delivering'' a message. 
A process indeed might ``receive'' a message $m$, but not necessarily ``deliver'' $m$ to its application until it is confident that the ``all-or-nothing'' property of the reliable broadcast is ensured.

A closer look at prior asynchronous implementations of reliable broadcast reveals, however, a gap between 
theory and practice.
The implementations described so far all assume a \emph{static} system.
Essentially, the set of processes in the system remains the same, except that some of them might fail.
The ability of a process to join or leave the system, which is very desirable 
in a long-lived application supposed to be highly available, is precluded in all asynchronous reliable broadcast protocols published so far. 



In this paper, we introduce the first specification of a \emph{dynamic} Byzantine reliable broadcast (\name) primitive that is amenable to an asynchronous implementation.
The specification allows any process outside the broadcast system to join;
any process that is inside the system can ask to leave.
Processes inside the system can broadcast and deliver messages, whereas processes outside the system  cannot.
Our specification is intended for an asynchronous system for it does not require the processes to agree on the system membership. 
Therefore, our specification does not build on top of a group membership scheme, as does the classical \emph{view synchrony} abstraction \cite{chockler2001group}.

Our asynchronous \name implementation ensures that if any correct process in the system broadcasts a message, then eventually every correct process, unless it asks to leave the system, delivers that message.
Moreover, if any correct process delivers a message, then every correct process, if it has not asked to leave prior to the delivery, delivers that message.
The main technical difficulty addressed by our algorithm is to combine asynchrony and dynamic membership, which makes it impossible for processes to agree on the exact membership.
Two key insights enable us to face this challenge.
First, starting from a known membership set at system bootstrap time, we construct a sequence of changes to this set; at any time, there is a majority of processes that record these changes. 
Based on this sequence, processes can determine the validity of messages.
Second, before transitioning to a new membership, correct processes exchange their current state with respect to ``in-flight'' broadcast messages and membership changes.
This prevents equivocation and conflicts.

Our algorithm assumes that, at any point in time, more than 2/3 of the processes inside the broadcast system are correct, which is tight.
Moreover, we show that the ``all-or-nothing'' property we ensure is, in some sense, maximal.
More precisely, we prove that in an asynchronous system, even if only one process in the system can fail, and it can merely fail by crashing, then it is impossible to implement a stronger property, ensuring that if any correct process in the system broadcasts (resp., delivers) a message, then every correct process in the system delivers that message, including those that are willing to leave.

The paper is organized as follows.
In~\Cref{sec:model}, we describe our system model and introduce the specification of \name. 
In~\Cref{sec:overview}, we overview the structure of our algorithm. 
In~\Cref{sec:algorithm}, we describe our implementation, and in~\Cref{sec:correctness-informal}, we argue its correctness. 
We conclude in~\Cref{sec:conclusions} with a discussion of related and future work.
Detailed proofs are delegated to the optional appendix.

%% file: sections/model.tex
\vspace*{-4mm}
\section{Model and Specification}
\label{sec:model}

\vspace*{-3mm}
We describe here our system model (\Cref{sec:universe_processes}) and specify our \name primitive (\Cref{sec:spec,sec:assumptions,sec:properties}).

\vspace*{-5mm}
\subsection{A Universe of Asynchronous Processes}
\label{sec:universe_processes}

\vspace*{-2mm}
We consider a universe $\mathcal{U}$
of processes, subject to \emph{Byzantine} failures: a faulty process may arbitrarily deviate from the algorithm it is assigned.  
Processes that are not subject to failures are \emph{correct}.
We assume an asymmetric cryptographic system.   
Correct processes communicate with signed messages: prior to sending a message $m$ to a process $q$, a process $p$ signs $m$, labeled $\langle m\rangle_{\sigma_{p}}$. 
Upon receiving the message, $q$ can verify its authenticity and use it to prove its origin to others (non-repudiation).
To simplify presentation, we omit the signature-related notation and, thus, whenever we write $m$, the identity of sender $p$ and the signature are implicit and correct processes only consider messages, whether received directly or relayed by other processes, if they are equipped with valid signatures.
We also use the terms ``send'' and ``disseminate'' to differentiate the points in our algorithm when a process sends a message, resp., to a single or to many destinations.

The system $\mathcal{U}$ is \emph{asynchronous}: we make no assumptions on communication delays or relative speeds of the processes.We assume that communication is \emph{reliable}, i.e., every message sent by a correct process to a correct process is eventually received.
To describe the events that occur to an external observer and prove the correctness of the protocol, we assume a global notion of time, outside the control of the processes (not used in the protocol implementation).
We consider a subset of $\mathcal{U}$ called the {\em broadcast system}.
We discuss below how processes join or leave the broadcast system.


%% file: sections/specification.tex

\vspace*{-5mm}
\subsection{DBRB Interface}
\label{sec:spec}

\vspace*{-2mm}
Our \name primitive exposes an interface with three operations and one callback:
\begin{compactenum}
    \item \dbrbJoin{}: used by a process outside the system to \emph{join}. 
    \item \dbrbLeave{}: used by a process inside the system to \emph{leave}.
    \item \dbrbBroadcast{m}: used by a process inside the system to broadcast a message $m$.
    \item \dbrbDeliver{m}: this callback is triggered to handle the delivery of a message $m$.
\end{compactenum}
If a process is in the system initially, or if it has returned from the invocation of a {\dbrbJoin{}} call, we say that it has \emph{joined} the system.
Furthermore, it is considered \emph{participating} (or, simply, a \emph{participant}) if it has not yet invoked {\dbrbLeave{}}.
When the invocation of {\dbrbLeave{}} returns, we say that the process \emph{leaves} the system.
Note that in the interval between the invocation and the response of a {\dbrbLeave{}} call, the process is no longer participating, but has not yet left the system.   

The following rules (illustrated in \Cref{fig:process-states}) govern the behavior of correct processes: 
(i) a \dbrbJoin{} operation can only be invoked if the process is not   participating; moreover, we assume that \dbrbJoin{} is invoked at most once;
(ii) only a participating process can invoke a \dbrbBroadcast{m} operation;
(iii) a \dbrbDeliver{m} callback can be triggered only if a process has previously joined but has not yet left the system;
(iv) a \dbrbLeave{} operation can only be invoked by a participating process.

\begin{figure}[htp]
    \centering
    \includegraphics[width=.5\columnwidth]{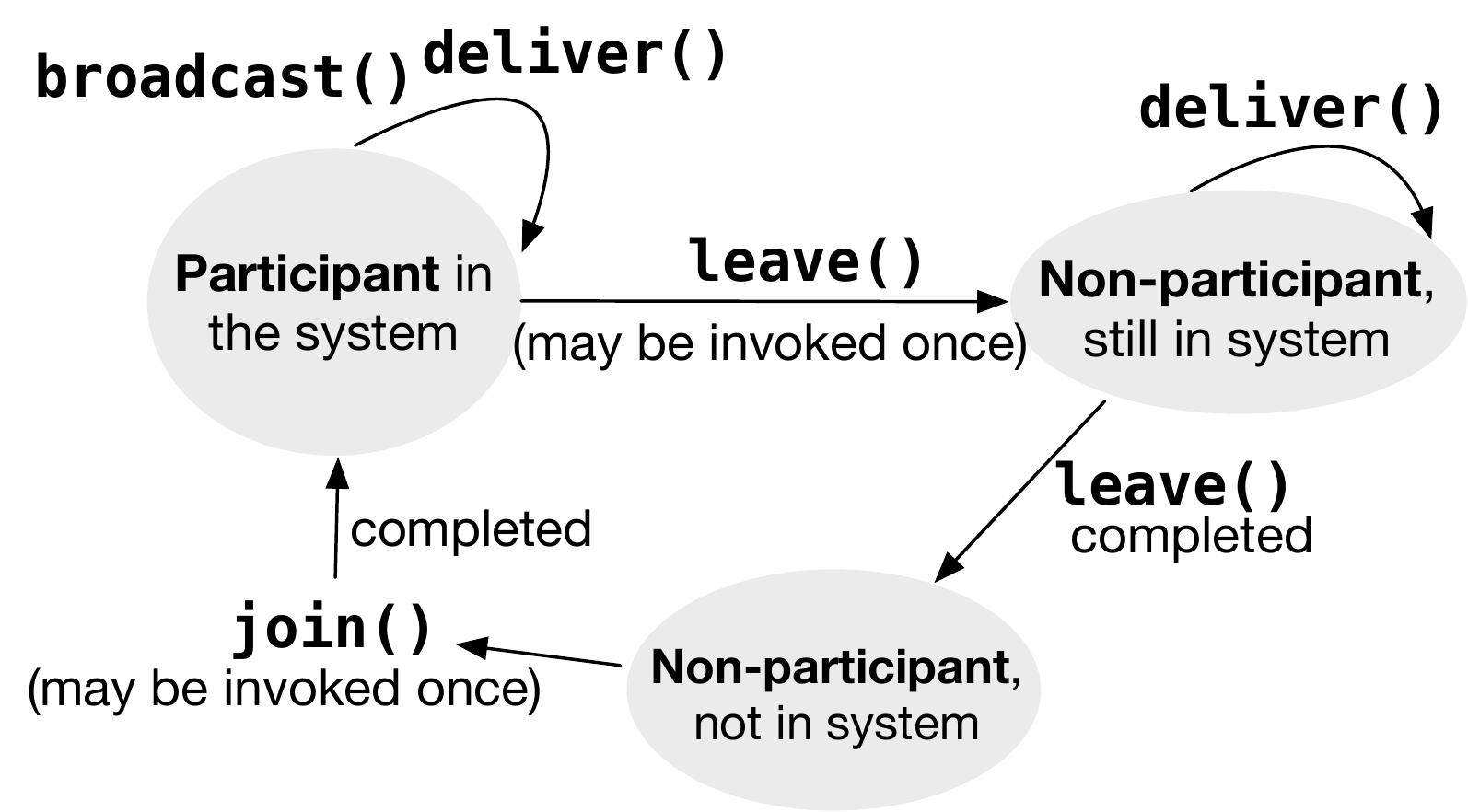}
    \vspace*{-3mm}
    \caption{State transition diagram for correct processes. 
    }
    \label{fig:process-states}
    \vspace{-.6cm}
\end{figure}

\vspace*{-2mm}
\subsection{Standard Assumptions}
\label{sec:assumptions}

\vspace*{-2mm}
We make two  standard assumptions in asynchronous reconfiguration protocols~\cite{aguilera2011dynamic,alchieri2016efficient,baldoni2009implementing,spiegelman2017liveness}, which we restate below for the sake of completeness.

\vspace*{-2mm}
\begin{assumption}[Finite number of reconfiguration requests]\label{assumption:requests}
In every execution, the number of processes that want to join or leave the system is finite.
\end{assumption}
\vspace*{-5mm}

\begin{assumption}\label{assumption:init}
Initially, at time $0$, the set of participants 
is nonempty 
and known to every process in $\mathcal{U}$.
\end{assumption}
\vspace*{-2mm}
\Cref{assumption:requests} captures the assumption that no new reconfiguration requests will be made for ``sufficiently long'', thus ensuring that started operations do complete.
\Cref{assumption:init} is necessary to bootstrap the system and guarantees that all processes have the same starting conditions.
Additionally, we make standard cryptographic assumptions regarding the power of the adversary, namely that it cannot subvert cryptographic primitives, e.g., forge a signature.

We also assume that a weak broadcast primitive is available.
The primitive guarantees that if a correct process broadcasts a message~$m$, then every correct process eventually delivers $m$.
%
%
In practice, such primitive can be implemented by some sort of a gossip protocol~\cite{kermarrec2007gossiping}.
This primitive is ``global'' in a sense that it does not require a correct process to know all the members of $\mathcal{U}$.

\vspace*{-3mm}
\subsection{Properties of DBRB}
\label{sec:properties}
\vspace*{-2mm}
For simplicity of presentation, we assume a specific 
instance of \name in which a predefined sender process $s$ disseminates a  single message via \dbrbBroadcastAlone{} operation.
The specification can easily be extended to the general case in which every participant can broadcast multiple messages, assuming that every message is uniquely identified.  

\vspace*{-2mm}
\begin{definition}[\name basic guarantees]
\label{definition:spec}
\phantom{Empty line}
\begin{compactitem}
\item \emph{Validity}. If a correct participant $s$ broadcasts a message $m$ at time $t$, then every correct process, if it is a participant at time $t' \geq t$ and never leaves the system, eventually delivers $m$.

\item \emph{Totality}. If a correct process $p$ delivers a message $m$ at time $t$, then every correct process, if it is a participant at time $t' \geq t$, eventually delivers $m$.

\item \emph{No duplication}. 
A message is delivered by a correct process at most once.

\item \emph{Integrity}. If some correct process delivers a message $m$ with sender $s$ and $s$ is correct, then $s$ previously broadcast $m$.\footnote{Recall that the identity of sender process $s$ for a given message $m$ is implicit in the message (\Cref{sec:universe_processes}).}

\item \emph{Consistency}. If some correct process delivers a message $m$ and another correct process delivers a message $m'$, then $m=m'$.




\item \emph{Liveness}. Every operation invoked by a correct process eventually completes.



\end{compactitem}
\end{definition}



\vspace*{-2mm}
To filter out implementations that involve \emph{all} processes in the broadcast protocol, we add  the following non-triviality property.


\vspace*{-2mm}
\begin{definition}[Non-triviality]
\label{definition:non-triv-spec}
No correct process sends any message before invoking {\normalfont \dbrbJoin{}} or after returning from {\normalfont \dbrbLeave{}} operation.
\end{definition}

%% file: sections/overview.tex
\vspace*{-6mm}
\section{Overview}
\label{sec:overview}

\vspace*{-2mm}
We now present the building blocks underlying our \name algorithm (\Cref{sec:building-blocks}) and describe typical scenarios:
 (1)~a correct process joining or leaving the system (\Cref{sec:join-overview}), and (2)~a broadcast (\Cref{sec:broadcast-overview}).

\vspace*{-2mm}
\subsection{Building Blocks}
\label{sec:building-blocks}

\vspace*{-2mm}
\textbf{Change.} We define a set of \emph{system updates} $\mathit{change} = \{+, -\} \times \mathcal{U}$, where the tuple $\langle +, p \rangle$ (resp., $\langle -, p \rangle$) indicates that process $p$ asked to join (resp., leave) the system.
This abstraction captures the evolution of system membership throughout time.
It is inevitable that, due to asynchrony, processes might not be able to agree on an unique system membership.
In other words, two processes may concurrently consider different sets of system participants to be valid.
To capture this divergence, we introduce the \emph{view} abstraction, which defines the system membership through the lenses of some specific process at a specific point in time.

\smallskip
\noindent\textbf{View.} A view $v$ comprises a set of updates $v.\mathit{changes}$.
The set determines the view \emph{membership} as $v.\mathit{members} = \{p \in \mathcal{U}: \langle +, p \rangle \in v.\mathit{changes} \land \langle -, p \rangle \notin v.\mathit{changes} \}$.
For simplicity, sometimes we use $p \in v$ instead of $p \in v.\mathit{members}$; $|v|$ is a shorthand for $|v.\mathit{members}|$.

Intuitively, each correct process $p$ in \name uses a view as an append-only set, to record all the \emph{changes} that the broadcast system underwent up to a point in time, as far as $p$ observed.
Some views are ``instantiated'' in \name protocol and those views are marked as \emph{valid} (a formal definition is deferred to \Cref{sec:correctness-informal}). 
Our protocol ensures  that  
all valid views are \emph{comparable}.
Formally, $v_1 \subset v_2$ means that $v_1.\mathit{changes} \subset v_2.\mathit{changes}$, we say that  $v_2$ is \emph{more recent} than $v_1$.
Two different views are comparable if one is more recent than the other, otherwise they \emph{conflict}.
We assume that the \emph{initial view}, i.e., the set of participants at time $0$, is publicly known (\Cref{assumption:init}).

A valid view $v$ must be equipped with a \emph{quorum system}: a collection of subsets of $v.\mathit{members}$. 
We choose the quorums to be all subsets of size $v.q = |v| - \lfloor \frac{|v| - 1}{3} \rfloor$.
\vspace*{-2mm}
\begin{assumption}[Quorum systems]\label{assumption:quorums} 
In every valid view $v$, the number of Byzantine processes is less than or equal to $\lfloor \frac{|v| - 1}{3} \rfloor$ and at least one quorum in $v$ contains only correct processes.  
\end{assumption}
\vspace*{-2mm}
Thus, every two quorums of a valid view have a correct process in common and at least one quorum contains only correct processes.\footnote{ 
%
%
Note that this bound applies both to processes that are active participants, as well as processes leaving the system. This requirement can be relaxed in practice by enforcing  a correct process that leaves the system to destroy its private key. Even if the process is later compromised, it will not be able to send any protocol messages. Note that we assume that messages sent while the process was correct cannot be withdrawn or modified.}

\smallskip
\noindent\textbf{Sequence of views.} 
We now build upon the comparability of views to obtain the abstraction of a \emph{sequence of views}, or just \emph{sequence}.
A sequence $seq$ is a set of mutually comparable views.
Note that a set with just one view is, trivially, a sequence of views, so is the empty set.

\smallskip
\noindent\textbf{Reliable multicast.}
In addition to the use of signed messages (\Cref{sec:universe_processes}), we build our algorithm on top of an elementary (static) reliable Byzantine broadcast protocol.
We instantiate this protocol from a standard solution in the literature for a static set of processes, as described e.g., in~\cite{cachin2011introduction}.
The terms ``\emph{R-multicast}'' and ``\emph{R-delivery}'' refer to the request to broadcast a message and deliver a message via this protocol to (or from) a static set of processes.
For completeness, we provide the pseudocode of the static reliable Byzantine broadcast primitive in \Cref{sec:preliminary_lemmata}.

\vspace*{-3mm}
 \subsection{DBRB-JOIN and DBRB-LEAVE Operations}
 \label{sec:join-overview}
 \begin{figure}[t]
     \centering
     \includegraphics[width=\columnwidth]{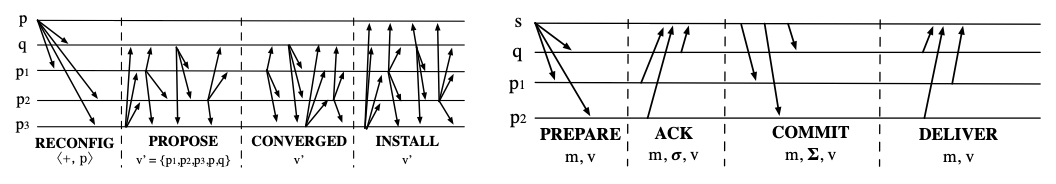}
     \vspace*{-7mm}
     \caption{Protocol overview for \dbrbJoin{} or \dbrbLeave{} (left), and \dbrbBroadcastAlone{} (right).}
     \label{fig:join-scenario}
     \vspace{-.4cm}
 \end{figure}

\vspace*{-2mm}
 Upon invoking the \dbrbJoin{} operation, a process $p$ first learns the current membership -- i.e., the most recent view $v$ -- of the broadcast system through a \emph{View Discovery} protocol (\Cref{sec:join-leave-algorithms}).
 %
 The joining operation then consists of four steps 
 (the left part of \Cref{fig:join-scenario}).
 First, $p$ disseminates a \operation{reconfig}{\langle +, p \rangle} message to members of $v$.
 In the second step, when any correct process $q$ from $v$ receives the \msg{reconfig} message, $q$ proposes to change the system membership to a view $v'$, where $v'$ is an extension of $v$ including the change $\langle +, p \rangle$.
 To do so, $q$ disseminates to members of $v$ a \msg{propose} message, containing the details of $v'$.
 Third, any other correct member in $v$ waits until $v.q$ matching \msg{propose} messages (a quorum of $v$ confirms the new view).
 Once a process collects the confirmation, it disseminates a \operation{converged}{v'} message to members of $v$.
 This concludes step three.
 In the fourth step, each correct process $q$ in $v$ waits to gather matching \msg{converged} messages from a quorum (i.e., $v.q$) of processes.
 We say that processes that are members of view $v$ are trying to \emph{converge} on a new membership.
 Then, $q$ triggers an \emph{R-multicast} of the \operation{install}{v'} message to members of $v \cup v'$; recall that the process $p$ belongs to $v'$.
 Upon \emph{R-delivery} of an \msg{install} message for $v'$, any process $q$ updates its current view to $v'$.
 The \dbrbJoin{} operation finishes at process $p$ once this process receives the \msg{install} message for $v' \ni p$.
 From this instant on, $p$ is a participant in the system.

 The steps executed after a correct process $p$ invokes the \dbrbLeave{} operation are almost identical, except for the fact that $p$ still executes its ``duties'' in \name until \dbrbLeave{} returns.\footnote{
 There is a detail we deliberately omitted from this high-level description and we defer to \Cref{sec:join-leave-algorithms}: multiple processes may try to join the system concurrently, and thereby multiple \msg{propose} messages may circulate at the same time.
 These messages comprise different views, e.g., one could be for a view $v'$ and another for $v''$.
 These conflicts are unavoidable in asynchronous networks.
 For this reason, \msg{propose} messages (and other protocol messages) operate at the granularity of sequences, not individual views.
 If conflicts occur, sequences support union and ordering, allowing reconciliation of $v'$ with $v''$ on a sequence that comprises their union.}

\vspace*{-3mm}
\subsection{DBRB-BROADCAST Operation}
\label{sec:broadcast-overview}

\vspace*{-2mm}
A correct process $s$ that invokes \dbrbBroadcast{m} first disseminates a \msg{prepare} message to every member of the $s$' current view $v$.
When a correct process $q$ receives this message, $q$ sends an \msg{ack} message to $s$, representing a signed statement asserting that $q$ indeed received $m$ from $s$.
Once $s$ collects a quorum of matching \msg{ack} messages for $m$, $s$ constructs a \emph{message certificate} $\Sigma$ out of the collected signatures $\rho$, and disseminates this certificate to every member of $v$ as part of a \msg{commit} message.
When any correct process $q$ receives a \msg{commit} message with a valid certificate for $m$ for the first time, $q$ relays this message to all members of view $v$.
Moreover, $q$ sends a \msg{deliver} message to the sender of the \msg{commit} message.
Once any process $q$ collects a quorum of matching \msg{deliver} messages, $q$ triggers \dbrbDeliver{m}.
The right part of \Cref{fig:join-scenario} presents the overview of this operation.
In \Cref{fig:join-scenario}, we depict process $s$ collecting enough \msg{deliver} messages to deliver $m$, assuming that all processes in the system use the same view. 
The details of how views are changed during an execution of a broadcast operation are given in~\Cref{sec:broadcast-algorithm}.



%% file: sections/algorithm.tex
\vspace*{-6mm}
\section{DBRB Algorithm}
\label{sec:algorithm}

\vspace*{-2mm}
In this section, we describe our \name algorithm, starting with dynamic membership (\Cref{sec:join-leave-algorithms}), and continuing with broadcast (\Cref{sec:broadcast-algorithm}).
We also present an illustrative execution of \name (\Cref{sec:execution}).

\Cref{algorithm:variables} introduces the variables that each process $p$ maintains, as well as two  helper functions to compute the least recent and most recent view of a given sequence, respectively.


\vspace*{-3mm}
\subsection{Dynamic Membership}
\label{sec:join-leave-algorithms}

\vspace*{-2mm}
\Cref{algorithm:dynamic,algorithm:install} contain the pseudocode of the \dbrbJoin{} and \dbrbLeave{} operations.
Let us first discuss the join operation. 

After a correct process $p$ 
invokes the \dbrbJoin{} operation, $p$ obtains the most recent view of the system, and it does so through the View Discovery protocol.
We describe the View Discovery protocol at the end of this section; for the moment it suffices to say that $p$ obtains the most recent view $v$ and updates its local variable $\mathit{cv}$ to reflect this view.
Next, process $p$ disseminates a \operation{reconfig}{\langle +, p \rangle, \mathit{cv}} message to every member of $\mathit{cv}$ (\Cref{line:send_join}) notifying members of $cv$ of its intention to join. 
The view discovery and the dissemination are repeated until the \emph{joinComplete} event triggers or a quorum of confirmation messages has been collected for some view $v$ to which \msg{reconfig} message was broadcast (\Cref{line:stop_asking_to_join}).

Every correct member $r$ of the view $cv$ proposes a new system membership that includes process $p$, once $r$ receives the aforementioned \msg{reconfig} message from process $p$.
The new proposal is incorporated within a \emph{sequence} of views $\mathit{SEQ^v}$, $v = cv$, (containing, initially, just one view) and disseminated to all members of the view $cv$ via a \msg{propose} message (\Cref{line:propose_system_configuration}).


\ignore{
Importantly, we note that \msg{propose} messages carry signed \msg{reconfig} messages for updates they propose and correct processes only account for messages if they have been able to verify all signatures in the messages (we omit this in the pseudocode, but it is relevant for the correctness).
}

The leaving operation invocation is similar: Process $p$ disseminates a \msg{reconfig} message with $\langle -, p \rangle$ as an argument, and process $r$ proposes a new system membership that \emph{does not} include $p$.
The main difference with the joining operation is that if $p$ delivered or is the sender of a message, $p$ must ensure validity and totality properties of \name before disseminating a \msg{reconfig} message (\Cref{line:totality_validity}).

Let us now explain how a new 
view
is installed in the system.
The correct process $r \in cv$ receives \msg{propose} messages disseminated by other members of $cv$.
First, $r$ checks whether it \emph{accepts}\footnote{Process $r$ accepts a sequence of views $seq$ to replace a view $v$ (see \Cref{sec:preliminary_definitions}) if $seq \in \mathit{FORMAT}^v$ or $\emptyset \in \mathit{FORMAT}^v$ (\Cref{line:deliver_propose_brief}). The following holds at every correct process that is a member of the initial view of the system $v_0$: $\emptyset \in \mathit{FORMAT^{v_0}}$. Note that $\mathit{FORMAT}^v$, for any view $v$, is a set of sequences, i.e., a set of sets.} the received proposal (recall that a proposal is a sequence of views).
Moreover, $r$ checks whether the received proposed sequence $seq$ is well-formed, i.e., whether $seq$ satisfies the following: (1) $seq$ is a sequence of views, (2) there is at least one view in $seq$ that $r$ is not aware of, and (3) every view in $seq$ is more recent than $cv$.


If all the checks have passed, the process $r$ uses the received \msg{propose} message to update its own proposal.
This is done according to two cases:
\begin{compactenum}
    \item There are conflicts between $r$'s and the received proposal (\Cref{line:conflicting_start} to \Cref{line:update_conflicting}).
    In this case, $r$ creates a new proposal containing $r$'s last converged sequence for the view\footnote{We say that $seq$ is the last converged sequence for a view $v$ of a process if the process receives the same proposal to replace the view $v$ from a quorum of members of $v$ (variable $\mathit{LCSEQ}^v$).} and a new view representing the union of the most recent views of two proposals.
    \item There are no conflicts (\Cref{line:update_nonconflicting}).
    In this case, $r$ executes the union of its previous and received proposal in order to create a new proposal.
\end{compactenum}
Once $r$ receives the same proposal from a quorum of processes, $r$ updates its last converged sequence (\Cref{line:lcseq}) and disseminates it within a \msg{converged} message (\Cref{line:send_converged}).

\begin{algorithm}[t]
\footnotesize
\caption{\name algorithm: local variables of process $p$ and helper functions.}
\label{algorithm:variables}
\begin{algorithmic}[1]
\Variables
    \State $\mathit{cv} = v_0$ \Comment{current view; $v_0$ is the initial view}
    \State $\mathit{RECV} = \emptyset$ \Comment{set of pending \emph{updates} (i.e., join or leave)}
    \State $\mathit{SEQ^v} = \emptyset$ \Comment{set of proposed sequences to replace $v$}
    \State $\mathit{LCSEQ^v} = \emptyset$ \Comment{last converged sequence to replace $v$}
    \State $\mathit{FORMAT^v} = \emptyset$ \Comment{replacement sequence for view $v$}
    \State $\mathit{cer} = \bot$ \Comment{message certificate for $m$}
    \State $\mathit{v_{cer}} = \bot$ \Comment{view in which certificate is collected}
    \LineComment{set of messages allowed to be acknowledged; initially, any message could be acknowledged by a process}
    \State $\mathit{allowed\_ack} = \bot$ \Comment{$\bot$ - any message, $\top$ - no message}
    \State $\mathit{stored} = \mathit{false}; \text{ } \mathit{stored\_value} = \bot$
    \State $\mathit{can\_leave} = \mathit{false}$ \Comment{process is allowed to leave}
    \State $\mathit{delivered} = \mathit{false}$ \Comment{$m$ delivered or not}
    \LineComment{for every process $q \in \mathcal{U}$ and every valid view $v$}
    \State $\mathit{acks[q, v]} = \bot; \text{ } \mathit{\Sigma[q, v]} = \bot; \text{ } \mathit{deliver[q, v]} = \bot$ 
    \State $\mathit{State} = \bot$ \Comment{state of the process; consists of $\mathit{ack}$, $\mathit{conflicting}$ and $\mathit{stored}$ fields}
\EndVariables


\algvspace
\Function{least\_recent}{seq} \textbf{ returns } $\omega \in seq : \nexists \omega' \in seq:\omega' \subset \omega$
\EndFunction
\Function{most\_recent}{seq} \textbf{ returns } $\omega \in seq : \nexists \omega' \in seq:\omega \subset \omega'$
\EndFunction
\algstore{dynamic}
\end{algorithmic}
\end{algorithm}
\setlength{\textfloatsep}{5pt}

When $r$ receives a \msg{converged} message for some sequence of views $seq'$ and some view $v$ (usually $v$ is equal to the current view $cv$ of process $r$, but it could also be a less recent view than $cv$) from a quorum of members of the view $v$ (\Cref{line:seq_conv}), $r$ creates and reliably disseminates an \msg{install} message that specifies the view that should be replaced (i.e., $v$), the least recent view of the sequence $seq'$ denoted by $\omega$ (\Cref{line:calculate_least_updated}) and the entire sequence $seq'$ (\Cref{line:send_reliable}).
Moreover, we say that $seq'$ is \emph{converged on} to replace $v$.
An \msg{install} message is disseminated to processes that are members of views $v$ or $\omega$ (\Cref{line:send_reliable}).
Note that \msg{install} messages include a quorum of signed \msg{converged} messages which ensures its authenticity (omitted in \cref{algorithm:dynamic,algorithm:install} for brevity).

Once the correct process $r$ receives the \msg{install} message (\Cref{line:install_seq_delivery}), 
$r$ enters the installation procedure in order to update its current view of the system.
There are four parts to consider: 
\begin{compactenum}
    \item Process $r$ was a member of a view $v$ (\cref{line:stop_processing,line:send_state_update}):
    Firstly, $r$ checks whether $cv \subset \omega$, where $cv$ is the current view of $r$.
    If this is the case, $r$ stops processing \msg{prepare}, \msg{commit} and \msg{reconfig} messages (\Cref{line:stop_processing}; see \Cref{sec:broadcast-algorithm}). 
    Therefore, process $r$ will not send any \msg{ack} or \msg{deliver} message for \msg{prepare} or \msg{commit} messages associated with $v$ (and views preceding $v$).
    The same holds for \msg{reconfig} messages.
    We refer to acknowledged and stored messages by a process as the \emph{state} of the process (represented by the $\mathit{State}$ variable).
    The fact that $r$ stops processing the aforementioned messages is important because $r$ needs to convey this information via the \msg{state-update} message (\Cref{line:send_state_update})  
    to the members of the new view $\omega$.
    Therefore, a conveyed information is ``complete'' since a correct process $r$ will never process any \msg{prepare}, \msg{commit} or \msg{reconfig} message associated with ``stale'' views (see \Cref{sec:broadcast-algorithm}).
    \item View $\omega$ is more recent than $r$'s current view $cv$ (\Cref{line:state-update} to \Cref{line:state-transfer}):
    Process $r$ waits for $v.q$ of \msg{state-update} messages (\Cref{line:state-update}) and processes received states (\Cref{line:state-transfer}).
    \msg{state-update} messages carry information about: (1) a message process is allowed to acknowledge ($\mathit{allowed\_ack}$ variable), (2) a message stored by a process ($\mathit{stored\_value}$ variable), and (3) reconfiguration requests observed by a process (see \Cref{sec:broadcast-algorithm}).
    Hence, a \msg{state-update} message contains at most two \msg{prepare} messages associated with some view and properly signed by $s$ (corresponds to (1)).
    Two \msg{prepare} messages are needed if a process observes that $s$ broadcast two messages and are used to convince other processes not to acknowledge any messages (variable $\mathit{State}.\mathit{conflicting}$; \cref{line:no_ack_1,line:no_ack_2}).
    Moreover, \msg{state-update} messages contain at most one \msg{commit} message associated with some view with a valid message certificate (variable $\mathit{State}.\mathit{stored}$) and properly signed by $s$ (corresponds to (2)), and a (possibly empty) list of properly signed \msg{reconfig} messages associated with some installed view (corresponds to (3)).
    Note that processes include only \msg{prepare}, \msg{commit} and \msg{reconfig} messages associated with some view $v'' \subseteq v$ in the \msg{state-update} message they send (incorporated in the $\mathit{state}(v)$ function). 
    The reason is that processes receiving these \msg{state-update} messages may not know whether views $v'' \supset v$ are indeed ``created'' by our protocol and not ``planted'' by faulty processes. 
    \item Process $r$ is a member of $\omega \supset cv$ (\Cref{line:p_in} to \Cref{line:trigger_view_installed}): If this is the case, $r$ updates its current view (\Cref{line:set_cv}).
    Moreover, $r$ \emph{installs} the (updated) current view $cv$ if the sequence received in the \msg{install} message does not contain other views that are more recent than $cv$ (\Cref{line:installed_cv}).
    \item Process $r$ is not a member of $\omega \supset cv$ (\Cref{line:finish_wait_updated} to \Cref{line:left}): 
    A leaving process $r$ executes the View Discovery protocol (\Cref{line:discover}) in order to ensure totality of \name (we explain this in details in \Cref{sec:broadcast-algorithm}).
    When $r$ has ``fulfilled'' its role in ensuring totality of \name, $r$ leaves the system (\Cref{line:left}).
\end{compactenum}

\begin{algorithm}[h!]
\footnotesize
\caption{\dbrbJoin{} and \dbrbLeave{} implementations at process $p$.}
\label{algorithm:dynamic}
\begin{algorithmic}[1]
\algrestore{dynamic}

\Procedure{\emph{\dbrbJoin{}}}{\hphantom{}}
    \Repeat
    \State $\mathit{cv} = \text{view\_discovery}(cv)$ \label{line:start_view_discovery_1}
    \State disseminate \operation{reconfig}{\langle +, p \rangle, \mathit{cv}} to all $q \in \mathit{cv}.\mathit{members}$ \label{line:send_join}
    \Until{\textit{joinComplete} is triggered or $v.q$ \operation{rec-confirm}{v} messages collected for some $v$} \label{line:stop_asking_to_join}
    \State \textbf{wait for} \textit{joinComplete} \textbf{to be triggered}
\EndProcedure

\algvspace
\Procedure{\emph{\dbrbLeave{}}}{\hphantom{}}
    \IfThen
        {$\mathit{delivered}$ $\lor$ $p = s$}
        {\textbf{wait until} $\mathit{can\_leave}$}
        \label{line:totality_validity}
    \Repeat{} in each installed view $\mathit{cv}$ \textbf{do} \Comment{in each subsequent view $p$ installs}
        \State disseminate \operation{reconfig}{\langle -, p \rangle, \mathit{cv}} to all $q \in \mathit{cv}.\mathit{members}$ \label{line:send_leave}
    \Until{\textit{leaveComplete} is triggered or $v.q$ \operation{rec-confirm}{v} messages collected for some $v$}
    \State \textbf{wait for} \textit{leaveComplete} \textbf{to be triggered}
    
\EndProcedure

\algvspace
\UponReceipt{\operation{reconfig}{\langle c, q \rangle, v}}{q} \Comment{$c \in \{-, +\}$}
    \If{$v = \mathit{cv}$ $\land$ $\langle c, q \rangle \notin v$ $\land$ $($\textbf{if} $(c = -)$ \textbf{then} $\langle +, q \rangle \in v)$}
        \State $\mathit{RECV} = \mathit{RECV} \cup \{\langle c, q \rangle\}$ \label{line:add_reconfig_brief}
        \State $\text{send \operation{rec-confirm}{cv}} \text{ to } q$ \label{line:send_rec-confirm}
    \EndIf
\EndUponReceipt

\algvspace
\UponCondition{$\mathit{RECV} \neq \emptyset$ $\land$ $\mathit{installed}(cv)$}
    \If{$\mathit{SEQ^{cv}}= \emptyset$} \label{line:check_propose_1}
        \State $\mathit{SEQ^{cv}} = \{cv \cup RECV\}$
        \State $\text{disseminate}$ \operation{propose}{\mathit{SEQ^{cv}}, \mathit{cv}} to all $q \in \mathit{cv}.\mathit{members}$ \label{line:propose_system_configuration}
    \EndIf
\EndUponCondition

\algvspace
\UponReceiptFromST{\operation{propose}{seq, v}}{q \in v.\mathit{members}}{seq \in \mathit{FORMAT^v} \;\lor\; \emptyset \in \mathit{FORMAT^v}} \label{line:deliver_propose_brief}
    \If{valid($seq$)} \label{line:only_more_updated_1} \Comment{filter incorrect proposals}
        \If{conflicting($seq, \mathit{SEQ^{v}}$)} 
            \State $\omega = most\_recent(seq)$ \label{line:conflicting_start}
            \State $\omega' = most\_recent(\mathit{SEQ^{v}})$
            \LineComment{merge the last view from the local and $q$'s proposal}
            \State $\mathit{SEQ^{v}} = \mathit{LCSEQ^{v}} \cup \{\omega \cup \omega'\}$ \label{line:update_conflicting} 
        \Else \Comment{no conflicts, just merge the proposals}
            \State $\mathit{SEQ^{v}} = \mathit{SEQ^{v}} \cup seq$ \label{line:update_nonconflicting}
        \EndIf
        \State disseminate \operation{propose}{\mathit{SEQ^v}, v} to all $q' \in v.\mathit{members}$ \label{line:propose_second_case}
    \EndIf
\EndUponReceiptFromST

\algvspace
\UponReceiptCond{\operation{propose}{SEQ^{v}, v}} \label{line:verification3}
    \State $LCSEQ^{v} = SEQ^{v}$ \label{line:lcseq} 
    \State disseminate \operation{converged}{SEQ^v, v} to all $q \in v.\mathit{members}$ \label{line:send_converged}
\EndUponReceiptCond

\algvspace
\UponReceiptCond{\operation{converged}{seq', v}} \label{line:seq_conv}
    \State $\omega = \textit{least\_recent}(seq')$ \label{line:calculate_least_updated}
    \State $\textit{R-multicast}(\{j: j \in v.\mathit{members} \lor j \in \omega.\mathit{members}\}, \text{\operation{install}{\omega, seq', v}})$ \label{line:send_reliable} 
\EndUponReceiptCond
\algstore{dynamic}
\end{algorithmic}
\end{algorithm}

\begin{algorithm}[h!]
\footnotesize
\caption{\name algorithm: installing a view at process $p$.}
\label{algorithm:install}
\begin{algorithmic}[1]
\algrestore{dynamic}

\UponCondition{R-delivery$(\{j: j \in v.\mathit{members} \lor j \in \omega.\mathit{members}\}, \emph{\operation{install}{\omega, seq, v}})$}\label{line:install_seq_delivery}

    \State $\mathit{FORMAT^{\omega}} = \mathit{FORMAT^{\omega}} \cup \{seq \setminus \{\omega\}\}$ \label{line:set_format} 
    
    
    \If{$p \in v.\mathit{members}$} \Comment{$p$ was a member of $v$}
        \IfThen
            {$\mathit{cv} \subset \omega$}
            {stop processing \msg{prepare}, \msg{commit} and \msg{reconfig} messages}
            \label{line:stop_processing}
        \State $\textit{R-multicast}(\{j: j \in v.\mathit{members} \lor j \in \omega.\mathit{members}\}, \text{\operation{state-update}{state(v), RECV}})$ \label{line:send_state_update}
    \EndIf
    \If{$\mathit{cv} \subset \omega$} \label{line:omega_more} \Comment{$\omega$ is more recent than $p$'s current view}
        \State \textbf{wait} for \operation{state-update}{*, *} messages from $v.q$ processes in $v$ \label{line:state-update} \Comment{from the reliable broadcast}
        \State $\mathit{req} = \{\text{reconfiguration requests from \msg{state-update}}$ messages\}
        \State $\mathit{RECV} = \mathit{RECV} \cup (req \setminus \omega.changes)$ \label{line:recv}
        \State $\mathit{states} = \{\text{states from \msg{state-update} messages}$\}
        \State $\mathit{installed}(\omega) = \mathit{false}$
        \State \textbf{invoke } \text{\emph{state-transfer($\mathit{states}$)}} \label{line:state-transfer} \Comment{\Cref{algorithm:b}}
        \If{$p \in \omega.\mathit{members}$} \Comment{$p$ is in $\omega$} \label{line:p_in}
            \State $\mathit{cv} = \omega$ \label{line:set_cv}
            \IfThen
                {$p \notin v.\mathit{members}$}
                {\textbf{trigger} \textit{joinComplete}}
                \label{line:join_completed} 
                \Comment{can return from \dbrbJoin{}}
            \If{$\exists \omega' \in seq: \mathit{cv} \subset \omega'$} \label{line:more_updated_views_1}
                \State $seq' = \{\omega' \in seq: \mathit{cv} \subset \omega'\}$ 
                \If{$\mathit{SEQ^{cv}} = \emptyset \land \forall \omega \in seq': \mathit{cv} \subset \omega$} \label{line:check_propose_2}
                    \State $\mathit{SEQ^{cv}} = seq'$
                    \State disseminate \operation{propose}{\mathit{SEQ^{cv}}, \mathit{cv}} to all $q \in \mathit{cv}.\mathit{members}$ \label{line:propose_rest} 
                \EndIf
            \Else
                \State $\mathit{installed}(cv) = \mathit{true}$ \label{line:installed_cv}
                \State resume processing \msg{prepare}, \msg{commit} and \msg{reconfig} messages\label{line:install_complete}
                \State \textbf{invoke } \text{\emph{new-view()}} \label{line:trigger_view_installed} \Comment{\Cref{algorithm:b}}
            \EndIf
        \Else \Comment{$p$ is leaving the system}
            
            \If{$\mathit{stored}$}  \label{line:finish_wait_updated}
                \While{\lnot \mathit{can\_leave}}
                    \State $\mathit{cv} = \text{view\_discovery}(\mathit{cv})$ \label{line:discover}
                    \State disseminate \operation{commit}{m, \mathit{cer}, \mathit{v_{cer}}, \mathit{cv}} to all $q \in \mathit{cv}.\mathit{members}$ \label{line:relay_store_brief_2}
                \EndWhile
            \EndIf
            
            \State\textbf{trigger} \textit{leaveComplete} \label{line:left} \Comment{can return from \dbrbLeave{}}
        \EndIf
    \EndIf
\EndUponCondition

\algstore{dynamic}
\end{algorithmic}
\end{algorithm}
\noindent\textbf{View Discovery.}
We show in \Cref{sec:preliminary_lemmata} that views ``created'' during an execution of \name form a sequence.
The View Discovery subprotocol provides information about the sequence of views incorporated in an execution so far.
Since every correct process in the system knows the initial view (\Cref{assumption:init}) and valid transition between views implies the existence of an \msg{install} message with a quorum of properly signed \msg{converged} messages, any sequence of views starting from the initial view of the system such that appropriate \msg{install} messages ``connect'' adjacent views can be trusted.

A correct process that has invoked the \dbrbJoin{} operation and has not left the system executes the View Discovery subprotocol constantly.
Once a correct process starts trusting a sequence of views, it disseminates that information to all processes in the universe.
A correct process executing the View Discovery subprotocol learns which sequences of views are trusted by other processes.
Once the process observes a sequence of views allegedly trusted by a process, it can check whether the sequence is properly formed (as explained above) and if that is the case, the process can start trusting the sequence and views incorporated in it (captured by the view\_discovery function for the joining and leaving process; \cref{line:start_view_discovery_1,line:discover}).

The View Discovery protocol addresses two main difficulties: (1) it enables processes joining and leaving the system to learn about the current membership of the system (\Cref{sec:dynamicity_proof}), (2) it is crucial to ensure the consistency, validity and totality properties of \name since it supplies information about views ``instantiated'' by the protocol and associated quorum systems (\Cref{sec:broadcast_proof}).
We formally discuss the View Discovery protocol in \Cref{subsection:view_discovery}.

\vspace*{-3mm}
\subsection{Broadcast}
\label{sec:broadcast-algorithm}

\begin{algorithm} [h!]
\footnotesize
\caption{\dbrbBroadcast{m} and \dbrbDeliver{m} implementations at process $p$.}
\label{algorithm:b}
\begin{algorithmic}[1]
\algrestore{dynamic}



\Procedure{state-transfer}{states}
    \If{($\mathit{allowed\_ack} = \bot \lor \mathit{allowed\_ack} = m) \land m$ is the only acknowledged message among $\mathit{states}$} \label{line:set_can_ack_brief}
        \State $\mathit{allowed\_ack} = m$; $\mathit{update\_if\_bot}(\mathit{State}.\mathit{ack}, \mathit{prepare\_msg)}$ \Comment{updated only if it is $\bot$}
    \ElsIf{there exist at least two different messages acknowledged among $\mathit{states}$} \label{line:no_ack_1}
        \LineComment{$p$ and $p'$ are different \msg{prepare} messages}
        \State $\mathit{allowed\_ack} = \top$; $\mathit{update\_if\_bot}(\mathit{State}.\mathit{conflicting}, p, p')$; $\mathit{State}.\mathit{ack} = \bot$
    \ElsIf{there exists a state among $\mathit{states}$ such that it provides two different broadcast messages} \label{line:no_ack_2}
        \State $\mathit{allowed\_ack} = \top$; $\mathit{update\_if\_bot}(\mathit{State}.\mathit{conflicting}, p, p')$; $\mathit{State}.\mathit{ack} = \bot$
    \EndIf
    
    \If{$\lnot \mathit{stored}$ $\land$ there exists a stored message $m$ with a valid message certificate among $\mathit{states}$} \label{line:verify_cer_2_brief}
        \State $\mathit{stored} = \mathit{true}; \text{ } \mathit{stored\_value} = (m, \mathit{cer}, \mathit{v_{cer}})$ \Comment{$\mathit{cer}$ is the message certificate collected in view $\mathit{v_{cer}}$} \label{line:store_first_part}
        \State $\mathit{update\_if\_bot}(\mathit{State}.\mathit{stored}, \mathit{commit\_msg})$ \label{line:state_stored_first_part} \Comment{updated only if it is $\bot$}
    \EndIf

\EndProcedure

\algvspace
\Procedure{new-view}{\hphantom{}} \label{line:new_view_start}
    \IfThen
        {$p = s \land \mathit{cer} = \bot$} 
        {$\text{disseminate}$ \operation{prepare}{m, \mathit{cv}} to all $q \in \mathit{cv}.\mathit{members}$}
        \label{line:rebroadcast_prepare_brief}
    \IfThen
        {$p = s$ $\land$ $\mathit{cer} \neq \bot$ $\land$ $\lnot \mathit{can\_leave}$} 
        {$\text{disseminate}$ \operation{commit}{m, \mathit{cer}, \mathit{v_{cer}}, \mathit{cv}} to all $q \in \mathit{cv}.\mathit{members}$}
        \label{line:rebroadcast_commit_brief}
    \IfThen
        {$p \neq s$ $\land$ $\mathit{stored}$ $\land$ $\lnot\mathit{can\_leave}$} 
        {$\text{disseminate}$ \operation{commit}{m, \mathit{cer}, \mathit{v_{cer}}, \mathit{cv}} to all $q \in \mathit{cv}.\mathit{members}$} 
        \label{line:relay_commit_2_brief}
\EndProcedure

\algvspace
\Procedure{\emph{\dbrbBroadcastAlone{}}}{m}
    \IfThen
        {$\mathit{installed}(cv)$} \label{line:is_installed_prepare}
        {$\text{disseminate}$ \operation{prepare}{m, \mathit{cv}} to all $ q \in \mathit{cv}.\mathit{members}$} 
        \label{line:send_prepare_brief}
\EndProcedure

\algvspace
\UponReceiptFromST{\operation{prepare}{m, v}}{s \in v.\mathit{members}}{v = \mathit{cv}} \label{line:deliver_prepare_brief}
    \If{$\mathit{allowed\_ack} = m$ $\lor$ $\mathit{allowed\_ack} = \bot$} \label{line:check_allowed_brief}
        \State $\mathit{allowed\_ack} = m$; $\mathit{update\_if\_bot}(\mathit{State}.\mathit{ack},$ \operation{prepare}{m, v}) \label{line:set_state_ack} \Comment{updated only if it is $\bot$}
        \State $\sigma = sign(m, \mathit{cv}); \text{ send \operation{ack}{m, \sigma, \mathit{cv}}} \text{ to } s$ \label{line:send_ack_brief} 
    \EndIf
\EndUponReceiptFromST

\algvspace
\UponReceipt{\operation{ack}{m, \sigma, v}}{q \in v.\mathit{members}} \Comment{only process $s$}
    \IfThen{$\mathit{acks[q, v]} = \bot$ $\land$ $\mathit{verifysig}(q, m, v, \sigma)$}
        {$\mathit{acks[q, v]} = m; \text{ } \Sigma[q, v] = \sigma$}
\EndUponReceipt

\algvspace
\UponExists{m \neq \bot \text{ and } v}{|\{q \in v.\mathit{members}| \mathit{acks[q, v]} = m\}| \geq v.q \land \mathit{cer} = \bot} \label{line:certificate_collected_brief}
    \State $\mathit{cer} = \{\Sigma[q, v]: \mathit{acks[q, v]} = m \}; \text{ } \mathit{v_{cer}} = v$
    \IfThen
        {$\mathit{installed}(cv)$}
        {$\text{disseminate}$ \operation{commit}{m, \mathit{cer}, \mathit{v_{cer}}, \mathit{cv}} to all $q' \in \mathit{cv}.\mathit{members}$}
        \label{line:send_commit_brief}
\EndUponExists


\algvspace
\UponReceiptFromST{\operation{commit}{m, cer, v_{cer}, v}}{q}{v = \mathit{cv}} \label{line:deliver_store}
    \If{$\mathit{verify\_certificate(cer, v_{cer}, m)}$} \label{line:verify_cert_3_brief}
        \If{$\lnot \mathit{stored}$} \label{line:verify_stored_brief}
            \State $\mathit{stored} = \mathit{true}; \text{ } \mathit{stored\_value} = (m, \mathit{cer}, \mathit{v_{cer}})$ \label{line:store}
            \State $\mathit{update\_if\_bot}(\mathit{State}.\mathit{stored},$ \operation{commit}{m, cer, v_{cer}, v}) \Comment{updated only if it is $\bot$} \label{line:update_state_stored}
            \State $\text{disseminate}$ \operation{commit}{m, \mathit{cer}, \mathit{v_{cer}}, \mathit{cv}} to all $q' \in \mathit{cv}.\mathit{members}$ \label{line:relay_store_brief}
        \EndIf
        \State $\text{send \operation{deliver}{m, \mathit{cv}}} \text{ to } q$ \label{line:send_ack-store_brief}
    \EndIf
\EndUponReceiptFromST

\algvspace
\UponReceipt{\operation{deliver}{m, v}}{q \in v.\mathit{members}}
    \IfThen{$\mathit{deliver[q, v]} = \bot$}
        {$\mathit{deliver[q, v]} = \top$}
\EndUponReceipt

\algvspace
\UponExists{v}{|\{q \in v.\mathit{members}| \mathit{deliver[q, v]} = \top\}| \geq v.q \text{ for the first time}} \label{line:ack-store_quorum_brief}
        \State $\mathit{delivered} = \mathit{true}$
        \State \textbf{invoke } \dbrbDeliver{m} \label{line:deliver}
        \State $\mathit{can\_leave} = \mathit{true}$ \label{line:set_can_leave} \Comment{If $p = s$, \dbrbBroadcastAlone{} is completed}
\EndUponExists


\end{algorithmic}
\end{algorithm}

\vspace*{-3mm}
In order to broadcast some message $m$, processes in \name use the following types of messages:

    \noindent\msg{prepare}: When a correct process $s$ invokes a \dbrbBroadcast{m} operation, the algorithm creates a $m_{prepare} = \text{\operation{prepare}{m, \mathit{cv_s}}}$ message, where $\mathit{cv_s}$ is the current view of the system of process $s$.
    Message $m_{prepare}$ is sent to every process that is a member of $\mathit{cv_s}$ (\Cref{line:send_prepare_brief}).
    Process $s$ disseminates the \msg{prepare} message if $\mathit{cv_s}$ is installed by $s$; otherwise, $s$ does not disseminate the message to members of $cv_s$ (\Cref{line:is_installed_prepare}), but rather waits to install some view and then disseminates the \msg{prepare} message (\Cref{line:rebroadcast_prepare_brief}).
    
    \noindent\msg{ack}: When a correct process $q$ receives $m_{prepare}$ message, $q$ firstly checks whether view specified in $m_{prepare}$ is equal to the current view of $q$ (\Cref{line:deliver_prepare_brief}). 
    If that is the case, $q$ checks whether it is allowed to send an \msg{ack} message for $m$ (see Consistency paragraph in \Cref{sec:correctness-informal}; \Cref{line:check_allowed_brief}) and if it is, $q$ sends $m_{ack} = \text{\operation{ack}{m, \sigma, \mathit{cv_q}}}$ message to process $s$ (i.e., the sender of $m_{prepare}$), where $\sigma$ represents the signed statement that $s$ sent $m$ to $q$ (\Cref{line:send_ack_brief}).
    When some process $q$ sends an \msg{ack} message for $m$ ($m$ is a second argument of the message), we say that $q$ \emph{acknowledges} $m$.
    Moreover, if an \msg{ack} message is associated with some view $v$, we say that $q$ acknowledges $m$ in a view $v$.
    
    \noindent\msg{commit}: When process $s$ receives a quorum of appropriate \msg{ack} messages associated with the same view $v$ for $m$ (\Cref{line:certificate_collected_brief}), $s$ collects received signed statements into a \emph{message certificate}.
    Process $s$ then creates $m_{commit} = \text{\operation{commit}{m, \mathit{cer}, \mathit{v_{cer}}, \mathit{cv_s}}}$ message and disseminates $m_{commit}$ to every process that is a member of $\mathit{cv_s}$ (\Cref{line:send_commit_brief}).
    Note that $\mathit{cv_s}$ may be different from $v$ (we account for this in the rest of the section).
    Moreover, $s$ disseminates the \msg{commit} message (\Cref{line:send_commit_brief}) if $\mathit{cv_s}$ is installed by $s$; otherwise, $s$ does not disseminate the message to members of $cv_s$, but rather waits to install some view and then disseminates the \msg{commit} message (\Cref{line:rebroadcast_commit_brief}).
    
    \noindent\msg{deliver}: When a correct process $q$ receives $m_{commit}$ message, it firstly checks whether view specified in $m_{commit}$ is equal to the current view of $q$ (\Cref{line:deliver_store}).
    If that is the case and the message certificate is valid (\Cref{line:verify_cert_3_brief}), $q$ ``stores'' $m$ (\Cref{line:store}) and sends $m_{deliver} = \text{\operation{deliver}{m, cv_q}}$ to process $s$ as a an approval that $s$ can deliver $m$ (\Cref{line:send_ack-store_brief}).
    When a process $q$ executes \Cref{line:store} or \Cref{line:store_first_part} for a message $m$, we say that $q$ \emph{stores} $m$.
    Observe that $q$ also disseminates $m_{commit}$ in order to deliver $m$ itself (\cref{line:relay_store_brief,line:relay_commit_2_brief}).
    
    Lastly, once a correct process receives a quorum of appropriate \msg{deliver} messages associated with the same view $v$ for $m$ (\Cref{line:ack-store_quorum_brief}), it delivers $m$ (\Cref{line:deliver}).

Every \msg{prepare}, \msg{ack}, \msg{commit} and \msg{deliver} message is associated with one specific view.

We can divide the broadcasting of message $m$ by the correct sender $s$ into two phases:
\begin{compactitem}
    \item\textbf{Certificate collection phase}: This phase includes a dissemination of an appropriate \msg{prepare} message and a wait for a quorum of \msg{ack} messages by process $s$.
    Note that \msg{prepare} and \msg{ack} messages are associated with the same view $v$.
    We say that certificate collection phase is \emph{executed} in view $v$.
    Moreover, if $s$ indeed receives a quorum of \msg{ack} messages associated with $v$, we say that certificate collection phase is \emph{successfully executed} in $v$.
    In that case, sometimes we say that $s$ collects a message certificate in $v$.
    
    \item\textbf{Storing phase}: In this phase, each correct process $p$ (including $s$) disseminates a \msg{commit} message (containing a valid message certificate collected in the previous phase), and waits for a quorum of \msg{deliver} messages.
    Note that \msg{commit} and \msg{deliver} messages are associated with the same view $v$.
    We say that storing phase is \emph{executed} in view $v$.
    Moreover, if $p$ indeed receives a quorum of \msg{deliver} messages associated with $v$, we say that storing phase is \emph{successfully executed} in $v$.
\end{compactitem}
Observe that the certificate collection phase can be successfully executed in some view $v$, whereas the storing phase can be executed in some view $v' \supset v$.
This is the reason why we include $v_{cer}$ argument in a \msg{commit} message, representing the view in which a message certificate is collected.
Lastly, in order to ensure validity and totality, processes must disseminate \msg{prepare} and \msg{commit} messages in new views they install until they collect enough \msg{ack} and \msg{deliver} messages, respectively.
This mechanism is captured in the \textit{new-view} procedure (\Cref{line:new_view_start}) that is invoked when a view is installed (\Cref{line:trigger_view_installed}).

\vspace*{-6mm}
\subsection{Illustration} 
\label{sec:execution}

\vspace*{-2mm}
Consider four participants at time $t = 0$.
Process $p_1$ broadcasts a message $m$.
Hence, $p_1$ sends to processes $p_1$, $p_2$, $p_3$, $p_4$ a $m_{prepare} = \text{\operation{prepare}{m, v_0}}$ message, where $v_0 = \{\langle +, p_1 \rangle, \langle +, p_2 \rangle, \langle +, p_3 \rangle, \langle +, p_4 \rangle\}$.
Since all processes consider $v_0$ as their current view of the system at time of receiving of $m_{prepare}$ message, they send to $p_1$ an appropriate \msg{ack} message and $p_1$ collects a quorum (with respect to $v_0$) of \msg{ack} messages for $m$.

However, process $p_5$ invokes a \dbrbJoin{} operation and processes $p_2$, $p_3$, $p_4$, $p_5$ set $v_1 = \{\langle +, p_1 \rangle, \langle +, p_2 \rangle, \langle +, p_3 \rangle, \langle +, p_4 \rangle, \langle +, p_5 \rangle\}$ as their current view of the system.
Process $p_1$ still considers $v_0$ as its current view of the system and disseminates a $m_{commit} = \text{\operation{commit}{m, \mathit{cer}, \mathit{v_{cer}} = v_0, v_0}}$ message.
Processes $p_2$, $p_3$ and $p_4$ do not store $m$, since $v_0$ (specified in $m_{commit}$ message) is not their current view.
Observe that $p_1$ stores $m$ since $v_0$ is still the current view of the system from $p_1$'s perspective.

Once process $p_1$ assigns $v_1$ as its current view of the system, it disseminates $m_{commit} = \text{\operation{commit}{m, \mathit{cer}, \mathit{v_{cer}} = v_0, v_1}}$ message to processes that are members of $v_1$ and they all store $m$ and relay $m_{commit}$ message to all processes that are members of $v_1$.
Hence, they all deliver message $m$ once they collect a quorum (with respect to $v_1$) of matching \msg{deliver} messages.
In this execution $p_1$ has successfully executed a certificate collection phase in $v_0$ and then reused the message certificate to relay an appropriate \msg{commit} message to processes that are members of $v_1$ (since the system has reconfigured to $v_1$).
Note that \Cref{fig:example} depicts the described execution.
For presentation simplicity, \Cref{fig:example} just shows the \msg{commit} and \msg{deliver} messages that allow process $p_1$ to deliver $m$.

\begin{figure}[h]
    \centering
    \includegraphics[width=.6\columnwidth]{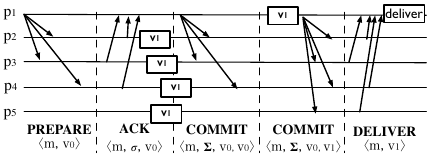}
    \vspace*{-3mm}
    \caption{Example of a broadcast operation in \name algorithm, considering a dynamic membership.}
    \label{fig:example}
\end{figure}
\vspace*{-8mm}

%% file: sections/correctness-informal.tex
\section{DBRB Algorithm Correctness}
\label{sec:correctness-informal}

\vspace*{-2mm}
We now give an intuition of why  our \name algorithm is correct; we give formal arguments in~\cref{sec:preliminary_definitions,sec:preliminary_lemmata,subsection:view_discovery,sec:dynamicity_proof,sec:broadcast_proof}.

We first define the notions of \emph{valid} and \emph{installed} views.
A view $v$ is \emph{valid} if: (1) $v$ is the initial view of the system, or (2) a sequence $seq = v \to ...$ is converged on to replace some valid view $v'$.
A valid view $v$ is \emph{installed} if a correct process $p \in v$ processed \msg{prepare}, \msg{commit} and \msg{reconfig} messages associated with $v$ during an execution.
By default, the initial view of the system is installed.
Lastly, our implementation ensures that installed views form a sequence of views (we prove this formally in \Cref{sec:preliminary_lemmata}).



\smallskip
\noindent\textbf{Liveness.}
\dbrbJoin{} and \dbrbLeave{} operations complete because any 
change ``noticed'' by a quorum of processes is eventually  processed.
Intuitively, a sequence can be converged on if a quorum of processes propose that sequence.
Moreover, noticed changes are transferred to new valid views.
\dbrbBroadcastAlone{} operation completes since a correct sender eventually collects a quorum of \msg{deliver} messages associated with an installed view (see the next paragraph).

\smallskip
\noindent\textbf{Validity.}
Recall that we assume a finite number of reconfiguration requests in any execution of \name (\Cref{assumption:requests}), which means that there exists a view $v_{final}$ from which the system will not be reconfigured (we prove this in \Cref{sec:preliminary_lemmata}).
In order to prove validity, it suffices to show that every correct member of $v_{final}$ delivers a broadcast message.

A correct process $s$ that broadcasts a message $m$ executes a certificate collection phase in some installed view $v$ (the current view of $s$).
Even if $s$ does not successfully execute a certificate collection phase in views that precede $v_{final}$ in the sequence of installed views, $s$ successfully executes a certificate collection phase in $v_{final}$ (proven in \Cref{sec:broadcast_proof}).
Note that process $s$ does not leave the system before it collects enough \msg{deliver} messages (ensured by the check at \Cref{line:totality_validity} and the assignment at \Cref{line:set_can_leave}).

Moreover, a correct process $p$ that stored a message $m$ eventually collects a quorum of \msg{deliver} messages associated with some installed view $v$.
As in the argument above, even if $p$ does not collect a quorum of \msg{deliver} messages associated with views that precede $v_{final}$ in the sequence of installed views, $p$ does that in $v_{final}$ (proven in \Cref{sec:broadcast_proof}).
Observe that if at least a quorum of processes that are members of some installed view $v$ store a message $m$, then every correct process $p \in v'$, where view $v' \supseteq v$ is installed, stores $m$.
Let us give the intuition behind this claim.
Suppose that view $v'$ directly succeeds view $v$ in the sequence of views installed in the system.
Process $p \in v'$ waits for states from at least a quorum of processes that were members of $v$ (\Cref{line:state-update}) before it updates its current view to $v'$.
Hence, $p$ receives from at least one process that $m$ is stored and then $p$ stores $m$ (\Cref{line:verify_cer_2_brief}).
The same holds for the correct members of $v$.

It now suffices to show that $s$ collects a quorum (with respect to some installed view) of confirmations that $m$ is stored, i.e., \msg{deliver} messages.
Even if the correct sender does not collect a quorum of \msg{deliver} messages in views that precede $v_{final}$, it collects the quorum when disseminating the \msg{commit} message to members of $v_{final}$.
Suppose now that the sender collects the aforementioned quorum of \msg{deliver} messages in some installed view $v$.
If $v \neq v_{final}$, every correct member of $v_{final}$ stores and delivers $m$ (because of the previous argument).
If $v = v_{final}$, the reliable communication and the fact that the system can not be further reconfigured guarantee that every correct member of $v_{final}$ stores and, thus, delivers $m$.

\smallskip
\noindent\textbf{Totality.}
The intuition here is similar to that behind ensuring validity.
Consider a correct process $p$ that delivers a message $m$: 
$p$ successfully executed a storing phase in some installed view $v$.
This means that every member of an installed view $v' \supseteq v$ stores $m$.
Consider a correct participant $q$ that expressed its will to leave after process $p$ had delivered $m$.
This implies that $q \in v''$, where $v'' \supseteq v$ is an installed view, which means that process $q$ eventually stores $m$.
As in the previous paragraph, we conclude that process $q$ eventually collects enough \msg{deliver} messages associated with some installed view and delivers $m$. 

\smallskip
\noindent\textbf{Consistency.}
A correct process delivers a message only if there exists a message certificate associated with the message (the check at \Cref{line:verify_cert_3_brief}).
Hence, the malicious sender $s$ must collect message certificates for two different messages in order for the consistency to be violated.

Suppose that process $s$ has successfully executed a certificate collection phase in some installed view $v$ for a message $m$.
Because of the quorum intersection and the verification at \Cref{line:check_allowed_brief}, it is impossible for $s$ to collect a valid message certificate in $v$ for some message $m' \neq m$.
Consider now an installed view $v'$ that directly succeeds view $v$ in the sequence of installed views.
Since $s$ collected a message certificate for $m$ in $v$,
every correct process $p \in v'$ receives from at least one process from the view $v$ that it is allowed to acknowledge only message $m$ (\Cref{line:set_can_ack_brief}).
It is easy to see that this holds for every installed view $v'' \supset v'$.
Therefore, if $s$ also collects a message certificate for some message $m'$, then $m' = m$ and the consistency holds.

\smallskip
\noindent\textbf{No duplication.}
Trivially follows from \Cref{line:ack-store_quorum_brief}.

\smallskip
\noindent\textbf{Integrity.}
Consider a correct process $q$ that delivers a message $m$.
There is a message certificate for $m$ collected in some installed view $v$ by $s$.
A message certificate for $m$ is collected since a quorum of processes in $v$ have sent an appropriate \msg{ack} message for $m$.
A correct process sends an \msg{ack} message only when it receives an appropriate \msg{prepare} message.
Consequently, message $m$ was broadcast by $s$.

%% file: sections/conclusion.tex
\vspace*{-4mm}
\section{Related Work \& Conclusions}
\label{sec:conclusions}



\vspace*{-2mm}
\noindent\textbf{DBRB vs. Static Byzantine Reliable Broadcast.}
Our \name abstraction generalizes \emph{static} \reliablbcast (Byzantine reliable broadcast~\cite{br85acb,ma97secure}).
\ignore{
A textbook specification of \reliablbcast ensures these five properties:
(1) \emph{Validity}: If a correct process $s$ broadcasts a message $m$, then every correct process eventually delivers $m$.
(2) \emph{Totality}: If a correct process $p$ delivers a message $m$, then every correct process $q$ eventually delivers $m$.
(3) \emph{No duplication}: No correct process delivers more than one message.
(4) \emph{Integrity}: If some correct process delivers a message $m$ with sender $s$ and process $s$ is correct, then $s$ previously broadcast $m$.
(5) \emph{Consistency}: If some correct process delivers a message $m$ and another correct process delivers a message $m'$, then $m = m'$.
}
Assuming that no process joins or leaves the system, the two abstractions coincide. 
In a dynamic setting, the validity property of \name stipulates that only processes that do not leave the system deliver the appropriate messages.
Moreover, the totality property guarantees that only processes that have not expressed their will to leave deliver the message.
We prove that stronger variants of these properties are impossible in our model.

\smallskip
\noindent\textbf{Passive and Active Reconfiguration.}
Some reconfigurable systems~\cite{BaldoniBKR09,AttiyaCEKW19,collect-churn-tr} assume that processes join and leave the system under a specific \emph{churn} model.  
Intuitively, the consistency properties of the implemented service, e.g., an atomic storage, are ensured assuming that the system does not evolve too quickly and there is always a certain fraction of correct members in the system.   
In \name, we model this through the quorum system assumption on valid views (\Cref{assumption:quorums}).
Our system model also assumes that booting finishes by time $0$ (\Cref{assumption:init}), thus avoiding the problem of unbounded booting times which could be problematic in asynchronous network~\cite{wid07booting}.

%
\emph{Active} reconfiguration allows the processes to explicitly propose configuration updates, e.g., sets of new process members. 
 In \emph{DynaStore}~\cite{aguilera2011dynamic},   reconfigurable dynamic atomic storage is implemented in an asynchronous environment (i.e., without relying on consensus).
Dynastore implicitly generates a graph of views which provides a way of identifying a sequence of views in which clients need to execute their r/w operations. 
%
%
SpSn~\cite{gafni2015elastic} proposes to capture this order via the \emph{speculating snapshot} algorithm (SpSn).
%
SmartMerge~\cite{jehl2015smartmerge} implements a reconfigurable storage 
in which not only system membership but also its quorum system can be reconfigured, assuming that a \emph{static} external lattice agreement is available. 
In~\cite{rla}, it was shown that \emph{reconfigurable lattice agreement} 
can 
get rid of this assumption and still implement a large variety of reconfigurable objects.
The approach was then extended to  the Byzantine fault model~\cite{bla}. 
 FreeStore~\cite{alchieri2016efficient} introduced \emph{view generator}, an abstraction that captures the agreement demands of reconfiguration protocols.
 Our work is highly inspired by FreeStore, which algorithmic and theoretical approach we adapt to an arbitrary failure model.

All reconfigurable solutions discussed above were applied exclusively to shared-memory emulations.
Moreover, most of them assumed the crash fault model.
In contrast, in this paper, we address the problem of dynamic \emph{reliable broadcast}, assuming an arbitrary (Byzantine) failure model. 
Also, we do not distinguish between clients and replicas, and assume 
that every process can only suggest itself as a candidate to join or leave the system.  
Unlike the concurrent work by Kumar and Welch on Byzantine-tolerant registers~\cite{KW19}, our solution can tolerate unbounded number of Byzantine failures, as long as basic quorum assumptions on valid views are maintained.        

\smallskip
\noindent\textbf{Broadcast Applications.}
Reliable broadcast is one of the most pervasive primitives in distributed applications~\cite{pedone2002handling}.
For instance, broadcast can be used for maintaining caches in cloud services~\cite{Geng2013}, or in a publish-subscribe network~\cite{broadcastApplications}.
Even more interestingly, Byzantine fault-tolerant reliable broadcast (e.g., dynamic solution such as our \name, as well as static solutions~\cite{br85acb,guerr19:11329,malkhi1997high}) are sufficiently strong for implementing decentralized online payments, i.e., cryptocurrencies~\cite{guerraoui2019consensus}.

\smallskip
\noindent\textbf{Summary.}
This paper presents the specification of \name (dynamic Byzantine reliable broadcast), as well as an asynchronous algorithm implementing this primitive.
\name generalizes traditional Byzantine reliable broadcast, which operates in static environments, to work in a dynamic network.
To the best of our knowledge, we are the first to investigate an arbitrary failure model in implementing dynamic broadcast systems.
The main merit of our approach is that we did not rely on a consensus building blocks, i.e., \name can be implemented completely asynchronously.
\vspace*{-6mm}

%% file: sections/appendix/correctness.tex


\subsection{Preliminary Definitions}
\label{sec:preliminary_definitions}

\vspace*{-2mm}
\begin{definition}
Suppose that a correct process $p \in v$ broadcasts a \emph{\operation{propose}{seq, v}} message at \Cref{line:propose_system_configuration} or at \Cref{line:propose_rest}.
We say that $p$ \textbf{proposes} a sequence of views $seq$ to replace a view $v$.
\end{definition}

\begin{definition}
Suppose that a correct process $p \in v$ collects a quorum (with respect to $v$, i.e., $v.q$) of \emph{\operation{converged}{seq, v}} messages (\Cref{line:seq_conv}).
We say that $p$ \textbf{converges on} a sequence of views $seq$ to replace a view $v$.
\end{definition}

\begin{definition}
We say that a sequence of views $seq$ is \textbf{converged on} to replace a view $v$ if some (correct or faulty) process $p \in v$ collects a quorum (with respect to $v$, i.e., $v.q$) of \emph{\operation{converged}{seq, v}} messages.
\end{definition}

\begin{definition} \label{definition:accepts}
Suppose that $seq \in FORMAT^v$ at a correct process $p$ for some view $v \ni p$.
We say that $p$ \textbf{accepts} a sequence of views $seq$ to replace a view $v$.
Moreover, if $\emptyset \in FORMAT^v$, then $p$ accepts \textbf{any} sequence of views to replace $v$.
\end{definition}

\begin{definition}
We say that a sequence of views $seq$ is \textbf{accepted} to replace a view $v$ if some correct process $p \in v$ accepts $seq$ to replace $v$.
\end{definition}

We prove in \Cref{sec:preliminary_lemmata} that if $seq$ is converged on or a correct process proposes $seq$ to replace some view, then $seq$ is a sequence of views.
Finally, if $seq$ is accepted to replace some view, then $seq$ is a sequence of views.

\input{sections/appendix/core-properties.tex}
\input{sections/appendix/view-discovery}
\input{sections/appendix/dynamic-properties.tex}
\input{sections/appendix/broadcast-properties.tex}

%% file: sections/appendix/core-properties.tex

\subsection{Preliminary Lemmata} \label{sec:preliminary_lemmata}

\vspace*{-2mm}
We define the \emph{system} as a state machine $\Pi = \langle s_i, \mathcal{S}, \mathcal{E}, \mathcal{D} \rangle$, where $s_i$ is the initial state of the system, $\mathcal{S}$ represents the set of all possible states, $\mathcal{E}$ is the set of events system can observe and $\mathcal{D} \subseteq \mathcal{S} \times \mathcal{E} \times \mathcal{S}$ is a set of transitions.

Each state $s \in \mathcal{S}$ is defined as a set of views, i.e., $s = \{v_1, ..., v_n\}$, where $n \geq 1$.
Moreover, for each view $v \in s$, we define a logical predicate $\alpha_s(v)$.
If $\alpha_s(v) = \top$ for some state $s \in \mathcal{S}$ and some view $v \in s$, we say that $v$ is \textbf{\emph{installable in}} $s$.
Otherwise, we say that $v$ is \textbf{\emph{not installable in}} $s$.
Initial state $s_i = \{v_0\}$, where $v_0$ represents the initial view of the system.
Lastly, $\alpha_{s_i}(v_0) = \top$.

Moreover, a set of sequences of views denoted by $\beta_s(v)$ is associated with each view $v \in s$.
If $seq \in \beta_s(v)$ and $v$ is not installable in $s$, we say that $v$ \emph{accepts $seq$ to replace $v$ in $s$}.
Consider a set $\mathcal{V}_v$ that represents the set of all possible sequences of views $seq$ such that $\forall \omega \in seq: v \subset \omega$.
If $v$ is installable in $s$, then $v$ accepts any sequence of views from $\mathcal{V}_v$.

Recall that we say that view $v_i \in seq$ is \emph{more recent} than view $v_j \in seq$ if $v_j \subset v_i$.
Additionally, we say that view $v_i \in seq$ is the \emph{most recent} view of $seq$ if $\not\exists v_j \in seq: v_i \subset v_j$.

Each event $e \in \mathcal{E}$ is a tuple $e = \langle v, seq = v_1 \to ... \to v_m \rangle$, such that $m \geq 1$ and it represents the fact that a sequence of views $seq$ is converged on to replace some view $v$.

If $(s, e, s') \in \mathcal{D}$, we say that event $e$ is \emph{enabled} in state $s$.
As we describe below, some event $e = \langle v, seq \rangle$ is enabled in some state $s$ if
$v \in s$.
Now, we define the transition set $\mathcal{D}$.
For some states $s$ and $s'$ and some event $e = \langle v, seq = v_1 \to ... \to v_m \rangle$, where $s \in \mathcal{S}$, $s' \in \mathcal{S}$ and $e \in \mathcal{E}$, $(s, e, s') \in \mathcal{D}$ if:
\begin{compactitem}
    \item if $v_1 \notin s$ and $m > 1$, then $s' = s \cup \{v_1\}$, $\alpha_{s'}(v_1) = \bot$, $\beta_{s'}(v_1) = \{seq \setminus \{v_1\}\}$, $\alpha_{s'}(v) = \alpha_{s}(v)$ and $\beta_{s'}(v) = \beta_s(v)$, for every $v \in s$; or
    \item if $v_1 \in s$ and $m > 1$, then $s' = s$, $\alpha_{s'}(v_1) = \alpha_{s}(v_1)$, $\beta_{s'}(v_1) = \beta_{s}(v_1) \cup \{seq \setminus \{v_1\}\}$, $\alpha_{s'}(v) = \alpha_{s}(v)$ and $\beta_{s'}(v) = \beta_s(v)$, for every $v \in s$, where $v \neq v_1$; or
    \item if $v_1 \notin s$ and $m = 1$, then $s' = s \cup \{v_1\}$, $\alpha_{s'}(v_1) = \top$, $\beta_{s'}(v_1) = \emptyset$, $\alpha_{s'}(v) = \alpha_{s}(v)$ and $\beta_{s'}(v) = \beta_s(v)$, for every $v \in s$; or
    \item if $v_1 \in s$ and $m = 1$, then $s' = s$, $\alpha_{s'}(v_1) = \top$, $\beta_{s'}(v_1) = \beta_{s}(v_1)$, $\alpha_{s'}(v) = \alpha_{s}(v)$ and $\beta_{s'}(v) = \beta_{s}(v)$, for every $v \in s$, where $v \neq v_1$.
\end{compactitem}

The system observes events without delay, i.e., as soon as some sequence of views $seq$ is converged on to replace some view $v$, system observes the event $e = \langle v, seq \rangle$.



\begin{definition} \label{definition:valid-view}
We say that view $v$ is \textbf{valid} if $v \in s$, where $s$ is the current state of the system.
\end{definition}





The assumptions regarding the size of the quorum and number of faulty processes in a view (see \Cref{sec:building-blocks}) apply only to valid views, i.e., only valid views need to satisfy the aforementioned assumptions.
Moreover, correct processes process \msg{propose}, \msg{converged} and \msg{install} messages only if they are associated with valid views.
Lastly, a correct process $p$ stores some message $m$ which message certificate is collected in some view $v$ once $p$ ``discovers'' that view $v$ is valid.


\begin{lemma} \label{lemma:sequence_of_views}
Suppose that the current state of the system is $s \in \mathcal{S}$.
Suppose that $seq$ is converged on to replace some view $v \in s$.
Then, $seq$ is a sequence of views, i.e., the following holds: $\forall v_i, v_j \in seq: (v_i \neq v_j) \implies (v_i \subset v_j \lor v_i \supset v_j)$.
\end{lemma}
\begin{proof}
Since $seq$ is converged on, a correct process $p \in v$ sent a \msg{propose} message for $seq$.
The value of $SEQ^v$ variable is sent inside of the \msg{propose} message.
Let us show that our algorithm satisfies the invariant that $SEQ^v$ is always a sequence of views. Initially $SEQ^v$ is empty.
We analyze the two possible ways for modifying the $SEQ^v$ variable of $p$:
\begin{compactitem}
    \item There are no conflicting views upon delivering a \msg{propose} message (\Cref{line:update_nonconflicting})\footnote{
        Recall that two different views $v$ and $v'$ conflict if $v \not\subset v' \land v \not\supset v'$.
    }:
    New value of $SEQ^v$ is a sequence of views because of the fact that there are no conflicting views and both value of $SEQ^v$ and sequence received in the \msg{propose} message (because of the verification at \Cref{line:only_more_updated_1}) are sequences of views.
    \item There are conflicting views upon delivering a \msg{propose} message (\Cref{line:conflicting_start} to \Cref{line:update_conflicting}):
    Suppose that $p$ has delivered a \msg{propose} message with $seq'$.
    Because of the verification at \Cref{line:only_more_updated_1}, $seq'$ is a sequence of views.
    Moreover, $LCSEQ^v$ variable of $p$ is either $SEQ^v$ or a subset of $SEQ^v$.
    After executing \Cref{line:update_conflicting}, the union of two most recent views results in a view more recent than the most recent view of $LCSEQ^v$ and $LCSEQ^v$ is a sequence of views (because $LCSEQ^v \subseteq SEQ^v$).
    Hence, new value of $SEQ^v$ variable is a sequence of views.
\end{compactitem}

\end{proof}

\Cref{lemma:sequence_of_views} proves that if $seq$ is converged on to replace some valid view $v$, then $seq$ is a sequence of views.
Moreover, since a set with a single view is a sequence of views and a correct process proposes $seq$ to replace some valid view at \Cref{line:propose_system_configuration} or at \Cref{line:propose_rest}, proposed $seq$ is a sequence of views.
Similarly, if $seq$ is accepted to replace some valid view, then $seq$ is a sequence of views.

\begin{lemma} \label{lemma:weak_accuracy}
Suppose that the current state of the system is $s \in \mathcal{S}$.
Suppose that sequences of views $seq_1$ and $seq_2$ are converged on to replace some view $v \in s$.
Then, either $seq_1 \subseteq seq_2$ or $seq_1 \supset seq_2$.
\end{lemma}
\begin{proof}
Suppose that some process $p \in v$ has collected a quorum of \msg{converged} messages for $seq_1$ to replace a view $v$.
Moreover, suppose that process $q \in v$ has collected a quorum of \msg{converged} messages for $seq_2$ to replace $v$.
This means that there is at least one correct process $k \in v$ that has sent a \msg{converged} message both for $seq_1$ and $seq_2$.
Suppose that $k$ first sent the \msg{converged} message for $seq_1$ and $seq_1 \neq seq_2$.

At the moment of sending a \msg{converged} message for $seq_1$ (\Cref{line:send_converged}), the following holds for $k$'s local variables: $LCSEQ^v = seq_1$ and $SEQ^v = seq_1$.
On the other hand, $SEQ^v = seq_2$ when $k$ sent a \msg{converged} message for $seq_2$.
Thus, we need to prove that $SEQ^v$ includes $seq_1$ from the moment $k$ sent a \msg{converged} message for $seq_1$.

There are two ways for modifying $SEQ^v$ local variable:
\begin{compactitem}
    \item Process $k$ receives a \msg{propose} message for some sequence of views $seq$ and there are no conflicting views in $SEQ^v$ and $seq$ (\Cref{line:update_nonconflicting}):
    Since $SEQ^v$ contains $seq_1$ and the simple union is performed, the resulting sequence of views contains $seq_1$.
    \item Process $k$ receives a \msg{propose} message for some sequence of views $seq$ and there are conflicting views in $SEQ^v$ and $seq$ (\Cref{line:conflicting_start} to \Cref{line:update_conflicting}):
    Since $LCSEQ^v$ contains $seq_1$ and $LCSEQ^v$ participates in the union, the resulting sequence of views contains $seq_1$.
\end{compactitem}
We proved that $SEQ^v$ always contains $seq_1$ after the \msg{converged} message for $seq_1$ is sent.
Consequently, the following holds: $seq_1 \subset seq_2$.

Using similar arguments, it is possible to prove the symmetrical execution: if process $k$ sent \msg{converged} message for $seq_2$ and then for $seq_1$, we have that $seq_1 \supset seq_2$.
Lastly, processes could have collected the \msg{converged} messages for the same sequence and then we have $seq_1 = seq_2$.
\end{proof}

\begin{lemma} \label{lemma:bigger}
Suppose that the current state of the system is $s \in \mathcal{S}$.
Moreover, suppose that a sequence of views $seq$ is converged on to replace some view $v \in s$.
Then, the following holds: $\forall \omega \in seq: v \subset \omega$.
\end{lemma}
\begin{proof}
The lemma holds because of the verifications at \cref{line:only_more_updated_1,line:check_propose_2}.
\end{proof}

\begin{property} \label{property:installable}
Suppose that the current state of the system is $s \in \mathcal{S}$.
The following holds for every view $v \in s$:
\begin{compactenum}
    \item If $v$ is installable in $s$ and $v$ is not the initial view of the system, then a sequence of views $seq = v_1 \to ... \to v_n \to v, (n \geq 0)$ is converged on to replace some view $v' \in s$, where $v'$ is installable in $s$, and views $v_1, ..., v_n \in s$ are not installable in $s$.
    Moreover, a sequence of views $seq_i = v_{i + 1} \to ... \to v_n \to v$ is converged on to replace view $v_i \in s$, for every $i \geq 1$ and $i < n$ and a sequence of views $seq_n = v$ is converged on to replace view $v_n \in s$; and
    \item If $v$ is not installable in $s$, then a sequence of views $seq = v_1' \to ... \to v_m' \to v \to v_1 \to ... \to v_n, (m \geq 0, n \geq 1)$ is converged on to replace some view $v' \in s$, where $v'$ is installable in $s$, and views $v_1', ..., v_m' \in s$ are not installable in $s$.
    Moreover, a sequence of views $seq_i = v_{i + 1}' \to ... \to v_m' \to v \to v_1 \to ... \to v_n$ is converged on to replace view $v_i' \in s$, for every $i \geq 1$ and $i < m$ and a sequence of views $seq_m = v \to v_1 \to ... \to v_n$ is converged on to replace view $v_m' \in s$.
\end{compactenum}
\end{property}

From now on, we assume that the \Cref{property:installable} holds for the current state of the system.
We define properties of the current state of the system given that \Cref{property:installable} holds.
Then, we prove that \Cref{property:installable} also holds for a new state of the system obtained from the current one.

\begin{definition} \label{definition:leads_to}
Suppose that the current state of the system is $s \in \mathcal{S}$.
Consider views $v \in s$ and $v' \in s$, where both $v$ and $v'$ are installable in $s$.
Suppose that a sequence of views $seq = v_1 \to ... \to v_n \to v', (n \geq 0)$ is converged on to replace $v$, where views $v_1, ..., v_n \in s$ are not installable in $s$ (the first claim of \Cref{property:installable}).
Moreover, suppose that a sequence of views $seq_i = v_{i + 1} \to ... \to v_n \to v'$ is converged on to replace view $v_i \in s$, for every $i \geq 1$ and $i < n$ and a sequence of views $seq_n = v'$ is converged on to replace view $v_n \in s$ (the first claim of \Cref{property:installable}).
We say that $v$ \textbf{leads to $v'$ in $s$}.
\end{definition}

\begin{lemma} \label{lemma:installable_views_compare}
Suppose that the current state of the system is $s \in \mathcal{S}$.
Moreover, suppose that view $v \in s$ leads to view $v' \in s$ in $s$.
Then, $v \subset v'$.
\end{lemma}
\begin{proof}
The lemma follows from \Cref{definition:leads_to} and \Cref{lemma:bigger}.
\end{proof}

\begin{lemma} \label{lemma:installable_from_one}
Suppose that the current state of the system is $s \in \mathcal{S}$.
Consider a view $v' \in s$, where $v'$ is installable in $s$ and $v'$ is different from the initial view of the system.
Then, some view $v \in s$, where $v$ is installable in $s$, leads to $v'$ in $s$.
\end{lemma}
\begin{proof}
The lemma follows directly from \Cref{property:installable} and \Cref{definition:leads_to}.
\end{proof}

\begin{lemma} \label{lemma:only_one_view_leads_to}
Suppose that the current state of the system is $s \in \mathcal{S}$.
Consider a view $v \in s$, where $v$ is installable in $s$.
If $v$ leads to $v' \in s$ in $s$ and $v$ leads to $v'' \in s$ in $s$, then $v' = v''$.
\end{lemma}
\begin{proof}
Because of the \Cref{definition:leads_to}, sequences of views $seq' = ... \to v'$ and $seq'' = ... \to v''$ are converged on to replace $v$.
\Cref{lemma:weak_accuracy} ensures that the following holds: $seq' \subseteq seq''$ or $seq' \supset seq''$.
Let us analyze the two possible cases:
\begin{compactitem}
    \item $seq' = seq''$: In this case, $v' = v''$, and the lemma holds.
    \item $seq' \neq seq''$: Without loss of generality, assume that $seq' \subset seq''$. 
    Therefore, $v' \in seq''$ and $v' \subseteq v''$ (\Cref{lemma:sequence_of_views}).
    If $seq' \subset seq''$, then \Cref{definition:leads_to}
    would not be satisfied for $v''$ since $v'$ is installable in $s$.
    It follows that $v' = v''$, and the lemma holds.
\end{compactitem}
\end{proof}

\begin{lemma} \label{lemma:from_only_one}
Suppose that the current state of the system is $s \in \mathcal{S}$.
Consider a view $v \in s$, where $v$ is installable in $s$.
If $v' \in s$ leads to $v$ in $s$ and $v'' \in s$ leads to $v$ in $s$, then $v' = v''$.
\end{lemma}
\begin{proof}
The lemma follows directly from \cref{lemma:installable_from_one,lemma:only_one_view_leads_to}.
\end{proof}

\begin{definition} \label{definition:auxiliary_view}
Suppose that the current state of the system is $s \in \mathcal{S}$.
Consider views $v \in s$ and $v' \in s$, such that $v$ is installable in $s$ and $v'$ is not installable in $s$.
Suppose that a sequence of views $seq = v_1 \to ... \to v_m \to v' \to v_1' \to ... \to v_n', (m \geq 0, n \geq 1)$ is converged on to replace $v$, where views $v_1, ..., v_m \in s$ are not installable in $s$ (the second claim of \Cref{property:installable}).
Moreover, a sequence of views $seq_i = v_{i + 1} \to ... \to v_m \to v' \to v_1' \to ... \to v_n'$ is converged on to replace view $v_i \in s$, for every $i \geq 1$ and $i < m$ and a sequence of views $seq_m = v' \to v_1' \to ... \to v_n'$ is converged on to replace view $v_m \in s$ (the second claim of \Cref{property:installable}).
We say that $v'$ is an \textbf{auxiliary view for $v$ in $s$}.
\end{definition}

\begin{lemma} \label{lemma:auxiliary_views_compare_1}
Suppose that the current state of the system is $s \in \mathcal{S}$.
Consider a view $v \in s$, where $v$ is an auxiliary view for $v' \in s$ in $s$.
Then, $v' \subset v$.
\end{lemma}
\begin{proof}
This follows from \Cref{definition:auxiliary_view} and \Cref{lemma:bigger}.
\end{proof}

\begin{lemma} \label{lemma:auxiliary_views_compare_2}
Suppose that the current state of the system is $s \in \mathcal{S}$.
Consider a view $v \in s$, where $v$ is an auxiliary view for $v' \in s$ in $s$.
Moreover, suppose that $v'$ leads to $v'' \in s$ in $s$.
Then, $v \subset v''$.
\end{lemma}
\begin{proof}
A sequence of views $seq'' = v_1 \to ... \to v_n \to v'', (n \geq 0)$ is converged on to replace $v'$ (\Cref{definition:leads_to}).
Moreover, a sequence of views $seq' = v_1' \to ... \to v_m' \to v \to v_1'' \to ... \to v_j'', (m \geq 0, j \geq 1)$ is converged on to replace $v'$ (\Cref{definition:auxiliary_view}).
Lastly, views $v_1, ..., v_n \in s$ and views $v_1', ..., v_m' \in s$ are not installable in $s$.

\Cref{lemma:weak_accuracy} shows that either $seq' \subseteq seq''$ or $seq' \supset seq''$.
Let us investigate two possible scenarios:
\begin{compactitem}
    \item $seq' = seq''$:
    We deduce that $v \subset v''$ (because $v \neq v''$) since \Cref{lemma:sequence_of_views} holds.
    \item $seq' \neq seq''$: We analyze two possible cases:
    \begin{compactitem}
        \item $seq' \subset seq''$:
        In this case, we can conclude that $v \in seq''$.
        Following \Cref{lemma:sequence_of_views} and the fact that $v \neq v''$, we conclude that $v \subset v''$.
        \item $seq' \supset seq''$:
        In this case, we can conclude that $v'' \in seq'$.
        If $v''$ is more recent than $v$ in $seq'$, the lemma holds.
        If $v'' = v$, then we come to the contradiction since $v$ is not installable in $s$.
        Lastly, if $v$ is more recent than $v''$, then $v$ is not an auxiliary view for $v'$ in $s$ (since $v''$ is an installable view in $s$ and $v'' \supset v'$ by \Cref{lemma:bigger}).
        Hence, we reach the contradiction even in this case.
    \end{compactitem}
\end{compactitem}

The lemma holds since in any case the following is satisfied: $v \subset v''$.
\end{proof}

\begin{lemma} \label{lemma:auxiliary_from_one}
Suppose that the current state of the system is $s \in \mathcal{S}$.
Consider a view $v' \in s$, where $v'$ is not installable in $s$.
Then, $v'$ is an auxiliary view for some view $v \in s$ in $s$.
\end{lemma}
\begin{proof}
The lemma follows directly from \Cref{property:installable} and \Cref{definition:auxiliary_view}.
\end{proof}

\begin{lemma} \label{lemma:one_auxiliary_view}
Suppose that the current state of the system is $s \in \mathcal{S}$.
Consider a view $v \in s$, where $v$ is an auxiliary view for $v' \in s$ in $s$ and $v'' \in s$ in $s$.
Then, $v' = v''$.
\end{lemma}
\begin{proof}
The lemma follows directly from \cref{lemma:installable_from_one,lemma:only_one_view_leads_to} and \Cref{definition:auxiliary_view}.
\end{proof}

\begin{property} \label{property:accepted_if_converged}
Suppose that the current state of the system is $s \in \mathcal{S}$.
Consider a view $v \in s$, where $v$ is an auxiliary view for $v' \in s$ in $s$.
If a sequence of views $seq = v_1 \to ... \to v_n, (n \geq 1)$ is accepted to replace $v$ in $s$, then a sequence of views $seq' = v_1' \to ... \to v_m' \to v \to v_1 \to ... \to v_n, (m \geq 0)$ is converged on to replace $v'$, where views $v_1', ..., v_m' \in s$ are not installable in $s$. Moreover, a sequence of views $seq_i = v_{i + 1}' \to ... \to v_m' \to v \to v_1 \to ... \to v_n$ is converged on to replace view $v_i' \in s$, for every $i \geq 1$ and $i < m$ and a sequence of views $seq_m = v \to v_1 \to ... \to v_n$ is converged on to replace view $v_m' \in s$.
\end{property}

When we say that \Cref{property:accepted_if_converged} holds for the current state $s$ of the system, we assume that it holds for every view $v \in s$ that is not installable in $s$.
Again, we assume that \Cref{property:accepted_if_converged} holds for the current state of the system (along with \Cref{property:installable}).
Then, we prove that \Cref{property:accepted_if_converged} also holds for a new state of the system obtained from the current one.

\begin{lemma} \label{lemma:auxiliary_comparable_accepted}
Suppose that the current state of the system is $s \in \mathcal{S}$.
Consider a view $v \in s$, where $v$ is not installable in $s$.
Suppose that sequences of views $seq_1$ and $seq_2 \neq seq_1$ are accepted to replace $v$ in $s$.
Then, either $seq_1 \subset seq_2$ or $seq_1 \supset seq_2$.
\end{lemma}
\begin{proof}
The lemma follows directly from \Cref{property:accepted_if_converged} and \cref{lemma:weak_accuracy,lemma:one_auxiliary_view}.
\end{proof}

\begin{property} \label{property:format_installable}
Suppose that the current state of the system is $s \in \mathcal{S}$.
Consider a view $v \in s$ installable in $s$.
Then, for any $seq_1 \in \beta_s(v)$ and $seq_2 \in \beta_s(v)$, if $seq_1 \neq seq_2$, then either $seq_1 \subset seq_2$ or $seq_1 \supset seq_2$.
\end{property}

When we say that \Cref{property:format_installable} holds for the current state $s$ of the system, we assume that it holds for every view $v \in s$ installable in $s$.
Again, we assume that \Cref{property:format_installable} holds for the current state of the system (along with \cref{property:installable,property:accepted_if_converged}).
Then, we prove that \Cref{property:format_installable} also holds for a new state of the system obtained from the current one.

\begin{lemma} \label{lemma:comparable_format}
Suppose that the current state of the system is $s \in \mathcal{S}$.
Consider a view $v \in s$ and a correct process $p \in v$.
For any sequences of views $seq_1$ and $seq_2$ , where $p$ accepts both $seq_1$ and $seq_2$ to replace $v$ and $p$ does not accept any sequence to replace $v$, if $seq_1 \neq seq_2$, then either $seq_1 \subset seq_2$ or $seq_1 \supset seq_2$.
\end{lemma}
\begin{proof}
The lemma follows directly from \Cref{lemma:auxiliary_comparable_accepted} and \Cref{property:format_installable}.
\end{proof}

\begin{lemma} \label{lemma:converged_then_accepted}
Suppose that the current state of the system is $s \in \mathcal{S}$.
Consider a view $v \in s$.
Suppose that a sequence of views $seq$ is converged on to replace $v$.
Then, $seq$ is accepted to replace $v$ in $s$.
\end{lemma}
\begin{proof}
In a case where $v$ is the initial view of the system, $seq$ is accepted by default.

Suppose now that $seq$ is converged on to replace $v$ and $v$ is not the initial view of the system.
This implies that at least $v.q$ processes have sent a \msg{propose} message for $seq$.
There is a correct process $p \in v$ that has sent the aforementioned \msg{propose} message.
Lastly, there are three places in the algorithm where $p$ could have sent the \msg{propose} message:
\begin{compactenum}
    \item \Cref{line:propose_system_configuration}:
    In this case, $p$ has executed \Cref{line:install_complete}.
    This implies that $\emptyset \in FORMAT^v$ at process $p$.
    Hence, $seq$ is accepted to replace $v$.
    
    \item \Cref{line:propose_rest}:
    In this case, $p$ has executed \Cref{line:set_format} for $seq$.
    Hence, $seq$ is accepted to replace $v$ in this case.
    
    \item \Cref{line:propose_second_case}:
    Since $v$ is not the initial view of the system, process $p$ received an \msg{install} message where $\omega$ parameter is equal to $v$.
    Now, we separate two cases with respect to \Cref{line:set_format}:
    \begin{compactenum}
        \item $\emptyset \in \mathit{FORMAT^v}$ at process $p$:
        In this case, $p$ accepts any sequence of views to replace $v$.
        Thus, $seq$ is accepted to replace $v$.
        
        \item $\emptyset \notin \mathit{FORMAT^v}$ at process $p$:
        Suppose that $p$ received a \msg{propose} message with $seq'$ as an argument and that $\mathit{SEQ^v}$ variable of $p$ is equal to $seq'' \neq seq'$.
        We conclude that $p$ accepts $seq'$ (otherwise, $p$ would ignore the \msg{propose} message) and $seq''$.
        By \Cref{lemma:comparable_format}, either $seq' \subset seq''$ or $seq' \supset seq''$.
        
        Since $p$ executes \Cref{line:propose_second_case} upon receiving the \msg{propose} message with $seq'$, we conclude that $seq' \supset seq''$.
        Moreover, after executing \Cref{line:update_nonconflicting}, $\mathit{SEQ^v}$ is equal to $seq'$.
        Given the fact that this \msg{propose} message sent by $p$ leads to a fact that $seq$ is converged on to replace $v$, we conclude that $seq = seq'$.
        Therefore, $p$ accepts $seq$ to replace $v$ even in this case.
    \end{compactenum}
\end{compactenum}
\end{proof}

It is easy to see that \cref{lemma:auxiliary_comparable_accepted,lemma:converged_then_accepted} hold for any view $v$, where \cref{property:accepted_if_converged,property:format_installable} are satisfied for $v$.

\begin{lemma} \label{lemma:unique_installable_views}
Suppose that the current state of the system is $s \in \mathcal{S}$.
Then, views that are installable in $s$ form a sequence of views.
\end{lemma}
\begin{proof}
Suppose that a view $\omega_1$ is an installable view.
\Cref{lemma:only_one_view_leads_to} shows that $\omega_1$ leads to at most one view.
Moreover, \cref{lemma:installable_from_one,lemma:from_only_one} shows that there is exactly one installable view $\omega_0$ that leads to $\omega_1$ (if $\omega_1$ is not the initial view of the system).
Lastly, if view $\omega_0$ leads to $\omega_1$, then $\omega_1 \supset \omega_0$ (\Cref{lemma:installable_views_compare}).
Therefore, installable views form a sequence of views starting with the initial view of the system.
\end{proof}

\begin{lemma} \label{lemma:valid_comparable}
Suppose that the current state of the system is $s \in \mathcal{S}$.
Consider a view $v \in s$ and a view $v' \in s$, where $v \neq v'$.
Then, either $v \subset v'$ or $v \supset v'$.
\end{lemma}
\begin{proof}
We introduce four possible scenarios:
\begin{compactitem}
    \item Both $v$ and $v'$ are installable in $s$:
    The lemma follows directly from the definition of a sequence of views and \Cref{lemma:unique_installable_views}.
    \item View $v$ is not installable in $s$, whereas view $v'$ is installable in $s$:
    View $v$ is an auxiliary view for some view $v_1 \in s$ in $s$ (\cref{lemma:auxiliary_from_one,lemma:one_auxiliary_view}).
    
    If $v_1 = v'$, then $v' \subset v$ (\Cref{lemma:auxiliary_views_compare_1}).
    If $v_1 \subset v'$, then $v \subset v'$ (by \cref{lemma:auxiliary_views_compare_2,lemma:unique_installable_views} and the definition of a sequence of views).
    Lastly, if $v_1 \supset v'$, then $v \supset v'$ (by \cref{lemma:unique_installable_views,lemma:auxiliary_views_compare_1} and the definition of a sequence of views).
    \item View $v$ is installable in $s$, whereas view $v'$ is not installable in $s$: 
    This case reduces to the previous one.
    \item Both $v$ and $v'$ are not installable in $s$:
    Suppose that $v$ is an auxiliary view for some view $v_1 \in s$ in $s$ and that $v'$ is an auxiliary view for some view $v_2 \in s$ in $s$ (\cref{lemma:auxiliary_from_one,lemma:one_auxiliary_view}).
    If $v_1 = v_2$, then the lemma holds because of \cref{lemma:sequence_of_views,lemma:weak_accuracy} and \Cref{definition:auxiliary_view}.
    If $v_1 \subset v_2$, then $v \subset v'$ (by \cref{lemma:unique_installable_views,lemma:auxiliary_views_compare_2,lemma:auxiliary_views_compare_1} and the definition of a sequence of views).
    Lastly, if $v_1 \supset v_2$, then $v \supset v'$ (by \cref{lemma:unique_installable_views,lemma:auxiliary_views_compare_2,lemma:auxiliary_views_compare_1} and the definition of a sequence of views).
\end{compactitem}

The lemma holds since $v \subset v'$ or $v \supset v'$ is satisfied in each of four possible cases.
\end{proof}

\begin{lemma} \label{lemma:unique_valid_views}
Suppose that the current state of the system is $s \in \mathcal{S}$.
Then, $s$ is a sequence of views.
\end{lemma}
\begin{proof}
This follows directly from \Cref{lemma:valid_comparable} and the definition of a sequence of views.
\end{proof}

\begin{lemma} \label{lemma:again_sequence}
Suppose that the current state of the system is $s \in \mathcal{S}$.
Moreover, suppose that some sequence of views $seq = v \to v_1 \to ... \to v_n, (n \geq 0)$ is converged on to replace some view $v' \in s$.
For every view $v'' \in s$, either $v \subseteq v''$ or $v \supset v''$.
\end{lemma}
\begin{proof}
If $v \in s$, the lemma holds because of \Cref{lemma:unique_valid_views}.
Therefore, we consider a case where $v \notin s$.

By \Cref{lemma:bigger}, we conclude that $v \supset v'$.
For every view $v'' \subset v'$, it follows that $v \supset v''$.

1) Suppose that $v'$ is installable in $s$.
Consider some view $v''$, where $v''$ is an auxiliary view for $v'$ in $s$.
Hence, a sequence of views $seq_{v''} = ... \to v'' \to ...$ is converged on to replace $v'$ (by \Cref{definition:auxiliary_view}).
By \Cref{lemma:weak_accuracy}, there are three possible cases:
\begin{compactitem}
    \item $seq = seq_{v''}$: 
    This is not possible since $v \notin s$.
    \item $seq \subset seq_{v''}$: 
    We conclude that $v \in seq_{v''}$. Because of \Cref{lemma:sequence_of_views} and $v \notin s$, $v \supset v''$.
    \item $seq \supset seq_{v''}$: Because of \Cref{lemma:sequence_of_views} and $v \notin s$, we conclude that $v \subset v''$.
\end{compactitem}

Suppose now that $v'$ leads to some view $v_{lead} \in s$ in $s$.
Hence, a sequence of views $seq_{v_{lead}} = ... \to v_{lead}$ is converged on to replace $v'$ (by \Cref{definition:leads_to}).
By \Cref{lemma:weak_accuracy}, there are three possible cases:
\begin{compactitem}
    \item $seq = seq_{v_{lead}}$: 
    This is not possible since $v \notin s$.
    \item $seq \subset seq_{v_{lead}}$: 
    This is not possible since $v \notin s$.
    \item $seq \supset seq_{v_{lead}}$: Because of \Cref{lemma:sequence_of_views} and $v \notin s$, we conclude that $v \subset v_{lead}$.
\end{compactitem}
Hence, for every view $v'' \supset v_{lead}$, where $v'' \in s$, it follows that $v \subset v''$.
Moreover, for every view $v'' \subset v'$, where $v'' \in s$, it follows that $v \supset v''$.
And we proved that either $v \subset v''$ or $v \supset v''$, where $v'' \in s$ is an auxiliary view for $v'$ in $s$.

2) Suppose that $v'$ is not installable in $s$.
By \Cref{property:accepted_if_converged} and \Cref{lemma:converged_then_accepted}, some sequence of views $seq_{v''} = ... \to v' \to v \to v_1 \to ... \to v_n$ is converged on to replace some view $v''$, where $v''$ is installable in $s$.
Following the similar argument as in the previous case, we conclude that the claim of the lemma holds even in this case.
\end{proof}

We prove now that \cref{property:installable,property:accepted_if_converged,property:format_installable} hold for $s'$, where $s'$ is obtained from $s$ such that \cref{property:installable,property:accepted_if_converged,property:format_installable} hold for $s$.

\begin{lemma} \label{lemma:property_1_part_1}
Consider some state $s \in \mathcal{S}$, such that \cref{property:installable,property:accepted_if_converged,property:format_installable} hold for $s$.
Moreover, consider some event $e = \langle v_{con}, seq = v \to v_1 \to ... \to v_m \rangle$, where $e \in \mathcal{E}$, $m \geq 0$ and $e$ is enabled in $s$.
Then, a sequence of views $seq = v_1' \to ... \to v_k' \to v \to v_1 \to ... \to v_m, (k \geq 0, m \geq 0)$ is converged on to replace some view $v' \in s$, where $v'$ is installable in $s$, and views $v_1', ..., v_k' \in s$ are not installable in $s$.
Moreover, a sequence of views $seq_i = v_{i + 1}' \to ... \to v_k' \to v \to v_1 \to ... \to v_m$ is converged on to replace view $v_i' \in s$, for every $i \geq 1$ and $i < k$ and a sequence of views $seq_k = v \to v_1 \to ... \to v_m$ is converged on to replace view $v_k' \in s$.
\end{lemma}
\begin{proof}
A sequence of views $seq = v \to v_1 \to ... \to v_m, (m \geq 0)$ is converged on to replace some view $v_{con} \in s$.
If $v_{con}$ is installable in $s$, then the lemma holds.

Suppose that $v_{con}$ is not installable in $s$.
Hence, $v_{con}$ is an auxiliary view for some view $v'' \in s$ in $s$.
By \Cref{lemma:converged_then_accepted}, $seq$ is accepted to replace $v_{con}$ in $s$.
Moreover, \Cref{property:accepted_if_converged} shows that a sequence of views $seq' = v_1' \to ... \to v_k' \to v_{con} \to v \to v_1 \to ... \to v_m, (k \geq 0)$ is converged on to replace $v''$ and views $v_1', ..., v_k' \in s$ are not installable in $s$.
Besides that, a sequence of views $seq_i = v_{i + 1}' \to ... \to v_k' \to v_{con} \to v \to v_1 \to ... \to v_m$ is converged on to replace view $v_i' \in s$, for every $i \geq 1$ and $i < k$ and a sequence of views $seq_k = v_{con} \to v \to v_1 \to ... \to v_m$ is converged on to replace view $v_k' \in s$.
Hence, the lemma holds.
\end{proof}

\begin{lemma} \label{lemma:property_1_part_2}
Suppose that the current state of the system is $s \in \mathcal{S}$.
Moreover, consider some event $e = \langle v_{con}, seq = v \to v_1 \to ... \to v_m \rangle$, where $e \in \mathcal{E}$, $m \geq 0$, $e$ is enabled in $s$ and the system observes $e$.
Suppose that $(s, e, s') \in \mathcal{D}$.
Then, \cref{property:installable,property:accepted_if_converged,property:format_installable} hold for $s'$.
\end{lemma}
\begin{proof}
We first show that $s'$ is a sequence of views. 
Notice that $seq = v \to v_1 \to ... \to v_m, (m \geq 0)$ is converged on to replace some view $v_{con} \in s$.
We conclude that a sequence of views $seq' = v_1' \to ... \to v_n' \to v \to v_1 \to ... \to v_m, (n \geq 0, m \geq 0)$ is converged on to replace some view $v' \in s$, where $v'$ is installable in $s$, and views $v_1', ..., v_n' \in s$ are not installable in $s$.
Moreover, a sequence of views $seq_i = v_{i + 1}' \to ... \to v_n' \to v \to v_1 \to ... \to v_m$ is converged on to replace view $v_i' \in s$, for every $i \geq 1$ and $i < n$ and a sequence of views $seq_n = v \to v_1 \to ... \to v_m$ is converged on to replace view $v_n' \in s$ (by \Cref{lemma:property_1_part_1}).
Hence, \Cref{property:installable} is satisfied for $v$ in $s'$.
Following \Cref{lemma:again_sequence}, we conclude that $s'$ is a sequence of views.

We divide the rest of the proof into four parts.

1) Suppose that $m = 0$ and $v \notin s$:
Since $m = 0$, we conclude that $v \in s'$ is installable in $s'$.

We argue that $v$ is the most recent view that is installable in $s'$.
Suppose that this is not the case.
Hence, $v' \in s$ leads to some installable view $v'' \in s$ in $s$.
Therefore, a sequence of views $seq'' = v_1''' \to ... \to v_p''' \to v'', (p \geq 0)$ is converged on to replace $v'$ (by \Cref{definition:leads_to}).
By \Cref{lemma:weak_accuracy}, the following holds: either $seq' \subseteq seq''$ or $seq' \supset seq''$.
Let us examine all possible cases:
\begin{compactitem}
    \item $seq' = seq''$: We conclude that $v = v''$, which means that $v \in s$. This represents the contradiction with the fact that $v \notin s$.
    \item $seq' \subset seq''$: We conclude that $v \in seq''$, which means that $v \in s$. This represents the contradiction with the fact that $v \notin s$.
    \item $seq' \supset seq''$: We conclude that $v'' \in seq$. Hence, this is a contradiction with the fact that views $v_1', ..., v_n' \in s$  are not installable in $s$.
\end{compactitem}
Hence, $v$ is the most recent installable view in $s'$.
Similarly, we prove that there is no view $v'' \in s'$, such that $v''$ is an auxiliary view for $v$ in $s'$.
Since the previous two claims hold and the fact that $\beta_{s'}(v) = \emptyset$, we conclude that \cref{property:installable,property:accepted_if_converged,property:format_installable} hold for $s'$.

\begin{figure}[htp]
    \centering
    \includegraphics[width=.8\linewidth]{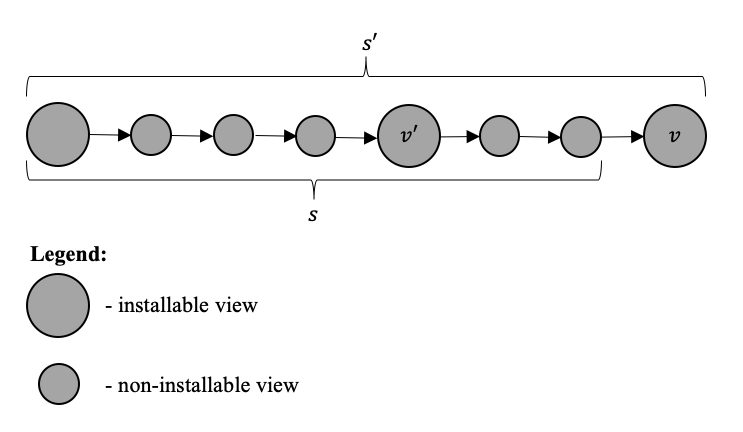}
    \vspace*{-7mm}
    \caption{Picture for the part 1 of \Cref{lemma:property_1_part_2}.}
    \label{fig:part1}
\end{figure}

2) Suppose that $m = 0$ and $v \in s$:
If $v$ is installable in $s$, then \cref{property:installable,property:accepted_if_converged,property:format_installable} trivially hold for $s'$.
Hence, we consider a case where $v$ is not installable in $s$.

We conclude that $v$ is an auxiliary view for $v'$ in $s$.
Suppose that $v'$ leads to $v'' \in s$ in $s$.
For every view $v_{ok} \in s$, such that $v_{ok} \subset v$ or $v_{ok} \supset v''$, claims stated by \cref{property:installable,property:accepted_if_converged,property:format_installable} hold in $s'$.

Now, we prove that $v$ leads to $v''$ in $s'$ and that the claim of \Cref{property:installable} is satisfied for $v''$ in $s'$.
Since $v'$ leads to $v''$ in $s$, a sequence of views $seq_{v''} = ... \to v''$ is converged on to replace $v'$ (by \Cref{definition:leads_to}).
By \Cref{lemma:weak_accuracy}, let us examine first two possible cases:
\begin{compactitem}
    \item $seq' = seq_{v''}$: We conclude that $v = v''$, which is a contradiction with the fact that $v$ is not installable in $s$.
    \item $seq' \supset seq_{v''}$: We conclude that $v'' \in seq'$, which is a contradiction with the fact that views $v_1', ..., v_n' \in s$  are not installable in $s$.
\end{compactitem}
Therefore, we conclude that $seq' \subset seq_{v''}$, that $v$ leads to $v''$ in $s'$ and that the first claim of \Cref{property:installable} is satisfied for $v''$ in $s'$.

Consider now a view $v'' \in s$, such that $v'' \supset v$ and $v''$ is auxiliary view for $v'$ in $s$.
We show that the second claim of \Cref{property:installable} is satisfied for $v''$ in $s'$.
Since $v''$ is an auxiliary view for $v'$ in $s$, a sequence of views $seq_{v''} = ... \to v'' \to ...$ is converged on to replace $v'$.
By \Cref{lemma:weak_accuracy}, let us examine first two possible cases:
\begin{compactitem}
    \item $seq' = seq_{v''}$: We conclude that $v \supset v''$, which is a contradiction with the fact that $v'' \supset v$.
    \item $seq' \supset seq_{v''}$: We conclude that $v'' \in seq'$, which is a contradiction with the fact that $v'' \supset v$.
\end{compactitem}
Therefore, we conclude that $v \in seq_{v''}$ and that the second claim of \Cref{property:installable} is satisfied for $v''$ in $s'$.

Consider now a view $v'' \in s$, such that $v'' \supset v$ and $v''$ is auxiliary view for $v'$ in $s$.
We show that the claim of \Cref{property:accepted_if_converged} is satisfied for $v''$ in $s'$.
Suppose that a sequence of views $seq_{v''}$ is converged on to replace $v'$ (\Cref{property:accepted_if_converged}).
Similarly to the previous arguments, the following holds: $v \in seq_{v''}$.
Hence, the claim of \Cref{property:accepted_if_converged} is satisfied for $v''$ in $s'$.

Lastly, \Cref{property:format_installable} holds for $v$ in $s'$ since \Cref{property:accepted_if_converged} holds for $v$ in $s$ and \Cref{lemma:weak_accuracy}.

In this case, \cref{property:installable,property:accepted_if_converged,property:format_installable} hold for $s'$.

\begin{figure}[htp]
    \centering
    \includegraphics[width=.8\linewidth]{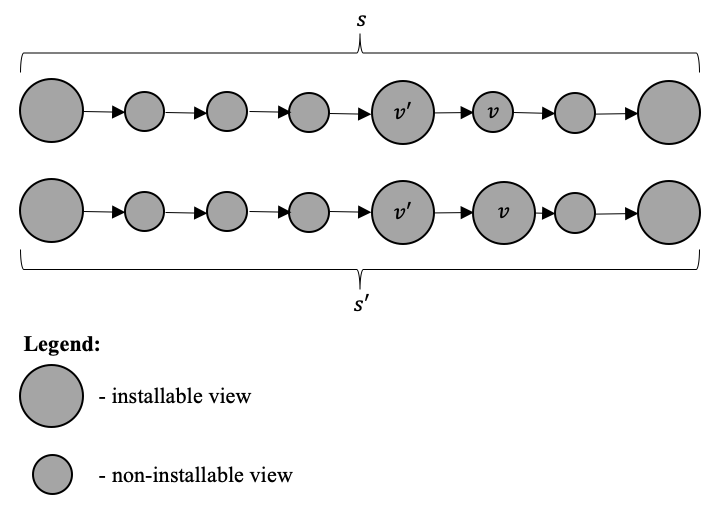}
    \vspace*{-4mm}
    \caption{Picture for the part 2 of \Cref{lemma:property_1_part_2}.}
    \label{fig:part2}
\end{figure}

3) Suppose that $m > 0$ and $v \notin s$:
Since $m > 0$, we conclude that $v \in s'$ is not installable in $s'$.
Because of \Cref{lemma:property_1_part_1} and the way $s'$ is obtained, we conclude that \Cref{property:installable} holds for $s'$.

The only statement left to prove is that \Cref{property:accepted_if_converged} holds for $s'$.
It is easy to see that the claim of \Cref{property:accepted_if_converged} holds for every view $v'' \in s$, because the set of sequences of views accepted to replace $v''$ is the same as in $s$.
Hence, we need to prove that the claim of \Cref{property:accepted_if_converged} holds for $v$.

We conclude that $v$ is an auxiliary view for $v'$ in $s'$.
Now, sequence of views $seq_{v} = v_1 \to ... \to v_m$ is accepted to replace $v$.
Since $v_{con}$ could be equal to $v'$ or \Cref{property:accepted_if_converged} holds for $v_{con}$ in $s$ and \Cref{lemma:converged_then_accepted}, we conclude that a sequence of views $seq_{v'} \supset seq_{v}$ is converged on to replace $v'$.
Hence, \Cref{property:accepted_if_converged} holds for $v$.

\Cref{property:installable,property:accepted_if_converged,property:format_installable} hold for $s'$.

\begin{figure}[htp]
    \centering
    \includegraphics[width=.8\linewidth]{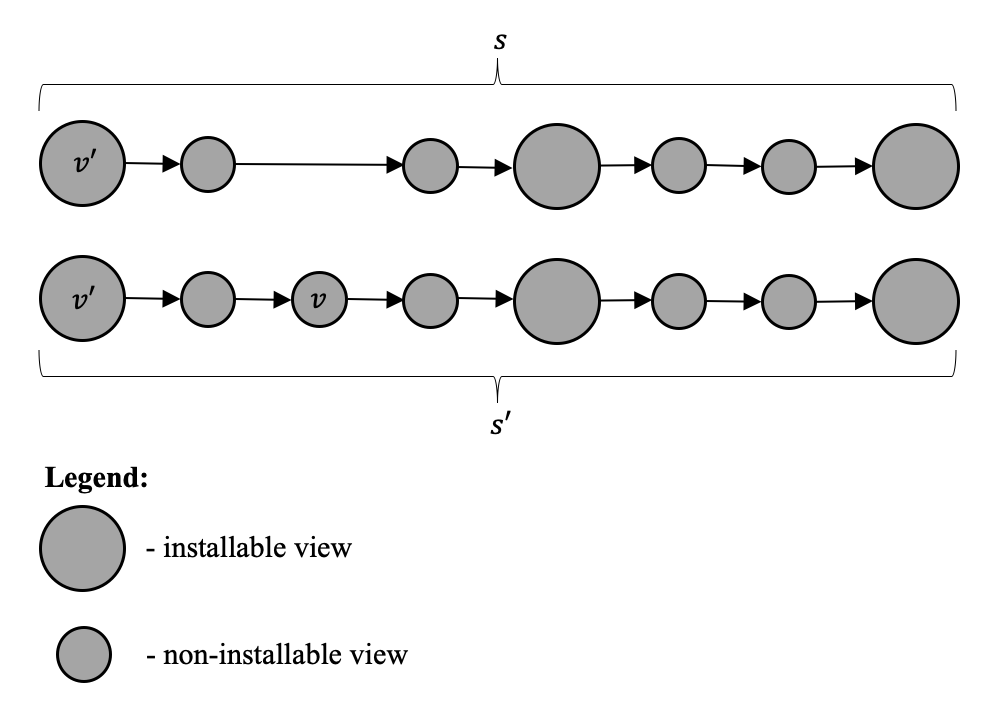}
    \vspace*{-6mm}
    \caption{Picture for the part 3 of \Cref{lemma:property_1_part_2}.}
    \label{fig:part3}
\end{figure}

4) Suppose that $m > 0$ and $v \in s$:
If $v$ is installable in $s$, then it is installable in $s'$.
In this case, \cref{property:installable,property:accepted_if_converged} hold for $s'$.
Moreover, since $seq'$ is converged on to replace $v'$ and \Cref{lemma:weak_accuracy}, we conclude that \Cref{property:format_installable} holds for $v$ (and consequently for $s'$).

Suppose that $v$ is not installable in $s$.
We conclude that $v \in s'$ is not installable in $s'$ and is an auxiliary view for $v' \in s'$ in $s'$.
Trivially, we conclude that \Cref{property:installable} holds for $s'$.
Moreover, we conclude that \Cref{property:accepted_if_converged} holds for every view $v'' \in s'$, where $v'' \neq v$.
Now, we need to prove that \Cref{property:accepted_if_converged} also holds for $v$.

A sequence of views $seq_{v} = v_1 \to ... \to v_m$ is also accepted to replace $v$, since a sequence of views $seq = v \to v_1 \to ... \to v_m$ is converged on to replace some view $v_{con} \in s'$.
Since $v_{con}$ could be equal to $v'$ or \Cref{property:accepted_if_converged} holds for $v_{con}$ and \Cref{lemma:converged_then_accepted}, we conclude that a sequence of views $seq_{v'} \supset seq_{v}$ is converged on to replace $v'$.
Hence, \Cref{property:accepted_if_converged} holds for $v$.

\Cref{property:installable,property:accepted_if_converged,property:format_installable} hold for $s'$ in this case, which concludes the lemma.
\end{proof}

As we have stated, at the beginning of an execution of \name, the state of the system $s_i$ includes only $v_0$, where $v_0$ is the initial view of the system.
Therefore, \cref{property:installable,property:accepted_if_converged,property:format_installable} trivially hold for $s_i$.
\Cref{lemma:property_1_part_2} proves that if $s \in \mathcal{S}$, for which \cref{property:installable,property:accepted_if_converged,property:format_installable} hold, evolves into $s' \in \mathcal{S}$, then \cref{property:installable,property:accepted_if_converged,property:format_installable} hold for $s'$.

So far, we proved the ``safety'' of \name.
Namely, we proved (\Cref{lemma:valid_comparable}) that all views in the current state of the system $s$ are comparable, i.e., $\forall v_1, v_2 \in s: (v_1 \neq v_2) \implies (v_1 \subset v_2 \lor v_1 \supset v_2)$.
Moreover, we showed that installable views in $s$ form a sequence of views (\Cref{lemma:unique_installable_views}).

\input{sections/appendix/liveness}

%% file: sections/appendix/liveness.tex
Now, we prove the ``liveness'' of \name.

We start by introducing the Byzantine Reliable Broadcast (\brb) used to disseminate \msg{install} and \msg{state-update} messages (\cref{line:send_reliable,line:install_seq_delivery,line:send_state_update}).
\brb{} exposes an interface with two primitives:
\begin{compactitem}
    \item \textit{R-multicast}$(\Psi, m)$: allows a process to send an \msg{install} or \msg{state-update} message $m$ to all processes in $\Psi$.
    \item \textit{R-delivery}$(\Psi, m)$: this callback triggers at process to handle the reception of an \msg{install} or \msg{state-update} message $m$.
\end{compactitem}
\brb{} satisfies following properties:
\begin{compactitem}
    \item \textbf{Validity}: If a correct process sends an \msg{install} or a \msg{state-update} message $m$ to set $\Psi$ of processes, then every correct process $p \in \Psi$ eventually \msg{brb}-delivers $m$ or leaves the \name system.
    \item \textbf{No duplication}: No correct process \msg{brb}-delivers the same \msg{install} or the same \msg{state-update} message more than once.
    \item \textbf{Agreement}: If an \msg{install} or a \msg{state-update} message $m$ is \msg{brb}-delivered by some correct process in $\Psi$, then every correct process $p \in \Psi$ eventually \msg{brb}-delivers $m$ or leaves the \name system. 
\end{compactitem}

\begin{algorithm}
\footnotesize
\caption{\brb{} algorithm. Code for process $p$.}
\label{algorithm:brb}
\begin{algorithmic}[1]

\Variables
    \State $\mathit{received} = \emptyset$
\EndVariables

\algvspace
\Procedure{\textit{R-multicast}}{\Psi, m}
    \State $\text{disseminate} \text{ } \langle \Psi, m \rangle \text{ to all } q \in \Psi$ 
\EndProcedure

\algvspace
\UponReceipt{$\langle \Psi, m \rangle$}{q}
    \If{$\mathit{is\_install\_type}(m) = \top \land \Psi = m.v.\mathit{members} \cup m.\omega.\mathit{members} \land \mathit{verify\_converged\_messages}(m) = \top \land \mathit{is\_least\_recent}(m.\omega, m.seq)$} \label{line:conditions_brb}
        \If{$m \notin \mathit{received}$} \label{line:received}
            \State $\text{disseminate} \text{ } \langle \Psi, m \rangle \text{ to all } q' \in \Psi$ \label{line:resend_relable}
            \State $\mathit{received} = \mathit{received} \cup \{m\}$ \label{line:delivered_brb}
            \State \textbf{invoke } \textit{R-delivery}$(\Psi, m)$
        \EndIf
    \EndIf
    
    \If{$\mathit{is\_state-update\_type}(m) = \top$} \label{line:conditions_brb_2}
        \If{$m \notin \mathit{received}$} \label{line:received_2}
            \State $\text{disseminate} \text{ } \langle \Psi, m \rangle \text{ to all } q' \in \Psi$ \label{line:resend_relable_2}
            \State $\mathit{received} = \mathit{received} \cup \{m\}$ \label{line:delivered_brb_2}
            \State \textbf{invoke } \textit{R-delivery}$(\Psi, m)$
        \EndIf
    \EndIf
\EndUponReceipt

\end{algorithmic}
\end{algorithm}

\begin{lemma}
\Cref{algorithm:brb} implements \brb{}.
\end{lemma}
\begin{proof}
Validity of \brb{} is ensured because of the reliable communication we assume.
No duplication is ensured by condition of \cref{line:received,line:received_2} of \Cref{algorithm:brb}.
Lastly, the agreement property is ensured since a correct process always retransmits a message $m$ before \msg{brb}-delivering it (\cref{line:resend_relable,line:resend_relable_2} of \Cref{algorithm:brb}).
\end{proof}

We start by showing that the system eventually ``stabilizes'', i.e., the system eventually stops evolving in any execution of \name.

\begin{lemma} \label{lemma:reconfiguration_done}
The reconfiguration of the system eventually finishes, i.e., there exists a state $s_{final} \in \mathcal{S}$ such that the system does not evolve from $s_{final}$ to a state $s' \in \mathcal{S}$.
\end{lemma}
\begin{proof}
This follows directly from the fact that a finite number of processes can invoke \dbrbJoin{} or \dbrbLeave{} operation in any execution of \name (\Cref{assumption:requests}).
\end{proof}

\Cref{lemma:reconfiguration_done} shows that there exists a ``final'' state of the system $s_{final} \in \mathcal{S}$ such that the system does not evolve from $s_{final}$.

\begin{lemma} \label{lemma:join_safety}
Consider a valid view $v$.
Suppose that $p \in v$, where $p$ is a correct process.
Then, $p$ has invoked \dbrbJoin{} operation.
\end{lemma}
\begin{proof}
Since process $p$ is a member of a valid view $v$, it follows that some process $z$ proposed a sequence of views containing view $v'$ such that $\langle +, p \rangle \in v'.\mathit{changes}$ to replace some view $pv \subset v$.  
Because correct processes do not process \msg{propose} messages that contain a view with $\langle +, p \rangle$ without a signed \msg{reconfig} message from process $p$ (which we omit in \Cref{algorithm:dynamic} for brevity), $p$ must have invoked a \dbrbJoin{} operation.

If $v$ is the initial view of the system, then $p$ has invoked \dbrbJoin{} operation by default (\Cref{assumption:init}).
\end{proof}

\begin{lemma} \label{lemma:leave_safety}
Consider a valid view $v$.
Suppose that $p \notin v$, where $p$ is a correct process, and there exists a valid view $v' \subset v$ such that $p \in v'$.
Then, $p$ has invoked \dbrbLeave{} operation.
\end{lemma}
\begin{proof}
Since process $p$ is not a member of a valid view $v$ and was a member of a valid view $v' \subset v$, it follows that some process $z$ proposed a sequence of views containing view $v''$ such that $\langle -, p \rangle \in v''.\mathit{changes}$ to replace some view $pv \subset v$ (because $p$ was a member of $v' \subset v$).
Because correct processes do not process \msg{propose} messages that contain a view with $\langle -, p \rangle$ without a signed \msg{reconfig} message from process $p$ (which we omit in \Cref{algorithm:dynamic} for brevity), $p$ must have invoked a \dbrbLeave{} operation.
\end{proof}

\Cref{lemma:join_safety,lemma:leave_safety} prove that valid views can incorporate $\langle +, p \rangle$/$\langle -, p \rangle$ change only if $p$ invoked \dbrbJoin{}/\dbrbLeave{} operations.
This is an important fact on which we rely to prove the correctness of our \name.

\begin{definition} \label{definition:active-view}
Valid view $v$ is \textbf{active} at time $t$ if at least $v.q$ correct processes that are members of $v$ ``know'' at $t$\footnote{This means that these correct processes has a proof that $v$ is ``instantiated'' by \name. For a more thorough explanation, see \Cref{subsection:view_discovery}.} that $v$ is a valid view and has not left the system before $t$ (i.e., executed \Cref{line:left}).
Otherwise, valid view $v$ is \textbf{inactive} at time $t$.
\end{definition}

\begin{lemma} \label{lemma:termination_active}
Suppose that the current state of the system is $s \in \mathcal{S}$.
Moreover, suppose that a correct process $p \in v$ disseminates a \msg{propose} message with a sequence of views $seq_p$ at time $t$ to replace some view $v \in s$.
Then, $p$ eventually converges on a sequence of views $seq_p'$ to replace $v$ or $v$ is inactive at some time $t' \geq t$.
\end{lemma}
\begin{proof}
When $p$ proposes $seq_p$ to replace $v$, $p$ disseminates a \operation{propose}{seq_p, v} message.
Since we assume the reliable communication between processes, messages sent by correct processes are received by correct processes.
For process $p$ to converge on a sequence of views $seq_p'$ to replace $v$, at least $v.q$ processes should send a \msg{converged} message for $seq_p'$.
During this period, $p$ changes the value of its variable $\mathit{SEQ^v}$ with sequences of views received from other processes.
Since there are at least $v.q$ correct processes in $v$ and a finite number of reconfiguration requests (i.e., even a malicious process can not propose new views an infinite number of times; \Cref{assumption:requests}), at least $v.q$ processes are going to send a \msg{converged} message for the same sequence. 
Hence, $p$ eventually converges on $seq_p'$ to replace a view $v$.

Otherwise, less than $v.q$ correct processes did not leave the system.
Hence, $v$ is inactive at some time $t' \geq t$.
\end{proof}

\begin{definition} [Path between views] \label{definition:path}
Suppose that the current state of the system is $s \in \mathcal{S}$.
Let $v, v' \in s$ be two different views in $s$.
We say that there exists a \textbf{path} between $v$ and $v'$ if there exist valid views $v_1, v_2, ..., v_n, (n \geq 0)$, such that at least one correct process $p_1 \in v \cup v_1$ receives $v.q$ \msg{state-update} messages associated with $v$, at least one correct process $p_{v'} \in v_{n} \cup v'$ receives $v_n.q$ \msg{state-update} messages associated with $v_n$ and at least one correct process $p_i \in v_{i - 1} \cup v_i$ receives $v_{i - 1}.q$ \msg{state-update} messages associated with $v_{i - 1}$, for every $2 \leq i \leq n$.
\end{definition}

\begin{lemma} \label{lemma:installable_reach_1}
Suppose that the current state of the system is $s \in \mathcal{S}$.
Consider views $v, v' \in s$, such that $v$ leads to $v'$ in $s$.
If every correct process $q \in v$ eventually sets $cv_q = v$, then every correct process $p \in v'$ eventually sets $cv_p = v'$.
\end{lemma}
\begin{proof}
By \Cref{lemma:installable_from_one}, there is a view $v \in s$, where $v$ is installable in $s$, that leads to $v'$ in $s$.
By \Cref{definition:leads_to}, a sequence of views $seq = v_1 \to ... \to v_n \to v', (n \geq 0)$ is converged on to replace $v$.
Moreover, a sequence of views $seq_{i} = v_{i + 1} \to ... \to v_n \to v'$ is converged on to replace view $v_i$, for every $1 \leq i < n$, and a sequence of views $seq_{n} = v'$ is converged on to replace $v_n$.
We can conclude that some correct process $q_0 \in v$ sent a \msg{propose} message that led to a fact that sequence of views $seq = v_1 \to ... \to v_n \to v'$ is converged on to replace $v$.
Similarly, there is at least one correct process $q_i \in v_i$ that sent a \msg{propose} message that led to a fact that sequence of views $seq_i = v_{i + 1} \to ... \to v_n \to v'$ is converged on to replace $v_i$ (for every $1 \leq i < n$).
Lastly, there is at least one correct process $q_n \in v_n$ that sent a \msg{propose} message that led to a fact that sequence of views $seq_n = v'$ is converged on to replace $v_n$.
Note that $\exists v_p \in \{v_1, v_2, ..., v_n, v'\}\not\exists v_p' \in \{v_1, v_2, ..., v_n, v'\}: p \in v_p \land p \in v_p' \land v_p' \subset v_p$.
In order to show that the lemma holds, we will show that eventually there exists a path between $v$ and $v'$.

Suppose that no different sequence of views is ever converged on to replace $v$ (i.e., $seq$ is the only sequence of views that is \emph{ever} converged on to replace $v$).
We conclude that there exists a path between $v$ and $v_n$, since $seq$ is the only sequence of views converged on to replace $v$ and all correct processes in $v$ eventually ``reach'' $v$ (by the claim of the lemma).
Moreover, it is easy to see that if a correct process $q' \in v'$ has already set $cv_{q'} = v'$, there exists a path between $v$ and $v'$.
However, it is not certain that there exists such a process $q'$.
Let us investigate both cases:
\begin{compactitem}
    \item There exists such a process $q'$: Hence, the properties of \brb{}, the fact that there exists a path between $v$ and $v'$ and the fact that there exists a view in $seq$ with $p$ being its member ensure that a correct process $p$ eventually sets $cv_p = v'$.
    \item There does not exist such a process $q'$:
    In this case, it is clear that $v_n$ eventually becomes active once every correct member of $v_n$ considers this view as its current view of the system (and that happens because of the fact that there exists a path between $v$ and $v_n$).
    Because of the fact that $seq$ is the only view converged on to replace $v$, $v_n$ is an auxiliary view for $v$ and \Cref{lemma:termination_active}, a correct process $r \in v_n$ eventually converges on a sequence of views $seq_n = v'$, ensuring that a correct process $p$ eventually receives an \msg{install} message for $v'$.
    Moreover, $p$ eventually receives the quorum of \msg{state-update} messages associated with $v_n$ (\Cref{line:state-update}) and updates its current view to $v'$ (\Cref{line:set_cv}).
\end{compactitem}

\bigskip
Suppose that a sequence of views $seq' = v_1' \to ... \to v_m' \to v'', (m \geq 0)$ is also converged on to replace $v$, where $seq' \neq seq$.
By \Cref{lemma:weak_accuracy}, either $seq \subset seq'$ or $seq \supset seq'$.
We investigate both cases:
\begin{compactitem}
    \item $seq \subset seq'$: 
    Hence, $v_1, v_2, ..., v_n, v'$ are in $seq'$.
    Eventually, either a correct process $r \in v_1$ or a correct process $r' \in v_1'$ receives the quorum of \msg{state-update} messages associated with $v$.
    Hence, eventually there exists a path between $v$ and either $v_1$ or $v_1'$.
    Let say that there exists a path between $v$ and $v_{new}$, where $v_{new} = v_1$ or $v_{new} = v_1'$.
    We consider two cases now:
    \begin{compactitem}
        \item $v_{new} \notin seq$: 
        Hence, there exists a path between $v$ and $v_{new}$.
        There can be at most one sequence of views accepted (and converged on) to replace $v_{new}$: $seq_1 = \{\omega \in seq': v_{new} \subset \omega\}$.
        We conclude that $v' \in seq_1$.
        Eventually, either a correct process $r \in v_{new}'$ receives a quorum of \msg{state-update} message associated with $v_{new}$ (since there exists a path between $v$ and $v_{new}$ and all correct processes from $v_{new}$ could ``reach'' $v_{new}$) or some correct process $r' \in v_{new}''$ receives the quorum of \msg{state-update} messages associated with some view, where $v_{new}'' \in seq$ or $v_{new}'' \in seq'$.
        Therefore, eventually there exists a path between $v$ and $v_{new}'$ or between $v$ and $v_{new}''$.
        Assume that there exists now a path between $v$ and $v_{new}$, where $v_{new}$ is ``updated'' (i.e., it is not the ``old'' $v_{new}$).
        If $v_{new} = v'$, then the lemma holds in case where $seq \subset seq'$.
        Otherwise, this or the second case apply to $v_{new}$.
        
        \item $v_{new} \in seq$: There exists a path between $v$ and $v_{new}$.
        Since $seq \subset seq'$, we conclude that $v_{new} \in seq'$.
        There can be at most two sequences of views accepted (and converged on) to replace $v_{new}$: $seq_1 = \{\omega \in seq: v_{new} \subset \omega\}$ and $seq_2 = \{\omega \in seq': v_{new} \subset \omega\}$.
        Note that $seq_1$ is converged on to replace $v_{new}$.
        Because of \Cref{lemma:weak_accuracy} and the fact that $seq \subset seq'$, we conclude that $seq_1 \subseteq seq_2$, $v' \in seq_1$ and $v' \in seq_2$.
        
        Every correct process from $v_{new}$ eventually ``reaches'' $v_{new}$.
        Eventually, a correct process $r \in v_{new}'$ receives a quorum of \msg{state-update} message associated with $v_{new}$ (since there exists a path between $v$ and $v_{new}$ and all correct processes from $v_{new}$ ``reach'' $v_{new}$).
        Assume that there exists now a path between $v$ and $v_{new}$, where $v_{new}$ is ``updated'' (i.e., it is not the ``old'' $v_{new}$).
        If $v_{new} = v'$, then the lemma holds in case where $seq \subset seq'$.
        Otherwise, this or the first case apply to $v_{new}$.
    \end{compactitem}
    Since both $seq$ and $seq'$ have finite number of views (because of \Cref{assumption:requests}), the recursion eventually stops with $v_{new} = v'$.
    The lemma holds in this case.
    
    \item $seq \supset seq'$:
    Hence, $v_1', v_2', ..., v_m', v''$ are in $seq$.
    Eventually, either a correct process $r \in v_1$ or a correct process $r' \in v_1'$ receives the quorum of \msg{state-update} messages associated with $v$.
    Hence, eventually there exists a path between $v$ and either $v_1$ or $v_1'$.
    Let say that there exists a path between $v$ and $v_{new}$, where either $v_{new} = v_1$ or $v_{new} = v_1'$.
    We consider three cases now:
    \begin{compactenum}
        \item $v_{new} \in seq$ and $v_{new} \notin seq'$:
        There can be at most one sequence of views accepted (and converged on) to replace $v_{new}$: $seq_1 = \{\omega \in seq: v_{new} \subset \omega\}$.
        Note that $seq_1$ is converged on to replace $v_{new}$.
        We conclude that $v' \in seq_1$.
        Eventually, either a correct process $r \in v_{new}'$ receives a quorum of \msg{state-update} message associated with $v_{new}$ (since there exists a path between $v$ and $v_{new}$ and all correct processes from $v_{new}$ could ``reach'' $v_{new}$) or some correct process $r' \in v_{new}''$ receives the quorum of \msg{state-update} messages associated with some view, where $v_{new}'' \in seq$ or $v_{new}'' \in seq'$.
        Therefore, eventually there exists a path between $v$ and $v_{new}'$ or between $v$ and $v_{new}''$.
        Assume that there exists now a path between $v$ and $v_{new}$, where $v_{new}$ is ``updated'' (i.e., it is not the ``old'' $v_{new}$).
        If $v_{new} = v'$, then the lemma holds in case where $seq \supset seq'$.
        Otherwise, this or other cases apply to $v_{new}$.
        
        \item $v_{new} \in seq$, $v_{new} \in seq'$ and $v_{new} \neq v''$:
        There can be at most two sequences of views accepted (and converged on) to replace $v_{new}$: $seq_1 = \{\omega \in seq: v_{new} \subset \omega\}$ and $seq_2 = \{\omega \in seq': v_{new} \subset \omega\}$.
        Note that $seq_1$ is converged on to replace $v_{new}$.
        
        Every correct process from $v_{new}$ eventually ``reaches'' $v_{new}$.
        Eventually, a correct process $r \in v_{new}'$ receives a quorum of \msg{state-update} message associated with $v_{new}$ (since there exists a path between $v$ and $v_{new}$ and all correct processes from $v_{new}$ ``reach'' $v_{new}$).
        Assume that there exists now a path between $v$ and $v_{new}$, where $v_{new}$ is ``updated'' (i.e., it is not the ``old'' $v_{new}$).
        If $v_{new} = v'$, then the lemma holds in case where $seq \supset seq'$.
        Otherwise, this or other cases apply to $v_{new}$.
        
        \item $v_{new} \in seq$, $v_{new} \in seq'$ and $v_{new} = v''$:
        We know that $seq_1 = \{\omega \in seq: v_{new} \subset \omega\}$ is converged on to replace $v_{new}$ (because of \Cref{definition:leads_to} and $v'' \in seq$).
        Without loss of generality, suppose that $v'' = v_1$.
        Hence, a sequence of views $seq_1 = v_2 \to ... \to v_n \to v'$ is converged on to replace $v''$.
        Now, there can be a finite number of sequences of views converged on to replace $v''$ (one is $seq_1$).
        Consider such a sequence of views $seq_2$.
        By \Cref{lemma:weak_accuracy}, we conclude that either $seq_1 \subseteq seq_2$ or $seq_1 \supset seq_2$.
        Therefore, we reduce this case to the case from the beginning (now $v''$ is what view $v$ was originally).
    \end{compactenum}
\end{compactitem}

Since only a finite number of sequences can be converged on to replace some valid view (because of \Cref{assumption:requests}), we conclude that the (potential) recursion eventually stops.
For the sake of simplicity, we have showed the proof in a case there are two different sequences of views converged on to replace $v$.
It is easy to generalize this proof for any (finite) number of different sequences of views converged on to replace $v$ (note that a finite number of different sequences of views can be converged on to replace $v$ because of \Cref{assumption:requests}).
\end{proof}

\begin{lemma} \label{lemma:installable_reach_2}
Suppose that the current state of the system is $s \in \mathcal{S}$.
Consider a view $v' \in s$, where $v'$ is installable in $s$.
Let $p \in v'$ be a correct process.
Then, $p$ sets $\mathit{cv_p} = v'$ (i.e., executes \Cref{line:set_cv}).
\end{lemma}
\begin{proof}
This follows directly from the fact that all correct processes that are members of the initial view $v_0$ have their current view set to $v_0$ and \Cref{lemma:installable_reach_1}.
\end{proof}

\begin{lemma} \label{lemma:last_installable}
Suppose that current state of the system is $s \in \mathcal{S}$.
Consider a view $v \in s$ installable in $s$.
If $v \neq v_{final}$, then either $v$ leads to some view $v' \in s$ or the system eventually transits to a state $s' \in \mathcal{S}$, where $v$ leads to some view $v' \in s'$ in $s'$.
\end{lemma}
\begin{proof}
Suppose that $v \neq v_{final}$ does not lead to any view in $s$, i.e., $v$ is the most recent installable view in $s$.
We now show that the system transits to a state $s' \in \mathcal{S}$, such that $v$ leads to some view $v' \in s'$ in $s'$.

Consider a correct process $p \in v$.
A sequence of views $seq$ is eventually converged on to replace $v$ since $v$ eventually becomes active (because of \Cref{lemma:installable_reach_2}).

Now, there are two possible scenarios:
\begin{compactenum}
    \item Number of different sequences of views ever converged on to replace $v$ is equal to 1:
    Let us denote that sequence of views with $seq = v_1 \to ... \to v_n \to v', (n \geq 0)$.
    Hence, every correct process $r \in v_1$ eventually receives an \msg{install} message associated with $v$, where $\omega$ parameter is equal to $v_1$.
    Moreover, $r$ executes \Cref{line:propose_rest}.
    By \Cref{lemma:termination_active}, the fact that $v_1$ is active at some time $t'$ and the fact that only $seq$ is converged on to replace $v$, $r$ converges on a sequence of views $seq_1 = v_2 \to ... \to v_n \to v'$ to replace $v_1$.
    This repeats until a correct process $r \in v'$ sets $v'$ as $cv_r$.
    Moreover, $v'$ is installable in the current state $s'$ of the system at that time and views $v_1, ..., v_n \in s'$ are not installable in $s'$, i.e., \Cref{definition:leads_to} is satisfied.
    Hence, the lemma holds.
    
    \item Number of sequences of views ever converged on to replace $v$ is bigger than 1:
    We can use the similar argument as in the proof of \Cref{lemma:installable_reach_1} in this case.
\end{compactenum}
Therefore, the lemma holds.
\end{proof}

\begin{lemma} \label{lemma:v_final_installable}
Suppose that the current state of the system is $s_{final} \in \mathcal{S}$.
Then, $v_{final}$ is installable in $s_{final}$.
\end{lemma}
\begin{proof}
Suppose that $v_{final}$ is not installable in $s_{final}$.
Let $v \neq v_{final}$ be the most recent installable view in $s_{final}$.
By \Cref{lemma:last_installable}, $v$ either leads to some view in $s_{final}$ or the system evolves into a new state.
If $v$ leads to some view in $s_{final}$, we reach a contradiction with the fact that $v$ is the most recent installable view in $s_{final}$.
Otherwise, we reach a contradiction with the fact that the system does not evolve anymore (since it is already in the state $s_{final}$).
\end{proof}

\begin{lemma} \label{lemma:install_v_final}
Every correct process $p \in v_{final}$ installs $v_{final}$.
\end{lemma}
\begin{proof}
Since $v_{final}$ is installable in $s_{final} \in \mathcal{S}$ (by \Cref{lemma:v_final_installable}), every correct process $p \in v_{final}$ eventually sets $cv_p = v_{final}$ (by \Cref{lemma:installable_reach_2}).
A correct process $p \in v_{final}$ would not install $v_{final}$ only if the condition at \Cref{line:more_updated_views_1} would hold.
Suppose that is the case. 
Hence, $p$ proposes some view $v' \supset v_{final}$ to replace $v_{final}$.
Because of \Cref{lemma:installable_reach_2}, $v_{final}$ eventually becomes active.
Hence, $p$ converges on some sequence of views to replace $v_{final}$ (because of \Cref{lemma:termination_active}).
However, this is a contradiction with the fact that $s_{final}$ is the state of the system from which the system does not evolve.
Hence, the lemma holds.
\end{proof}

\begin{lemma} \label{lemma:install_v_final_2}
Consider a correct process $p$ that has completed the \dbrbJoin{} operation (i.e., $p$ has joined the system) and does not invoke the \dbrbLeave{} operation ever.
Then, $p \in v_{final}$ installs $v_{final}$.
\end{lemma}
\begin{proof}
Follows directly from \cref{lemma:install_v_final,lemma:leave_safety}.
\end{proof}

\begin{lemma} \label{lemma:eventually_joins}
Suppose that the current state of the system is $s \in \mathcal{S}$.
Consider a correct process $p$ and views $v \in s$ and $v' \in s$.
Suppose that $v \not\ni p$ leads to $v' \ni p$ in $s$.
Then, $p$ eventually joins the system, i.e., executes \Cref{line:join_completed}.
\end{lemma}
\begin{proof}
By \Cref{definition:leads_to}, a sequence of views $seq = v_1 \to ... \to v_n \to v', (n \geq 0)$ is converged on to replace $v$.
Moreover, $\exists v_p \in seq: p \in v_p$.
Hence, $p$ eventually receives an \msg{install} message for $v_p$ and joins the system (similarly to the proof of \Cref{lemma:installable_reach_1}).
\end{proof}

\begin{lemma} \label{lemma:installable_leaves}
Suppose that the current state of the system is $s \in \mathcal{S}$.
Moreover, suppose that some correct process $p$ has $cv_p = v$ at current time $t$, where $v \in s$ and $v \ni p$ is an installable view in $s$.
Suppose that $v$ leads to $v' \not\ni p$ in $s$.
Then, $p$ eventually executes \Cref{line:finish_wait_updated}.
\end{lemma}
\begin{proof}
By \Cref{definition:leads_to}, a sequence of views $seq = v_1 \to ... \to v_n \to v', (n \geq 0)$ is converged on to replace $v$.
Moreover, $\exists v_p \in seq: p \notin v_p$.
Suppose that $seq$ is the only sequence of views ever converged on to replace $v$.
Hence, $p$ eventually receives an \msg{install} message for a view $v_p \not\ni p$ and then executes \Cref{line:finish_wait_updated}.
Hence, the lemma holds in this case.

Suppose now that $seq$ is not the only sequence of views ever converged on to replace $v$.
We can use the similar argument as in the proof of \Cref{lemma:installable_reach_1} in this case.
\end{proof}

%% file: sections/appendix/view-discovery.tex
\subsection{View Discovery Protocol}
\label{subsection:view_discovery}

\vspace*{-2mm}
The \emph{View Discovery} protocol is of the utmost importance for \name. 
Informally, there are two crucial roles View Discovery protocol plays in \name:
\begin{compactitem}
    \item A correct process $p$ that wants to join the system must be able to discover the current constitution of the system in order to join.
    Process $p$ achieves this using the View Discovery protocol.
    \item A correct process $p$ that joined the system must be able to discover valid views.
    Otherwise, $p$ can be tricked into delivering some invalid message $m$ ``planted'' by a malicious process, i.e., the consistency of \name can be violated.
    Moreover, totality and validity properties can be violated.
    Again, the View Discovery protocol provides the key part of the solution for these problems.
\end{compactitem}




\begin{definition} [Valid view history]
An alternating sequence of views and messages $his = v_0\, m_0\, v_1\, m_1\, ... \,v_{n - 1}\, m_{n - 1}\, v_n$ is said to be a \textbf{valid view history} if:
\begin{compactenum}
    \item View $v_0$ is the initial view of the system; and
    \item For each message $m_i (0 \leq i \leq n - 1)$, $m_i$ is an \msg{install} message with the following parameters: i) $\omega = v_{i+1}$; ii) $seq = v_{i + 1} \to ...$; iii) $v = v_i$; and iv) $m_i$ carries $v_i.q$ signed appropriate \msg{converged} messages for $seq$.
\end{compactenum}
Moreover, we say that $his$ is a valid view history up to a view $v_i$, for every $0 \leq i \leq n$.
\end{definition}

Let us now give an overview of the protocol:
Whenever a correct process starts trusting a sequence of views\footnote{Note that this implies that the process has not left the system.} (i.e., obtains a valid view history up to some view $v$), it disseminates that information (i.e., the valid view history) to all processes in the universe.
A correct process $q$ receives view histories sent by other processes, checks them and if the verification is successful, it starts trusting the view history, i.e., it knows that the most recent view of the view history is valid.

When a correct process invokes the \dbrbJoin{} operation, it floods the network asking for valid view histories of the processes in order to obtain the most recent view of the system.
Once it receives an answer, it can verify whether the received view history is valid.
Moreover, whenever a correct process executes \Cref{algorithm:install}, it does so because it has received an \msg{install} message and ``knows'' that the view $v$ from the \msg{install} message is valid.
Hence, whenever a correct process executes \Cref{algorithm:install}, it disseminates to all processes a valid view history up to view $v$.
The reason is that processes that are members of $\omega$ (from the \msg{install} message), but are not members of $v$, might not ``know'' that $v$ is indeed a valid view.
In this way, they eventually receive the valid view history up to $v$ and can verify that $v$ is indeed valid.

In summary, the View Discovery protocol simply ensures that correct processes are able to reason about the validity of a view (along with allowing correct processes to discover the current membership of the system).
Another simple principle for implementing the ``validity'' part of the protocol is to simply associate a valid view history with a view.
For example, if a message is associated with some view $v$, view $v$ is associated with a valid view history up to $v$.
Otherwise, the message is considered invalid.


\begin{lemma} \label{lemma:valid_view_history_up_to_v}
Suppose that the current state of the system is $s \in \mathcal{S}$.
Moreover, suppose that some correct process $p$ executed \Cref{algorithm:install} since $p$ has received a message $m = \text{\operation{install}{v', seq, v}}$, where $m$ carries $v.q$ appropriate \msg{converged} messages and views $v, v' \in s$.
Then, $p$ has a valid view history up to view $v$.
\end{lemma}
\begin{proof}
A correct process executes \Cref{algorithm:install} upon receiving an \msg{install} message associated with view $v$ only if it has a valid view history up to $v$.
Therefore, the lemma holds.
\end{proof}

\begin{lemma} \label{lemma:vd_first}
Suppose that the current state of the system is $s \in \mathcal{S}$.
Moreover, suppose that some correct process $p$ executes \Cref{line:set_cv}, which assigns view $v' \ni p$ as $p$'s current view and $v' \in s$.
Then, $p$ has a valid view history up to $v'$.
\end{lemma}
\begin{proof}
A correct process $p$ executes \Cref{line:set_cv} since it has received a corresponding \msg{install} message with $\omega = v'$.
Suppose that the aforementioned \msg{install} message is associated with some view $v \in s$.
By \Cref{lemma:valid_view_history_up_to_v}, $p$ has a valid view history up to $v$.
Let us denote $p$'s valid view history up to $v$ by $his_v$.

Given the fact that $p$ has received the \msg{install} message associated with $v$, we can conclude that $p$ also has a valid view history up to $v'$ ($his_v || m || v'$, where $m$ is the received \msg{install} message and $||$ represents the concatenation) and the statement of the lemma follows.
\end{proof}

\begin{lemma} \label{lemma:history_installable}
Suppose that the current state of the system is $s \in \mathcal{S}$.
Moreover, suppose that some correct process $p$ has a valid view history $his_{v'}$ up to some view $v' \in s$.
Suppose that $v'$ is installable in $s$ and $v'$ is a view different from the initial view of the system.
Then, $p$ has a valid view history up to some installable view $v \in s$, where $v$ leads to $v'$ in $s$.
\end{lemma}
\begin{proof}
By \Cref{lemma:unique_installable_views}, installable views in $s$ form a sequence of views.
Moreover, one view can lead to at most one other view (\Cref{lemma:only_one_view_leads_to}).
Lastly, there is exactly one view $v$ that leads to $v'$ (\Cref{lemma:from_only_one}).

Since $p$ has a valid view history up to an installable view $v'$ (different from the initial view of the system), only view that leads to $v'$ in $s$ is $v$ and \Cref{definition:leads_to}, we conclude that $v \in his_{v'}$ and hence the lemma holds.
\end{proof}

\begin{lemma} \label{lemma:history_installable_2}
Suppose that the current state of the system is $s \in \mathcal{S}$.
Moreover, suppose that some correct process $p$ has a valid view history $his_{v'}$ up to some view $v' \in s$.
Suppose that $v'$ is installable in $s$, $v'$ is a view different from the initial view of the system and some view $v \in s$ leads to $v'$ in $s$.
Let the system transit into the state $s' \in \mathcal{S}$, where some view $v'' \in s$, $v'' \neq v$, leads to $v'$.
Then, $p$ has a valid view history up to $v''$.
\end{lemma}
\begin{proof}
By \Cref{definition:leads_to}, a sequence of views $seq' = v_1 \to ... \to v_n \to v', (n \geq 1)$ is converged on to replace $v'$.
In the state $s'$ of the system, a view $v''$ leads to $v'$. 
This implies that $v'' \supset v$.
In order to show that the lemma holds, it suffices to show that $v'' \in seq'$.

We deduce that $v$ leads to $v''$ in $s'$ (since $v$ leads to $v'$ in $s$ and $v''$ leads to $v'$ in $s'$).
By \Cref{definition:leads_to}, a sequence of views $seq'' = v_1' \to ... \to v_m \to v'', (m \geq 1)$ is converged on to replace $v$.
By \Cref{lemma:weak_accuracy} and the fact that $v' \neq v''$, either $seq' \subset seq''$ or $seq' \supset seq''$.
Let us consider both cases:
\begin{compactitem}
    \item $seq' \subset seq''$: Hence, $v' \in seq''$. However, this is a contradiction with the fact that $v'' \subset v'$.
    \item $seq' \supset seq''$: In this case, we have that $v'' \in seq'$.
\end{compactitem}
Since $v'' \in seq'$ and \Cref{lemma:valid_view_history_up_to_v}, the lemma holds.
\end{proof}

\begin{lemma} \label{lemma:history_auxiliary}
Suppose that the current state of the system is $s \in \mathcal{S}$.
Suppose that some correct process $p$ has a valid view history $his_{v'}$ up to some view $v' \in s$.
Moreover, suppose that $v'$ is not an installable view in $s$, i.e., it is an auxiliary view for some installable view $v \in s$ in $s$.
Then, $p$ has a valid view history up to $v$.
\end{lemma}
\begin{proof}
By \Cref{lemma:one_auxiliary_view}, exactly one view $v$ can exist such that $v'$ is an auxiliary view for $v$.
Since $p$ has a valid view history up to $v'$ and only view $v$ exists such that $v'$ is an auxiliary view for $v$ in $s$ and \Cref{definition:auxiliary_view}, then $v \in his_{v'}$ and we can conclude that the lemma holds.
\end{proof}

\begin{lemma} \label{lemma:history_auxiliary_2}
Suppose that the current state of the system is $s \in \mathcal{S}$.
Moreover, suppose that some correct process $p$ has a valid view history $his_{v'}$ up to some view $v' \in s$.
Suppose that $v'$ is not installable in $s$, i.e., it is an auxiliary view for some installable view $v \in s$ in $s$.
Let the system transit into the state $s' \in \mathcal{S}$, such that $v'$ is an auxiliary view for $v'' \in s'$ in $s'$, $v'' \neq v$.
Then, $p$ has a valid view history up to $v''$.
\end{lemma}
\begin{proof}
Similarly to the proof of \Cref{lemma:history_installable_2}.
\end{proof}

\begin{lemma} \label{lemma:vd_second}
Suppose that the current state of the system is $s \in \mathcal{S}$.
Moreover, suppose that some correct process $p$ has a valid view history up to some view $v' \in s$.
For every view $v \subseteq v'$ installable in $s$, $p$ has a valid view history up to $v$.
\end{lemma}
\begin{proof}
Let us denote $p$'s valid view history up to $v'$ with $his_{v'}$.
Hence, we should show that $v \in his_{v'}$, for every installable view $v \subseteq v'$.

Suppose that $v'$ is an installable view different from the initial view of the system.
Using the backward induction (\Cref{lemma:history_installable}), we deduce that the lemma holds.

Suppose that $v'$ is not installable and that $v'$ is an auxiliary view for some view $v \subset v'$ in $s$. 
Then, \Cref{lemma:history_auxiliary} claims that $p$ has a valid view history up to $v$.
Similarly to the previous case, we reach the conclusion that the lemma holds (using the backward induction and \Cref{lemma:history_installable}).
\end{proof}

\begin{lemma} \label{lemma:vd_temp}
Suppose that some correct process $p$ has a valid view history up to some valid view $v'$.
Consider the state $s_{final}$ of the system.
Then, $p$ has a valid view history up to any view $v \in s_{final}$ installable in $s_{final}$, where $v \subseteq v'$.
\end{lemma}
\begin{proof}
Follows directly from \cref{lemma:history_installable_2,lemma:history_auxiliary_2,lemma:vd_second}.
\end{proof}


\begin{lemma} \label{lemma:vd_final}
Suppose that the current state of the system is $s \in \mathcal{S}$.
Moreover, suppose that a view $v' \in s$ is the most recent view in $s$ such that a correct process $p \in v'$ assigned $v'$ as its current view (executed \Cref{line:set_cv}).
For every view $v \subseteq v'$ installable in $s_{final} \in \mathcal{S}$, $p$ has a valid view history up to $v$.
\end{lemma}
\begin{proof}
The lemma follows directly from \cref{lemma:vd_first,lemma:vd_temp}.
\end{proof}

Consider a correct process $p$ that sets a view $v'$ as its current view at time $t$, where $v'$ is the most recent view such that a correct process assigned the view as its current view.
\Cref{lemma:vd_final} proves that $p$ has a valid view history up to $v$, where this holds for every installable view $v \subseteq v'$.
Hence, at any time $t$ there is at least one correct process $p$ that has a valid view history up to every installable view $v \subseteq v'$.
Moreover, \Cref{lemma:installable_reach_2} shows that processes ``progress'' and update their current views.
For example, if a correct process $q$ wants to join the system, $q$ executes the View Discovery protocol by flooding the network in order to learn view histories of processes and eventually $q$ receives the valid view history of process $p$.
In this way, $q$ will be able to bootstrap itself, eventually join the system and learn all the installable views ``created'' by \name during an execution.

However, once $q$ joins the system it stops executing the View Discovery protocol.
Then, it is important to ensure that $q$ will ``observe'' every installable view that is ``instantiated'' by \name after $q$'s join is complete.

\begin{lemma} \label{lemma:valid_view_history_after_join}
Suppose that the current state of the system is $s \in \mathcal{S}$.
Moreover, suppose that some correct process $p$ completes its join (i.e., executes \Cref{line:join_completed}).
Suppose that $cv = v$ at process $p$ at this time, where $v \in s$.
Then, $p$ eventually has a valid view history up to a view $v' \supseteq v$ in some state $s' \in \mathcal{S}$, for every installable view $v' \ni p$, where $s'$ is a state to which the system may transit.
\end{lemma}
\begin{proof}
Suppose that $v$ is installable.
\Cref{lemma:vd_first} ensures that $p$ has a valid view history up to $v$.
If $v$ leads to $v' \ni p$ in some state of the system, \cref{lemma:installable_reach_2,lemma:last_installable} ensure that $p$ eventually sets $v'$ as $cv_p$.
Again, \Cref{lemma:vd_first} ensures that $p$ has a valid view history up to $v'$.
This reasoning can be now repeated if $v'$ leads to some view $v''$ in some state of the system.

Suppose now that $v$ is not installable.
This means that $v$ is an auxiliary view for some installable view $v''$.
If $v''$ leads to $v' \ni p$ in some state of the system, \cref{lemma:installable_reach_2,lemma:last_installable} ensure that $p$ eventually sets $v'$ as $cv_p$.
\Cref{lemma:vd_first} then ensures that $p$ has a valid view history up to $v'$.
Now, if $v'$ leads to some view $v''' \ni p$, the reasoning from the case above can be applied.
The lemma holds.
\end{proof}

%% file: sections/appendix/dynamic-properties.tex
\subsection{Dynamicity} \label{sec:dynamicity_proof}

\vspace*{-2mm}
Note that the set of installed views is a subset of the set of installable views.
Hence, if a view $v$ is installed, then $v$ is installable.
We say that $V(t)$ is the most recent view installed in the system at a global time $t$ and $I(t)$ is the most recent view installable in the system at a global time $t$.

\begin{lemma} \label{lemma:collect_rec_confirm}
Suppose that a correct process $p$ invokes a \dbrbJoin{} operation at time $t$.
There exists a time $t' \geq t$ such that $p$ receives $V(t').q$ \msg{rec-confirm} messages or the system transits to a state $s \in \mathcal{S}$ where a view $v' \in s$ is installable in $s$ and $p \in v'$.
\end{lemma}
\begin{proof}
Eventually, a correct process $q \in v$ has a valid view history up to $v = I(t)$ (\cref{lemma:installable_reach_2,lemma:vd_final}).
During the View Discovery protocol, $p$ floods the universe of processes in order to update its view history.
Eventually, $p$ learns the view history of process $q$.

Then, $p$ broadcasts a \msg{reconfig} message to $v$.
There are two possible scenarios:
\begin{compactenum}
    \item Process $p$ receives $v.q$ \msg{rec-confirm} messages.
    Then, $v$ is installed (since correct processes send \msg{rec-confirm} messages only if a view is installed) and the lemma holds.
    \item Process $p$ does not receive $v.q$ \msg{rec-confirm} messages.
    This means that the system reconfigured from $v$ to some view $v'$, where $v'$ is the ``new'' most recent view installable in the system.
    Again, some correct process $q' \in v'$ eventually obtains a valid view history up to $v'$.
    Process $p$ learns the valid view history up to $v'$ (from $q'$) and broadcasts a \msg{reconfig} message to $v'$.
    
    It is possible that $p$ does not collect $v'.q$ \msg{rec-confirm} messages again.
    In this case, $p$ repeats the broadcasting of a \msg{reconfig} message to some view $v''$.
    
    Recall the fact that there exists a finite number of reconfiguration requests in every execution of \name (\Cref{assumption:requests}).
    This guarantees that the system eventually reaches the view from which it is impossible to reconfigure.
    Let us denote that view with $v_{quasi-final}$.
    When that happens, $p$ receives $v_{quasi-final}.q$ \msg{rec-confirm} messages and the lemma holds.
\end{compactenum}

If $p$ does not receive enough \msg{rec-confirm} messages, that is because there exists an installable view that contains $p$ as its member.
Therefore, the lemma holds. 
\end{proof}

\begin{lemma} \label{lemma:collect_rec_confirm_leave}
Suppose that a correct process $p$ invokes a \dbrbLeave{} operation at time $t$.
There exists a time $t' \geq t$ such that $p$ receives $V(t').q$ \msg{rec-confirm} messages or the system transits to a state $s \in \mathcal{S}$ where a view $v' \in s$ is installable in $s$ and $p \notin v'$.
\end{lemma}
\begin{proof}
Following the similar arguments as in the proof of \Cref{lemma:collect_rec_confirm}, we conclude that the lemma holds.
\end{proof}

\begin{lemma} \label{lemma:v(t')}
Suppose that a correct process $p$ invokes a \dbrbJoin{}/\dbrbLeave{} operation at time $t$, by disseminating a \msg{reconfig} message to processes that are members of $V(t)$.
Moreover, suppose that $p$ receives $V(t).q$ \msg{rec-confirm} messages (\cref{lemma:collect_rec_confirm,lemma:collect_rec_confirm_leave}).
Then, the system eventually transits to a state $s \in \mathcal{S}$ where a view $v \in s$ is installable in $s$, $v \neq V(t)$ and $\langle +, p \rangle/\langle -, p \rangle \in v$.
\end{lemma}
\begin{proof}
Let us say that a correct process $p$ disseminates a $m_{join} = \text{\operation{reconfig}{\langle +, p \rangle, cv}}$ message, where $\mathit{cv} = V(t)$, to processes that are members of $V(t)$ at time $t$. 
Two cases are possible:
\begin{compactenum}
    \item All correct processes in $V(t)$ propose a sequence of views with a view that contains $c = \langle +, p \rangle$ to replace $V(t)$. 
    In this case, the fact that all proposed sequences of correct processes include a view with change $c$ ensures that there will be a view $v$, such that $v \neq V(t)$, where $c \in v$.
    
    \item Some correct processes in $V(t)$ propose a sequence, before the receipt of $m_{join}$, with a view that does not contain $c$.
    In this case, considering a time $t' > t$ when there is a new view $V(t')$ installed, two scenarios are possible:
    \begin{compactenum}
        \item $c \in V(t')$: since the proposed sequences from processes that received $m_{join}$ and proposed a view with $c$ were computed in $V(t')$.
        \item $c \not\in V(t')$: however, $c$ was sent to all processes that are members of $V(t')$ (\Cref{line:recv}) that will propose a sequence that contains a view $v'$, such that $c \in v'$, to replace $V(t')$, reducing this situation to the first case.
    \end{compactenum}
\end{compactenum}
Clearly, the same argument can be made when a correct process $p$ wants to leave the system.
Therefore, in any case there will be an installable view $v$, such that $v \neq V(t)$, where $\langle +, p \rangle/\langle -, p \rangle \in v$.
\end{proof}

\begin{lemma} \label{lemma:join_liveness}
A correct process that has invoked a \dbrbJoin{} operation eventually joins the system.
\end{lemma}
\begin{proof}
Suppose that a correct process $p$ executes a \dbrbJoin{} operation at time $t$, by sending a $m_{join} = \text{\operation{reconfig}{\langle +, p \rangle, cv = V(t)}}$ message to processes that are members of $V(t)$ and receives $V(t).q$ \msg{rec-confirm} messages.
\Cref{lemma:v(t')} claims that there will be an installable view $v$, such that $v \neq V(t)$, where $\langle +, p \rangle \in v$.

If $p$ does not receive a quorum of \msg{rec-confirm} messages, \Cref{lemma:collect_rec_confirm} shows that there exists an installable view $v \ni p$.

By \Cref{lemma:eventually_joins}, the join of process $p$ is eventually completed.
Thus, the lemma holds.
\end{proof}

\begin{lemma} \label{lemma:leave_liveness}
A correct participant that has invoked a \dbrbLeave{} operation eventually leaves the system.
\end{lemma}
\begin{proof}
Suppose that a correct participant $p$ executes a \dbrbLeave{} operation at time $t$, by sending a $m_{leave} = \text{\operation{reconfig}{\langle -, p \rangle, cv = V(t)}}$ message to processes that are members of $V(t)$ and receives $V(t).q$ \msg{rec-confirm} messages.
\Cref{lemma:v(t')} claims that there will be an installable view $v$, such that $v \neq V(t)$, where $\langle -, p \rangle \in v$.

If $p$ does not receive a quorum of \msg{rec-confirm} messages, \Cref{lemma:collect_rec_confirm_leave} shows that there exists an installable view $v \not\ni p$.

By \cref{lemma:installable_reach_2,lemma:installable_leaves}, process $p$ eventually executes \Cref{line:finish_wait_updated}.
Suppose that $\mathit{stored} = \mathit{false}$ or $\mathit{can\_leave} = \mathit{true}$, then $p$ leaves the system.

However, suppose that $\mathit{stored} = \mathit{true}$ and $\mathit{can\_leave} = \mathit{false}$, then $p$ needs to set its $\mathit{can\_leave}$ variable to $\mathit{true}$.
A correct process $p$ sets $\mathit{can\_leave} = \mathit{true}$ at \Cref{line:set_can_leave} and $p$ does that once it receives $v.q$ of \msg{deliver} messages associated with some valid (and installed) view $v$.
\Cref{lemma:reconfiguration_done} ensures that the reconfiguration of \name eventually stops which implies that process $p$ eventually collects the aforementioned \msg{deliver} messages and then leaves.
\end{proof}

%% file: sections/appendix/broadcast-properties.tex
\subsection{Broadcast} \label{sec:broadcast_proof}

\vspace*{-2mm}
Recall that some valid view $v$ is said to be installed if a correct process $p \in v$ 
processed \msg{prepare}, \msg{commit} and \msg{reconfig} messages associated with $v$. 
Moreover, a view $v$ can only be installed if it is installable (in some state of the system).


\begin{lemma} \label{lemma:collected_in_installed}
Suppose that some (correct or faulty) process $s$ collects a message certificate for a message $m$ in some valid view $v$.
Then, $v$ is installed in the system.
\end{lemma}
\begin{proof}
Since $s$ receives $v.q$ \msg{ack} messages associated with $v$ for $m$, a correct process $p \in v$ sent an \msg{ack} message associated with $v$ for $m$ to the sender $s$.
Because of the fact that $p$ processes \msg{prepare} messages associated with view $v$ only if $v$ is installed by $p$ (\cref{line:installed_cv,line:install_complete}), the lemma holds.
\end{proof}

\begin{lemma} \label{lemma:commited_message}
If message $m$ is delivered by a correct process $p$, then the sender $s$ has collected a message certificate for $m$ in some installed view $v$.
\end{lemma}
\begin{proof}
Process $p$ delivers some message $m$ upon receiving $v.q$ of appropriate \msg{deliver} messages associated with some installed view $v$ (\Cref{line:ack-store_quorum_brief}).
Since there is at least one correct process that sent the \msg{deliver} message to $p$ and that process firstly checks a message certificate (\Cref{line:verify_cert_3_brief}), we conclude that $s$ has collected a message certificate for $m$ in some installed view $v$ (\Cref{lemma:collected_in_installed}).
\end{proof}

\begin{lemma} \label{lemma:only_most_updated}
Suppose that a view $v$ is the most recent view installed in the system at time $t$, i.e., $V(t) = v$.
If process $s$ broadcasts a message $m$ at time $t$, then $s$ can not collect a message certificate for $m$ in some view $v' \subset v$.
\end{lemma}
\begin{proof}
Note that in order for $s$ to collect a message certificate for $m$ in some view $v'$, $v'$ must be installed in the system and $s$ must send a $m_{prepare} = \text{\operation{prepare}{m, v'}}$ message to processes that are members of $v'$ (\Cref{line:send_prepare_brief}).
A correct process $q$ sends an appropriate \msg{ack} message for $m_{prepare}$ only if $v'$ is the current view of $q$ and $q$ processes \msg{prepare} messages associated with $v'$.

Let us say that there is a correct process $z$ that considers $v = V(t)$ as its current view at time $t$ and processes \msg{prepare}, \msg{commit} and \msg{reconfig} messages (i.e., $v = V(t)$ is installed by $z$).
Without loss of generality, consider that $v'$ is a view immediately before $v$ in the sequence of installed views (follows directly from \Cref{lemma:unique_installable_views}).
Consequently, at least $v'.q$ processes that are members of $v'$ stopped processing \msg{prepare} messages associated with $v'$ (\Cref{line:stop_processing}).*
Suppose that $s$ has collected a message certificate in $v'$.
That means that at least $v'.q$ processes have sent an \msg{ack} message when they have received a \operation{prepare}{m, v'} message.**
From * and **, we can conclude that at least one correct process has firstly stopped processing \msg{prepare} messages associated with $v'$ and then processed a \msg{prepare} message associated with $v'$.
This is conflicting with the behaviour of a correct process.
Hence, process $s$ can not collect a valid message certificate for $m$ in $v'$.
\end{proof}

\begin{theorem}[No duplication]
No correct process delivers more than one message.
\end{theorem}
\begin{proof}
The verification at \Cref{line:ack-store_quorum_brief} ensures that no correct process delivers more than one message.
\end{proof}

\begin{theorem}[Integrity]
If some correct process delivers a message $m$ with sender $s$ and $s$ is correct, then $s$ previously broadcast $m$.
\end{theorem}
\begin{proof}
Suppose that a correct process $q$ delivers a message $m$.
That means that there is a message certificate for $m$ collected in some installed view $v$ by $s$ (\Cref{lemma:commited_message}).
A message certificate for $m$ is collected since a quorum of processes in $v$ have sent an appropriate \msg{ack} message for $m$.
A correct process sends an \msg{ack} message only when it receives an appropriate \msg{prepare} message from $s$.
Since $s$ is correct, that means that it sent the aforementioned \msg{prepare} message.
Consequently, $s$ broadcast $m$.
\end{proof}

\begin{lemma} \label{lemma:no_eq_same_view}
Suppose that process $s$ has collected a message certificate for a message $m$ in some installed view $v$.
If $s$ has also collected a message certificate for a message $m'$ in $v$, then $m = m'$.
\end{lemma}
\begin{proof}
Because of the quorum intersection, there is at least one correct process that has sent an \msg{ack} message for both $m$ and $m'$.
The verification at \Cref{line:check_allowed_brief} prevents a correct process from sending an \msg{ack} message for two different messages.
Hence, \Cref{lemma:no_eq_same_view} holds.
\end{proof}

\begin{lemma}\label{lemma:same_ack}
Suppose that process $s$ has collected a message certificate for $m$ in some installed view $v$.
Moreover, suppose that $s$ has collected a message certificate for $m'$ in some installed view $v' \supset v$.
Then, $m = m'$.
\end{lemma}
\begin{proof}
Since $s$ has collected a message certificate for $m$ in $v$, at least $v.q$ processes have acknowledged $m$ by sending an appropriate \msg{ack} message associated with $v$ to $s$ and their respective states reflect this fact (\cref{line:set_state_ack,line:send_ack_brief}).
Note that $\mathit{State}.\mathit{ack} = $ \operation{prepare}{m, v''}, for some $v'' \subseteq v$, at all correct processes from the set of $v.q$ processes that have acknowledged $m$ in $v$.

In order to prove \Cref{lemma:same_ack}, it is enough to prove that every correct process $q$ that is member of a view $v' \supset v$ will send an \msg{ack} message associated with $v'$ for $m'$ only if $m' = m$.
We prove this by induction. 

\textit{Base Step:} Suppose that $v'$ is an installed view which directly succeeds view $v$ in the sequence of installed views (follows directly from \Cref{lemma:unique_installable_views}). 
Because of the fact that members of $v$ have stopped processing \msg{prepare} messages associated with $v$ (\Cref{line:stop_processing}), every correct process $q \in v'$ receives from at least one process $z \in v$ that $z$ has acknowledged $m$ in $v$ (\cref{line:state-update,line:set_can_ack_brief,line:set_state_ack}). 
Hence, every correct process $q \in v'$ knows that it is allowed to send an \msg{ack} message associated with $v'$ for $m'$ only if $m' = m$ (\Cref{line:check_allowed_brief}).
Note that there might be some auxiliary views ``between'' $v$ and $v'$ (\Cref{property:installable} and \Cref{definition:leads_to}).
However, the outcome is the same: every correct process that is a member of $v'$ ``discovers'' that $m$ was acknowledged in the past and sends an \msg{ack} message associated with $v'$ for $m'$ only if $m' = m$.

\textit{Induction Step:} There is an installed view $v'$ such that every correct process $q \in v'$ knows that it can send an \msg{ack} message associated with $v'$ for $m'$ only if $m' = m$. 
We should prove that every correct process $r \in v''$, such that $v''$ is an installed view that directly succeeds view $v'$, will know that it is allowed to send an \msg{ack} message associated with $v''$ for $m''$ only if $m'' = m$.

Process $r$ receives from at least one process in $v'$ (\cref{line:state-update,line:set_can_ack_brief}) that $r$ is allowed to send an \msg{ack} message associated with $v''$ for $m''$ only if $m'' = m$.
Thus, $r$ ``discovers'' that message $m$ has been acknowledged in some view that precedes $v''$.

Therefore, every correct process $p \in v'$, where view $v' \supset v$ is installed, acknowledges a message $m'$ in $v'$ only if $m' = m$.
Consequently, if $s$ collects a message certificate for $m'$ in $v'$, then $m' = m$.
\end{proof}

\begin{lemma} \label{lemma:no_eq_multiple_views}
Suppose that process $s$ has collected a message certificate for $m$ in some installed view $v$.
If $s$ has also collected a message certificate for $m'$ in some installed view $v'$, then $m = m'$.
\end{lemma}
\begin{proof}
\Cref{lemma:no_eq_multiple_views} follows directly from \Cref{lemma:no_eq_same_view} if $v = v'$.

\noindent If $v \neq v'$, \Cref{lemma:no_eq_multiple_views} is a consequence of \Cref{lemma:same_ack}.
\end{proof}

\begin{theorem}[Consistency]\label{theorem:consistency}
If some correct process delivers a message $m$ and another correct process delivers a message $m'$, then $m=m'$.
\end{theorem}
\begin{proof}
\Cref{lemma:commited_message} claims that if message $m$ is delivered by a correct process, then sender $s$ has collected a message certificate for $m$ in some installed view $v$. 
Besides that, $s$ has also collected a message certificate for $m'$ in some installed view $v'$ (\Cref{lemma:commited_message}).
\Cref{lemma:no_eq_multiple_views} states that if $s$ has collected a message certificate for both $m$ and $m'$, then $m = m'$. 
Consequently, \Cref{theorem:consistency} holds.
\end{proof}

\begin{lemma} \label{lemma:store_deliver}
Consider a correct process $p$ such that $p$ stored a message $m$.
Then, $p$ eventually delivers $m$.
\end{lemma}
\begin{proof}
Correct process $p$ stores a message $m$ at \cref{line:store_first_part,line:store}.

Let us first investigate a case where $p$ stores $m$ at \Cref{line:store_first_part}.
Suppose that $cv_p = v$ at that time and consider a view $v' \supseteq v$ such that $v'$ is the first next view installed by $p$ (potentially, $v' = v$).
Since $v'$ is installed by $p$, we know that $p \in v'$ and $p$ disseminates a \msg{commit} message with $m$ at \Cref{line:relay_commit_2_brief}, if $p \neq s$.
Otherwise, $p$ disseminates a \msg{commit} message with $m$ at \Cref{line:rebroadcast_commit_brief}.

However, it is also possible that $p$ does not install any new view (since it has previously invoked the \dbrbLeave{} operation).
In that case, $p$ disseminates a \msg{commit} message with $m$ at \Cref{line:relay_store_brief_2}.


There are two possible scenarios:
\begin{compactenum}
    \item Process $p$ collects a quorum of appropriate \msg{deliver} messages associated with $v'$ and delivers $m$ (\Cref{line:ack-store_quorum_brief}).
    \item Process $p$ does not collect a quorum of appropriate \msg{deliver} messages associated with $v'$.
    However, $p$ disseminates a \msg{commit} message with $m$ to a view $v'' \supset v'$.
    Process $p$ collects a quorum of \msg{deliver} messages in view $v''$ or disseminates a \msg{commit} message again to members of some new view $v'''$.
    Eventually, view $v'''$ ``becomes'' $v_{final}$ (\Cref{lemma:reconfiguration_done}) and then $p$ delivers $m$ since the system can not reconfigure further.
\end{compactenum}

Suppose now that $p$ stores $m$ at \Cref{line:store}.
Process $p$ disseminates a \msg{commit} message right away (at \Cref{line:relay_store_brief}).
Similar argument as in the case above proves that $p$ eventually collects a quorum of \msg{deliver} messages and delivers $m$.
\end{proof}

\begin{lemma} \label{lemma:f+1_store}
Consider an installed view $v$ such that $v$ is the least recent view such that at least $v.q - \lfloor \frac{|v| - 1}{3} \rfloor$ correct processes that are members of $v$ have stored message $m$.
Then, every correct process $q \in v'$, where $v'$ is an installed view and $v' \supseteq v$, stores $m$.
\end{lemma}
\begin{proof}
We prove this lemma by induction.
Note that $\mathit{State}.\mathit{stored} = $ \operation{commit}{m, cer, v_{cer}, v''}, for some $v_{cer} \subseteq v$ and $v'' \subseteq v$, at the aforementioned correct processes.

\textit{Base Step:} View $v'$ is a view which directly succeeds view $v$ in a sequence of installed views.
Since correct processes stop processing \msg{commit} messages at \Cref{line:stop_processing}, a process $q \in v$ receives from at least one process in $v$ that message $m$ was stored, checks the message certificate and stores $m$ (\cref{line:verify_cer_2_brief,line:store_first_part,line:state_stored_first_part}).
Moreover, this also holds for a correct process $r \in v'$.
The reason is that \msg{state-update} messages are sent to members of both $v$ and $v'$ (\Cref{line:send_state_update}).

\textit{Induction Step:} A view $v'$ is installed in the system such that every correct process $q \in v'$ has stored message $m$. 
We should prove that every correct process $r \in v''$, such that $v''$ is a view that directly succeeds view $v'$ in a sequence of installed views, will store $m$.
Similarly to the base step, process $r$ receives from at least one process in $v'$ that $m$ was stored and stores $m$ (\cref{line:verify_cer_2_brief,line:store_first_part,line:state_stored_first_part}).

Lastly, we need to show that the lemma holds even if $v = v_{final}$ (\Cref{lemma:reconfiguration_done}).
If this is the case, then at least one correct process sent a \msg{commit} message associated with $v_{final}$ for $m$ (\Cref{line:relay_store_brief}).
Since the system cannot reconfigure further, every correct process eventually receives the \msg{commit} message and stores $m$.
\end{proof}

\begin{theorem}[Validity]\label{theorem:validity}
If a correct participant $s$ broadcasts a message $m$ at time $t$, then every correct process, if it is a participant at time $t' \geq t$ and never leaves the system, eventually delivers $m$.
\end{theorem}
\begin{proof}
\Cref{lemma:reconfiguration_done} claims that there is a view $v_{final}$ which includes every possible change that can be proposed in the execution.
Therefore, we should prove that every correct process $q \in v_{final}$ eventually delivers $m$.
Moreover, every correct process $q \in v_{final}$ has a valid view history up to every installable view $v \subseteq v_{final}$ (\cref{lemma:vd_second,lemma:vd_final,lemma:valid_view_history_after_join}).
Lastly, every correct process $p \in v_{final}$ eventually sets $cv_p = v_{final}$ (\Cref{lemma:installable_reach_2}).

When process $s$ broadcasts $m$, $s$ includes its view of the system inside of a \msg{prepare} message (\Cref{line:send_prepare_brief}).
Suppose that the current view of $s$ is $v$.
There are two possible scenarios:
\begin{compactenum}
    \item Process $s$ collects a message certificate for $m$ in $v$ (\Cref{line:certificate_collected_brief}).
    \item Process $s$ does not collect a message certificate for $m$ in $v$.
    However, $s$ eventually installs a new view $v'$ and rebroadcasts $m$ to processes in $v'$ (\Cref{line:rebroadcast_prepare_brief}).
    Process $s$ collects a message certificate in view $v'$ or rebroadcasts the message again to processes in a new view $v''$.
    Eventually, view $v''$ ``becomes'' $v_{final}$ and then $s$ collects a message certificate.
    
    Consider a case where $s$ wants to leave the system.
    Then, $s$ will not collect a message certificate for $m$ in $v_{final}$, because $s$ is not a member of $v_{final}$.
    However, $s$ eventually reaches a view $v_{final}'$ in which all reconfiguration requests are processed, except for $\langle -, s \rangle$.
    And the system can not reconfigure from $v_{final}'$.
    Hence, $s$ collects a message certificate for $m$ in $v_{final}'$.
\end{compactenum}

The same argument can be made for the \msg{commit} and \msg{deliver} messages.
Hence, when $s$ receives a quorum of \msg{deliver} messages associated with some view $v'$, \cref{lemma:store_deliver,lemma:f+1_store} prove that every correct process that is a member of an installed view $v'' \supseteq v'$ delivers $m$.
This ensures that every correct process that is a member of $v_{final}$  delivers $m$.
Thus, the validity property is satisfied.
\end{proof}

\begin{lemma} \label{lemma:broadcast_liveness}
\dbrbBroadcastAlone{} operation invoked by a correct process eventually terminates.
\end{lemma}
\begin{proof}
The proof follows the same argument as the proof of \Cref{theorem:validity}.
\end{proof}

\begin{theorem} [Liveness]
Every operation invoked by a correct process eventually completes.
\end{theorem}
\begin{proof}
The theorem follows directly from \cref{lemma:join_liveness,lemma:leave_liveness,lemma:broadcast_liveness}.
\end{proof}

\begin{lemma} \label{lemma:totality_helper}
Consider a correct process $p$ that invokes \dbrbLeave{} operation at time $t$.
Suppose that $v$ is the most recent valid view in the system at time $t$.
Then, $p \in v$.
\end{lemma}
\begin{proof}
We conclude that $p \in v$ because of the fact that $p$ did not invoke \dbrbLeave{} before time $t$ and the fact that \msg{propose} messages carry signed \msg{reconfig} messages for changes they propose, i.e., \Cref{lemma:leave_safety}.
\end{proof}

\begin{theorem}[Totality]
If a correct process $p$ delivers a message $m$ at time $t$, then every correct process, if it is a participant at time $t' \geq t$, eventually delivers $m$.
\end{theorem}
\begin{proof}

Suppose that a correct participant $p$ delivers a message $m$ at time $t$.
Process $p$ does this since it received $v.q$ \msg{deliver} messages for $m$ associated with a view $v$.
Moreover, $v \subseteq V(t)$.
By \Cref{lemma:f+1_store}, every correct process $r \in v'$ stores $m$, where $v' \supseteq v$ is an installed view.
Moreover, process $r$ has a valid view history up to every installable view $v_p \subseteq v'$ (\cref{lemma:vd_second,lemma:vd_final,lemma:valid_view_history_after_join}).
Then, by \Cref{lemma:store_deliver}, every correct process $r \in v'$ delivers $m$.
It follows that every correct process $r \in v_{final}$ (\cref{lemma:reconfiguration_done,lemma:install_v_final}) delivers $m$.

Suppose that some correct participant $q$ invokes \dbrbLeave{} operation at time $t' \geq t$.
By \Cref{lemma:totality_helper}, $q \in v_{mr}$ and $v_{mr} \supseteq V(t)$, where $v_{mr}$ is the most recent valid view in the system at time $t'$.
We conclude that there exists an installed view $v_{installed} \supseteq V(t)$ such that $q \in v_{installed}$.
Hence, $q$ eventually delivers $m$ and the lemma holds.
\end{proof}

\begin{theorem}[Non-triviality]
No correct process sends any message before invoking {\normalfont \dbrbJoin{}} or after returning from {\normalfont \dbrbLeave{}} operation.
\end{theorem}
\begin{proof}
A correct process $p$ starts sending messages once it invokes a \dbrbJoin{} operation.
Moreover, $p$ halts and stops sending any messages when $p$ leaves the system.
\end{proof}

%% file: sections/appendix/optimality.tex


\subsection{Impossibility Proofs}
\label{sec:optimality-impossibilities}

\vspace*{-2mm}
As we have already stated, our \name is, in a precise sense, maximal.
Namely, we prove in this section that even if only one process in the system can fail, and it can fail by crashing, then it is impossible to implement a stronger primitive denoted with \emph{Strong Dynamic Byzantine Reliable Broadcast} (\msg{sdbrb}).
\msg{sdbrb} satisfies the same set of properties as \name (\cref{definition:spec,definition:non-triv-spec}), except that validity and totality properties are exchanged for strong validity and strong totality, respectively.

First, we define strong validity and strong totality.
Second, we present the model we consider and prove that neither strong validity nor strong totality can be implemented in an asynchronous system where processes can leave and one process can crash. 

\begin{definition}[Strong Validity]
If a correct participant $s$ broadcasts a message $m$ at time $t$, then every correct process that is a participant at time $t$ eventually delivers $m$.
\end{definition}

\begin{definition}[Strong Totality]
If a correct process $p$ delivers a message $m$ at time $t$, then every correct process $q$ that did not leave the system before $t$ eventually delivers $m$.
\end{definition}

\noindent\textbf{Model.}
We consider an asynchronous system of $N$ processes that communicate by exchanging messages.
There are no bounds on message transmission delays and processing times.
Moreover, at most one process can crash in any execution.
Process that crashes at any point in time is called \emph{faulty}.
A process that is not faulty is said to be \emph{correct}.
Moreover, \emph{reliable} links connect all pairs of processes. 

The described model $\mathcal{M}$ is the same as the model we assume in \name, except the failure model considered is crash-stop (not arbitrary failure model). 
From that perspective, this model is even stronger.

Let $R_1$ be a run of some deterministic algorithm $A$, in which some process $p$ has taken action $a$ at time $t$.
Let $R_2$ be another run of $A$, such that (i) $p$'s initial state is the same in $R_1$ and $R_2$, (ii) until time $t$, process $p$ observes the same environment (e.g., sequence of messages delivered) in $R_1$ and $R_2$.
It follows that $p$ also takes action $a$ at time $t$ in run $R_2$.

\textbf{\begin{theorem}[Strong Validity Impossibility] No algorithm can implement Strong Validity in $\mathcal{M}$.
\end{theorem}}
\begin{proof}
Suppose that process $s$ crashes at time $t_s = 0$ in run $R_1$.
Moreover, some process $p$ invokes a \sdbrbLeave{} operation of \msg{sdbrb} at time $t' > 0$.
Because of the leave liveness of \msg{sdbrb}, $p$ leaves the system at time $t \geq t'$.

Consider run $R_2$ in which process $s$ broadcasts a message $m$ at time $t_s = 0$, but all messages from $s$ are delayed until time $t^+$.
Moreover, process $p$ invokes a \sdbrbLeave{} operation of \msg{sdbrb} at time $t'$ and until time $t$ observes the exact same environment as in $R_1$.
From $p$'s perspective, run $R_2$ is indistinguishable from $R_1$.
Hence, $p$ leaves the system at time $t$ and does not deliver $m$ which violates the strong validity of \msg{sdbrb}.
\end{proof}

\textbf{\begin{theorem} [Strong Totality Impossibility] No algorithm can implement Strong Totality in $\mathcal{M}$.
\end{theorem}}
\begin{proof}
Recall that a correct process $p$ triggers an internal event \textit{leaveComplete} that signals that $p$ left the system and this process halts.
We introduce an internal event \textit{almostLeaveComplete}, which any correct process $p$ triggers before actually triggering \textit{leaveComplete}.
Intuitively, \textit{almostLeaveComplete} represents a symbolic action, the penultimate step of any leaving process $p$, which does \emph{not} execute atomically with \textit{leaveComplete}.

Consider a run $R_1$.
Suppose that a correct process $p$ invokes a \sdbrbLeave{} operation of \msg{sdbrb} at time $t' > 0$.
Because of the leave liveness of \msg{sdbrb}, $p$ leaves the system at time $t \geq t'$ without having delivered any message.
Specifically, $p$ triggers \textit{almostLeaveComplete} at time $t_a$ and then triggers \textit{leaveComplete} at time $t$.
Suppose that some correct process $q$ delivers a message $m$ at time $t_m > t$ in $R_1$.

Consider a run $R_2$.
Suppose that a correct process $p$ invokes a \sdbrbLeave{} operation of \msg{sdbrb} at time $t' > 0$.
Suppose that run $R_2$ is the same as run $R_1$, including the internal event at time $t_a$, except that process $p$ triggers \textit{leaveComplete} at time $t_m^+ > t_m$.
Run $R_2$ is indistinguishable from $R_1$ to every process.
Hence, $p$ does not deliver message $m$, whereas process $q$ delivers $m$ which violates the strong totality of \msg{sdbrb}.
\end{proof}

%% file: main.bbl
\begin{thebibliography}{10}

\bibitem{aguilera2011dynamic}
Marcos~K Aguilera, Idit Keidar, Dahlia Malkhi, and Alexander Shraer.
\newblock Dynamic atomic storage without consensus.
\newblock {\em Journal of the ACM (JACM)}, 58(2):7, 2011.

\bibitem{alchieri2016efficient}
Eduardo Alchieri, Alysson Bessani, Fab{\'\i}ola Greve, and Joni Fraga.
\newblock Efficient and modular consensus-free reconfiguration for
  fault-tolerant storage.
\newblock {\em arXiv preprint arXiv:1607.05344}, 2016.

\bibitem{AttiyaCEKW19}
Hagit Attiya, Hyun~Chul Chung, Faith Ellen, Saptaparni Kumar, and Jennifer~L.
  Welch.
\newblock Emulating a shared register in a system that never stops changing.
\newblock {\em {IEEE} Trans. Parallel Distrib. Syst.}, 30(3):544--559, 2019.

\bibitem{collect-churn-tr}
Hagit Attiya, Sweta Kumari, Archit Somani, and Jennifer~L. Welch.
\newblock Store-collect in the presence of continuous churn with application to
  snapshots and lattice agreement.
\newblock {\em CoRR}, abs/2003.07787, 2020.
\newblock URL: \url{https://arxiv.org/abs/2003.07787}.

\bibitem{baldoni2009implementing}
Roberto Baldoni, Silvia Bonomi, Anne-Marie Kermarrec, and Michel Raynal.
\newblock Implementing a register in a dynamic distributed system.
\newblock In {\em DISC}, pages 639--647. IEEE, 2009.

\bibitem{BaldoniBKR09}
Roberto Baldoni, Silvia Bonomi, Anne{-}Marie Kermarrec, and Michel Raynal.
\newblock Implementing a register in a dynamic distributed system.
\newblock In {\em ICDCS}, pages 639--647, 2009.

\bibitem{bracha1987asynchronous}
Gabriel Bracha.
\newblock Asynchronous byzantine agreement protocols.
\newblock {\em Information and Computation}, 75(2):130--143, 1987.

\bibitem{br85acb}
Gabriel Bracha and Sam Toueg.
\newblock {Asynchronous Consensus and Broadcast Protocols}.
\newblock {\em JACM}, 32(4), 1985.

\bibitem{cachin2011introduction}
Christian Cachin, Rachid Guerraoui, and Lu{\'\i}s Rodrigues.
\newblock {\em Introduction to reliable and secure distributed programming}.
\newblock Springer Science \& Business Media, 2011.

\bibitem{chockler2001group}
Gregory~V Chockler, Idit Keidar, and Roman Vitenberg.
\newblock Group communication specifications: a comprehensive study.
\newblock {\em ACM Computing Surveys (CSUR)}, 33(4):427--469, 2001.

\bibitem{broadcastApplications}
P.~Th. Eugster, R.~Guerraoui, S.~B. Handurukande, P.~Kouznetsov, and A.-M.
  Kermarrec.
\newblock Lightweight probabilistic broadcast.
\newblock {\em ACM Trans. Comput. Syst.}, 21(4):341–374, November 2003.
\newblock \href {https://doi.org/10.1145/945506.945507}
  {\path{doi:10.1145/945506.945507}}.

\bibitem{gafni2015elastic}
Eli Gafni and Dahlia Malkhi.
\newblock Elastic configuration maintenance via a parsimonious speculating
  snapshot solution.
\newblock In {\em DISC}, pages 140--153. Springer, 2015.

\bibitem{Geng2013}
Haoyan Geng and Robbert {Van Renesse}.
\newblock {Sprinkler - Reliable broadcast for geographically dispersed
  datacenters}.
\newblock {\em Lecture Notes in Computer Science (including subseries Lecture
  Notes in Artificial Intelligence and Lecture Notes in Bioinformatics)}, 8275
  LNCS:247--266, 2013.
\newblock \href {https://doi.org/10.1007/978-3-642-45065-5_13}
  {\path{doi:10.1007/978-3-642-45065-5_13}}.

\bibitem{guerraoui2019consensus}
Rachid Guerraoui, Petr Kuznetsov, Matteo Monti, Matej Pavlovi{\v{c}}, and
  Dragos-Adrian Seredinschi.
\newblock The consensus number of a cryptocurrency.
\newblock In {\em PODC}, pages 307--316, 2019.

\bibitem{guerr19:11329}
Rachid Guerraoui, Petr Kuznetsov, Matteo Monti, Matej Pavlovic, and
  Dragos-Adrian Seredinschi.
\newblock Scalable byzantine reliable broadcast.
\newblock In {\em DISC}, pages 22:1--22:16, 2019.

\bibitem{jehl2015smartmerge}
Leander Jehl, Roman Vitenberg, and Hein Meling.
\newblock Smartmerge: A new approach to reconfiguration for atomic storage.
\newblock In {\em International Symposium on Distributed Computing}, pages
  154--169. Springer, 2015.

\bibitem{kermarrec2007gossiping}
Anne-Marie Kermarrec and Maarten Van~Steen.
\newblock Gossiping in distributed systems.
\newblock {\em ACM SIGOPS operating systems review}, 41(5):2--7, 2007.

\bibitem{KW19}
Saptaparni Kumar and Jennifer~L. Welch.
\newblock Byzantine-tolerant register in a system with continuous churn.
\newblock {\em CoRR}, abs/1910.06716, 2019.
\newblock URL: \url{http://arxiv.org/abs/1910.06716}, \href
  {http://arxiv.org/abs/1910.06716} {\path{arXiv:1910.06716}}.

\bibitem{rla}
Petr Kuznetsov, Thibault Rieutord, and Sara Tucci-Piergiovanni.
\newblock Reconfigurable lattice agreement and applications.
\newblock In {\em OPODIS}, 2019.

\bibitem{bla}
Petr Kuznetsov and Andrei Tonkikh.
\newblock Asynchronous reconfiguration with byzantine failures.
\newblock In {\em DISC}, pages 27:1--27:17, 2020.

\bibitem{ma97secure}
Dahlia Malkhi, Michael Merritt, and Ohad Rodeh.
\newblock Secure {R}eliable {M}ulticast {P}rotocols in a {WAN}.
\newblock In {\em ICDCS}, 1997.

\bibitem{malkhi1997high}
Dahlia Malkhi and Michael Reiter.
\newblock A high-throughput secure reliable multicast protocol.
\newblock {\em Journal of Computer Security}, 5(2):113--127, 1997.

\bibitem{nakamotobitcoin}
Satoshi Nakamoto.
\newblock Bitcoin: A peer-to-peer electronic cash system.
\newblock {\em Whitepaper}, 2008.

\bibitem{pedone2002handling}
Fernando Pedone and Andr{\'e} Schiper.
\newblock Handling message semantics with generic broadcast protocols.
\newblock {\em Distributed Computing}, 15(2):97--107, 2002.

\bibitem{spiegelman2017liveness}
Alexander Spiegelman and Idit Keidar.
\newblock On liveness of dynamic storage.
\newblock In {\em SIROCCO}, pages 356--376. Springer, 2017.

\bibitem{wid07booting}
Josef Widder and Ulrich Schmid.
\newblock Booting clock synchronization in partially synchronous systems with
  hybrid process and link failures.
\newblock {\em Distributed Computing}, 20(2):115--140, 2007.

\end{thebibliography}
